\documentclass{article}
\usepackage{graphicx}
\usepackage{amsmath}
\usepackage{amsthm}
\usepackage{amssymb}
\usepackage{cite}
\usepackage{amsfonts}
\usepackage{color}
\usepackage{geometry}

\def\End{{\rm End}}

\DeclareMathOperator{\alldiff}{all-diff}

\DeclareMathOperator{\NAE}{NAE}
\DeclareMathOperator{\Aut}{Aut}
\DeclareMathOperator{\Pol}{Pol}

\DeclareMathOperator{\Phyl}{Phyl}
\DeclareMathOperator{\Csp}{CSP}

\DeclareMathOperator{\tx}{tx}

\DeclareMathOperator{\yca}{yca}

\def\End{{\rm End}}

{
\theoremstyle{definition}
\newtheorem{definition}{Definition}
}

\newtheorem{observations}{Observations}
\newtheorem{lemma}{Lemma}
{
\theoremstyle{definition}
\newtheorem{example}{Example}
}
\newtheorem{theorem}{Theorem}
\newtheorem{corollary}{Corollary}
\newtheorem{proposition}{Proposition}
\newcommand{\TIMES}{\ensuremath{\times}}

\newcommand{\cproblem}[3]{
\vspace{.2cm}
\noindent {\bf #1} \\
INSTANCE: #2 \\
QUESTION: #3 \\}

\newcommand{\mL}{{\mathbb L}}
\renewcommand{\L}{\mathbb{L}}

\newcommand{\Fresse}{Fra\"{i}ss\'{e}}
\newcommand{\Nesetril}{Ne\v{s}etr\v{r}il}

\begin{document}

\title{The Complexity of\\Phylogeny Constraint Satisfaction Problems\thanks{The first author has received funding from the European Research Council (ERC, grants no. 257039 and no. 681988) and funding from the German Science Foundation (DFG, project no. 622397).   The second author is partially supported by the Swedish Research Council (VR) under grant 621-2012-3239. The third author received funding from the European Research Council under the European Community's Seventh Framework Programme (Grant no. 257039), the project P27600 of the Austrian Science Fund (FWF), and the Vietnam National Foundation for Science and Technology Development (NAFOSTED) under grant number 101.99-2016.16.}
}
\author{Manuel Bodirsky\footnote{Institut f\"ur Algebra, TU Dresden, 01062 Dresden, Germany}~, Peter Jonsson\footnote{Department of Computer and Information Science, Link\"opings Universitet, SE-581 83 Link\"oping, Sweden}~,~and 
 Trung Van Pham\footnote{Department of Mathematics for Computer Science, Institute of Mathematics, Vietnam Academy of Science and Technology, 18 Hoang Quoc Viet Road, Cau Giay District, Hanoi, Vietnam}}



\maketitle
\begin{abstract}
We systematically study the computational complexity of a broad class of 
computational problems in phylogenetic reconstruction. The class contains
 for example the rooted triple consistency problem, forbidden subtree 
problems, the quartet consistency problem, and many other problems studied 
in the bioinformatics literature. The studied problems can be described as  
\emph{constraint satisfaction problems} where the constraints have a 
first-order definition over the rooted triple relation. We show that every 
such phylogeny problem can be solved in polynomial time or is NP-complete. 
On the algorithmic side, we generalize a well-known polynomial-time 
algorithm of Aho, Sagiv, Szymanski, and Ullman for the rooted triple 
consistency problem. Our algorithm repeatedly solves linear equation 
systems to construct a solution in polynomial time.
We then show that every phylogeny problem that cannot be solved by our 
algorithm is NP-complete. Our classification establishes a dichotomy for 
a large class of infinite structures that we believe is of independent 
interest in universal algebra, model theory, and topology. The proof of 
our main result combines results and techniques from various research areas: 
a recent classification of the model-complete cores of the reducts of the 
homogeneous binary branching C-relation, Leeb's Ramsey theorem for 
rooted trees, and universal algebra.
\end{abstract}



%
 

\section{Introduction}
Phylogenetic consistency problems are computational problems that have been studied 
for phylogenetic reconstruction in computational biology, but also in other areas 
dealing with large amounts of possibly inconsistent data about trees, 
such as database theory~\cite{ASSU}, computational genealogy, and computational linguistics.
Given a collection of \emph{partial information} 
about a tree,
we would like to know whether the information
is \emph{consistent} in the sense that there exists a single tree that it is compatible with all the given partial information. A concrete example of a computational problem in this context is the \emph{rooted triple consistency problem}. For an informal description of this problem we consider the evolution process as a rooted binary tree in which each node presents a species and the root presents the origin of life. In an instance of the problem, we are given a set $V$ of variables, and a set of triples from $V^3$, written in the form $ab|c$ where $a,b,c \in V$, and we would like to know whether there exists a rooted tree $T$ whose leaves are from $V$ such that for each of the given triples $ab|c$ the youngest common ancestor of 
$a$ and $b$ in this tree is a descendant of the youngest common ancestor of $a$ and $c$. Aho, Sagiv, Szymanski, and Ullmann presented a polynomial-time algorithm
for this problem~\cite{ASSU}.

Many computational problems 
that are defined similarly to the
rooted triple consistency problem have been studied
in the literature. Examples include 
the \emph{subtree avoidance problem} (Ng, Steel, and Wormald~\cite{NgSteelWormald})
and the \emph{forbidden triple problem} (Bryant~\cite{Bryant}) which are NP-hard problems. Bodirsky \& Mueller~\cite{phylo-long}
have determined the complexity of rooted phylogeny
problems for the special case where
 the constraint relations are disjunctions of atomic formulas of form $xy|z$.
This result covers, for instance, the subtree avoidance problem and the forbidden triple problem.

We present a considerable strengthening of the result of Bodirsky \& Mueller~\cite{phylo-long},
and classify the complexity of phylogeny problems for all sets of phylogeny constraints that can be defined as a Boolean combination of the mentioned rooted triple relation and the equality relation (on
leaves). The reader should be aware that many problems
of this type may appear exotic from a biological point of view~---~the name ``phylogeny'' should not be taken too literally. Our results show that each of the problems
obtained in this way is polynomial-time solvable or NP-complete.
As we will demonstrate later (see Section~\ref{sect:phylo}), 
this class of problems is expressive enough to contain also \emph{unrooted
phylogeny problems}. A famous example of
such an unrooted phylogeny problem is the NP-complete 
\emph{quartet consistency problem}~\cite{Steel}:
here we are given 
a set $V$ of variables, and a set of quartets $ab{:}cd$ with $a,b,c,d \in V$, and we would like to know whether there exists a tree $T$ with
leaves from $V$ such that for each of the given quartets $ab{:}cd$
the shortest path from $a$ to $b$ does not intersect the shortest
path from $c$ to $d$ in $T$.
Another phylogeny problem that has been studied in the literature
and that falls into the framework of this paper (but not into the one in~\cite{phylo-long}) is the \emph{tree discovery problem}~\cite{ASSU}: here, the input consists of a set of 
4-tuples of variables, and the task is to find a rooted tree
$T$ such that for each $4$-tuple $(x,y,u,v)$ in the input
the youngest common ancestor of $x$ and $y$ is a proper descendant of the youngest common ancestor of $u$ and $v$.

The proof of the complexity classification is based on a variety of methods and results. Our first step is that
we give an alternative description of phylogeny problems as constraint satisfaction problems (CSPs) over a countably infinite domain where the constraint relations are first-order definable over the (up to isomorphism unique) \emph{homogeneous binary branching $C$-relation}, a well-known structure in model theory. A central result that simplifies our work considerably is a recent analysis of the endomorphism monoids of such relations~\cite{BodJonsPham}. Informally, 
this result implies that there are precisely four types of phylogeny problems: (1) trivial (i.e., if there is a solution, there is a constant solution), (2) rooted,
(3) unrooted, and (4) degenerate cases that have been called equality CSPs~\cite{ecsps}.
We will show that all unrooted phylogeny problems are NP-hard, and the complexity of all equality CSPs is already known.

The basic method to proceed from there is the {\em algebraic approach} to constraint satisfaction problems.
Here, one studies certain sets of operations (known as {\em polymorphisms})
instead of analysing the constraints themselves. 
An important tool to work with polymorphisms over
infinite domains is Ramsey theory. 
In this paper, we need a Ramsey result for rooted
trees due to Leeb~\cite{Lee-vorlesungen-ueber-pascaltheorie}, for proving
that polymorphisms behave canonically on large parts of the domain (in the sense of Bodirsky \& Pinsker~\cite{BP-reductsRamsey}),
and this allows us to perform a simplified combinatorial analysis.

Interestingly, all phylogeny problems that can be solved
in polynomial time fall into one class and can be solved by the same algorithm.
This algorithm is a considerable extension
of the algorithm by Bodirsky \& Mueller~\cite{phylo-long} for the rooted triple consistency problem. It repeatedly solves systems of linear Boolean equations to decide satisfiability of a phylogeny problem from this class.  
An illustrative example of a phylogeny problem that can be solved in polynomial time by our algorithm, but not the algorithms from~\cite{ASSU,phylo-long}, is the following computational problem: the input is a 4-uniform hypergraph with vertex set $V$; the question is whether there exists a rooted tree $T$ with leaf set $V$ such that for every hyperedge in the input $T$ has
 two disjoint subtrees that each contain precisely two of the vertices of the hyperedge. 
 
All phylogeny problems that cannot be solved by our algorithm are NP-complete. Our results are stronger than this complexity dichotomy, though, and we prove that every phylogeny problem satisfies a universal-algebraic dichotomy statement that holds for a large class of infinite structures (Theorem~\ref{thm:main}), 
which is of independent interest in the study of homogeneous structures and their polymorphism clones. In this respect, the situation is similar to previous classifications for CSPs where the constraints are first-order definable over the order of the rationals $({\mathbb{Q}};<)$ from~\cite{tcsps-journal}, or the random graph~\cite{BodPin-Schaefer}. In comparison to these previous
works, the dichotomy we present here 
is easier to state (there is just one tractable class),
but harder to prove with existing methods: in particular, unlike the situation
for constraints that are first-order definable over the 
random graph~\cite{BodPin-Schaefer},
the polymorphisms that characterise the tractable cases cannot be chosen to
be canonical (in the sense of Bodirsky \& Pinsker~\cite{BP-reductsRamsey}) 
on the entire domain.  
As such, our dichotomy provides an important test case for potentially much wider classifications of CSPs of homogeneous structures. 

The paper has the following structure. We provide basic definitions concerning phylogeny problems in
Section~\ref{sect:phylo}, and also explain how these problems can be viewed as constraint satisfaction problems for reducts of the homogeneous binary branching $C$-relation. Section~\ref{sect:algebraic} provides a brief but self-contained introduction to the universal-algebraic approach to the complexity of constraint satisfaction, and in Section~\ref{sect:prelims} we collect known results that we will use in our proof. 
Section~\ref{sect:violating} applies the universal algebraic approach to phylogeny problems, and we derive structural  properties of phylogeny problems that do not simulate
a known hard phylogeny problem. In Section~\ref{sect:Horn} we translate these structural properties 
into definability properties in terms of syntactically restricted formulas, called affine Horn formulas. 
This section also contains our algorithm for solving the tractable cases.  In Section~\ref{sec:treeoperations}, we present a characterisation of our tractable class of phylogeny problems based on polymorphisms.   
Finally, in Section~\ref{sect:main} we put everything together and state and prove our main results, including the mentioned complexity dichotomy.

This article is a revised and extended version of an earlier conference publication~ \cite{BodJonsPhamCSP}. 

\section{Phylogeny Problems}
\label{sect:phylo}
In this section, we define (in Sections~\ref{sect:rootedtrees} and \ref{sect:phylogenyproblems}) the class of phylogeny problems studied in this article and illustrate it by providing examples from the literature. 
We continue in Section~\ref{sect:phylo-csp}
by showing how to formulate such phylogeny problems as \emph{constraint satisfaction problems} 
over an infinite domain.  

\subsection{Rooted trees}
\label{sect:rootedtrees}
We fix some standard terminology concerning rooted trees.
Let $T$ be a tree (i.e., an undirected, acyclic, and connected graph) with a distinguished vertex $r$, 
the \emph{root} of $T$. The vertices of $T$ are denoted by $V(T)$.
All trees in this paper will be \emph{binary},
i.e., all vertices except for the root have either degree $3$ or $1$,
and the root has either degree $2$ or $0$.
The \emph{leaves} $L(T)$ of $T$ are the vertices of $T$ of degree one. 

For $u,v \in V(T)$, we say that $u$ \emph{lies below} $v$
if the path from $u$ to $r$ passes through $v$. 
We say that $u$ \emph{lies strictly below} $v$ if
$u$ lies below $v$ and $u \neq v$.
The \emph{youngest common ancestor (yca)} of a set of vertices $S \subseteq V(T)$ is the node $u$ that lies above all
vertices in $S$ and has maximal distance 
from $r$; this node is uniquely determined by $S$. 

\begin{definition}\label{def:leaf-struct}
The \emph{leaf structure} of a binary rooted tree $T$ is the relational structure $(L(T);C)$ 
where $C(a,b,c)$ holds in $C$ if and only if $\yca(\{b,c\})$ lies strictly below $\yca(\{a,b,c\})$ in $T$.
We also call $T$ the \emph{underlying tree} of the leaf structure. 
\end{definition}

It is well-known that a rooted tree is uniquely determined by its leaf structure (Theorem 3 in~\cite{Steel}).

\begin{definition}
For finite $S_1,S_2 \subseteq L(T)$, we write $S_1|S_2$ if 
neither of $\yca(S_1)$ and $\yca(S_2)$ lies below
the other. For arbitrary sequences of (not necessarily distinct) vertices $x_1,\dots,x_n$ and $y_1,\dots,y_m$ with $n,m \geq 1$ we write $x_1,\dots,x_n|y_1,\dots,y_m$ if 
$\{x_1,x_2,\dots,x_n\} \big | \{y_1,y_2,\dots,y_m\}$.
\end{definition}
In particular, $x|yz$ (this notation is widespread in the literature on phylogeny problems~\cite{Steel,Steel16,JanssonSung16,warnow17})
is equivalent to $C(x,y,z)$.
Note that if $x|yz$ then this includes the possibility that $y=z$; however, $x|yz$ implies that
$x \neq y$ and $x \neq z$. 
Hence, for every triple $x,y,z$ of leaves in a rooted binary tree, we either have $x|yz$, $y|xz$, $z|xy$, or $x=y=z$. 
Also note that $x_1,\dots,x_n|y_1,\dots,y_m$ if and only if $x_ix_j|y_k$ and $x_i|y_ky_l$ for
all $i,j \leq n$ and $k,l \leq m$. 

\subsection{Phylogeny problems} \label{sect:phylogenyproblems}
An \emph{atomic phylogeny formula} is a formula 
of the form $x|yz$ or of the form $x=y$. 
A \emph{phylogeny formula} is a quantifier-free formula $\phi$ that is
built from atomic phylogeny formulas with the usual Boolean connectives
(disjunction, conjunction and negation).

We say that a phylogeny formula $\phi$ with variables $V$ is 
\emph{satisfiable} if there exists a rooted binary
tree $T$ and a mapping $s \colon V \rightarrow L(T)$ 
such that $\phi$ is satisfied by $T$ under $s$ (with the usual 
semantics of first-order logic). In this case we also say that 
$(T, s)$ is a \emph{solution} to $\phi$.

Let $\Phi = \{\phi_1,\phi_2,\dots\}$ be a finite set of phylogeny formulas. 
Then the \emph{phylogeny problem} for $\Phi$ is the following 
computational problem.

\cproblem{Phylo$(\Phi)$}
{A finite set $V$ of variables,
and a finite set $\Psi$ of phylogeny formulas obtained from phylogeny formulas $\phi \in \Phi$ by substituting the variables from $\phi$ by variables from $V$.}
{Is there a tree $T$ and a mapping $s \colon V \rightarrow L(T)$
such that $(T,s)$ satisfies all formulas from $\Psi$?}

If $a_1,a_2,\dots,a_n$ and $b_1,b_2,\dots,b_m$ are sequences of leaves in a 
binary tree $T$, then by our observation above $\{a_1,a_2,\dots,a_n\}|\{b_1,b_2,\dots,b_m\}$ 
holds in $T$ if and only if $a_ia_j|b_k$ and $b_pb_q|a_r$ hold in $T$ 
for arbitrary $i,j,r\in \{1,2,\dots,n\}$ and $p,q,k\in \{1,2,\dots,m\}$. 
Thus for variables $x_1,x_2,\dots,x_n,y_1,y_2,\dots,y_m$,  we may use 
$x_1,\dots,x_n|y_1,\dots,y_m$ as a shortcut for the formula
\begin{displaymath} 
 \bigwedge_{i,j \in \{1,\dots,n\}, k,l \in \{1,\dots,m\}} (y_k|x_ix_j \wedge 
x_i|y_ky_l) \; ,
\end{displaymath}
and we use $\alldiff(x_1,\dots,x_k)$ as a shortcut for $\bigwedge_{1 \leq i < j \leq k} x_i \neq x_j$.

\begin{example}
A fundamental problem in phylogenetic reconstruction is the
rooted triple consistency problem~\cite{ASSU,BryantSteel,HenzingerKingWarnow,Steel} that was already mentioned in the introduction.
This problem can be stated conveniently as Phylo$(\{x|yz\})$.
That is, an instance of the rooted triple consistency problem
consists of a finite set of variables $V$ and a finite set of atomic
formulas of the form $x|yz$ where $x,y,z \in V$, and the
question is whether there exists a tree $T$ and a mapping $s \colon V \to L(T)$ such that for every formula $x|yz$ in the input,
$s(x)|s(y)s(z)$ holds in $T$. \qed
\end{example}

\begin{example}
The following NP-complete problem was introduced and studied 
in a closely related form by Ng, Steel, and 
Wormald~\cite{NgSteelWormald}. 
We are given a set of rooted trees on a common leaf set $V$,
and we would like to know whether there exists a tree $T$ with leaf
set $V$ such that, intuitively, for each of the given trees $T'$ the subtree of $T$ induced by the leaves of $T'$ is not the same as $T'$.

The hardness proof for this problem given Ng, Steel, and Wormald~\cite{NgSteelWormald} 
shows that already the phylogeny problem
$\text{Phylo}\big(\{(\neg x|yz) \wedge \alldiff(x,y,z),  \neg (u|xy \wedge v|yu) \wedge \alldiff(x,y,u,v) \}\big)$,
which can be seen as a special case of the problem above, is 
NP-hard.
\qed \end{example}

\begin{example}\label{expl:Nd}
The hardness proof for the rooted subtree avoidance problem given by Ng, Steel, and Wormald~\cite{NgSteelWormald} cannot be adapted to show hardness of Phylo$(\{(\neg x|yz) \wedge \alldiff(x,y,z)\})$; a hardness proof can be found in Bryant's PhD thesis~\cite{Bryant} (Section 2.6.2).
\qed \end{example}

\begin{example}
The quartet consistency problem described in the introduction
can be cast as Phylo$(\{\phi\})$
where $\phi$ is 
the following phylogeny formula.
\begin{displaymath}
(xy|u \wedge xy|v) \vee (x|uv \wedge y|uv)\;.
\end{displaymath}

Indeed, this formula describes all rooted trees with leaves $x,y,u,v$ where
the shortest path from $x$ to $y$ does not intersect the shortest
path from $u$ to $v$ (whether or not this is true is in fact independent
from the position of the root). 
\qed \end{example}

\begin{example}
Let $\phi$ be the formula 
$x_1x_2|x_3x_4 \, \vee \, x_1x_3|x_2x_4 \, \vee \, x_1x_4|x_2x_3$.
Then Phyl$(\{\phi\})$ models the following computational problem. The input consists of a 4-uniform hypergraph with a finite set of vertices $V$; the task is to determine a binary tree $T$ with
leaf set $V$ such that for every hyperedge $\{x_1x_2x_3x_4\}$ in the input, 
exactly two out of $\{x_1,\dots,x_4\}$ lie below each child of $\yca(x_1,\dots,x_4)$ in $T$.
This example cannot be solved by the algorithm of Aho, Sagiv, Syzmanski, and Ullman~\cite{ASSU}, and neither by the generalisation of this algorithm presented in~\cite{phylo-long}. However, the problem \emph{can} be solved in polynomial time by the algorithm presented in Section~\ref{sect:alg}. 
\qed \end{example}

Our main results (stated in Section \ref{sect:main}) imply a full classification of the computational complexity of Phylo$(\Phi)$.

\begin{theorem}\label{thm:complexity-main}
Let $\Phi$ be a finite set of phylogeny formulas. 
Then Phylo$(\Phi)$ is in P or NP-complete.
\end{theorem}

\subsection{Phylogeny problems as CSPs}\label{sect:phylo-csp}
As mentioned in the introduction, every phylogeny problem can be formulated as a constraint satisfaction problem over 
an infinite domain. This reformulation will be essential
to use universal-algebraic and Ramsey-theoretic tools in our complexity classification of phylogeny problems. 


Let $\Gamma$ be a structure with relational signature $\tau = \{R_1,R_2,\dots\}$. This is,
$\Gamma$ is a tuple $(D;R^\Gamma_1,R^\Gamma_2,\dots)$ where $D$ is the (finite or infinite) \emph{domain} of $\Gamma$ and 
where $R^\Gamma_i \subseteq D^{k_i}$ is a relation of arity $k_i$ over $D$. When $\Delta$ and $\Gamma$ are
two $\tau$-structures, then a \emph{homomorphism}
from $\Delta$ to $\Gamma$ is a mapping $h$ from the 
domain of $\Delta$ to the domain of $\Gamma$
such that for all $R \in \tau$ and for all $(x_1,\dots,x_k) \in R^\Delta$ we have $(h(x_1),\dots,h(x_k)) \in R^\Gamma$.

Suppose that the signature $\tau$ of $\Gamma$ is finite. 
Then
the \emph{constraint satisfaction problem for $\Gamma$}, 
denoted by 
CSP$(\Gamma)$, is the following computational problem.

\cproblem{CSP$(\Gamma)$}
{A finite $\tau$-structure $\Delta$.}
{Is there a homomorphism from $\Delta$ to $\Gamma$?}

We say that $\Gamma$ is the \emph{template} or \emph{constraint language} of
the problem CSP$(\Gamma)$. We now formulate phylogeny problems as
constraint satisfaction problems. 
Let $\Phi=\{\phi_1,\dots,\phi_n\}$ be a finite set of phylogeny formulas. If $x_1,\dots,x_{k_i}$ are the variables of $\phi_i$,
then we introduce a new relation symbol $R_i$ of arity $k_i$,
and we write $\tau$ for the set of all these relation symbols.  

If $\Psi$ is an instance of $\Phyl(\Phi)$ with variables $V$, 
then we associate to $\Psi$ a $\tau$-structure
$\Delta_\Psi$ with domain $V$ as follows. 
For $R \in \tau$ of arity $k$, the relation $R^{\Delta}$ contains the tuple $(y_1,\dots,y_k) \in V^k$ if and only if the instance $\Psi$ 
contains a formula
$\psi$ that has been obtained from a formula $\phi \in \Phi$
by replacing the variables $x_1,\dots,x_k$ of $\phi$ 
by the variables $y_1,\dots,y_k \in V$.

\begin{proposition}\label{prop:phylo-csp}
Let $\Phi$ be a finite set of phylogeny formulas.
Then there exists a $\tau$-structure $\Gamma_\Phi$ 
with countable domain $\mL$ 
and the following property: 
an instance $\Psi$ of $\Phyl(\Phi)$ is
satisfiable if and only if 
$\Delta_\Psi$ homomorphically maps to $\Gamma_\Phi$. 
\end{proposition}

The structure $\Gamma_\Phi$ in Proposition~\ref{prop:phylo-csp} is by no means
unique, and such structures are easy to construct. 
The specific choice for $\Gamma_\Phi$ presented
below is important later in the proof of our complexity
classification for phylogeny problems; as we will see, it has many pleasant model-theoretic properties. 
To define $\Gamma_\Phi$, 
we first define a `base structure' $(\L;C)$,
and then define $\Gamma_\Phi$ in terms of $(\L;C)$. 
The structure $(\L;C)$ is a well-studied object
in model theory and the theory of infinite permutation groups,
and it will be defined via \emph{\Fresse-amalgamation}. 

We need a few preliminaries from model theory. 
Injective homomorphisms that also preserve the complement of each relation are called \emph{embeddings}. 
Let $D$ be the domain of a relational $\tau$-structure $\Gamma$,
and arbitrarily choose $S \subseteq D$. 
Then the \emph{substructure induced by $S$ in $\Gamma$} 
is the $\tau$-structure $\Delta$ with domain $S$
such that $R^{\Delta} = R^\Gamma \cap S^n$ for each $n$-ary $R \in \tau$; we also write $\Gamma[S]$ for $\Delta$. 
Let $\Gamma_1$ and $\Gamma_2$ be $\tau$-structures 
with not necessarily disjoint 
domains $D_1$ and $D_2$, respectively. 
The \emph{intersection} $\Gamma_1 \cap \Gamma_2$ of $\Gamma_1$ and $\Gamma_1$ is the structure $\Delta$ with domain $D_1 \cap D_2$ 
such that $R^{\Delta}=R^{\Gamma_1}\cap R^{\Gamma_2}$ for all $R\in\tau$. 
A $\tau$-structure $\Delta$ is an \emph{amalgam of $\Gamma_1$ and $\Gamma_2$} if for $i=1,2$ there are embeddings $f_i$ of $\Gamma_i$ to $\Delta$ such that $f_1(a)=f_2(a)$ for all $a\in D_1 \cap D_2$.  
  A class $\mathcal A$ of
$\tau$-structures has the \emph{amalgamation property} if for all
$\Gamma_1,\Gamma_2\in\mathcal A$ there is a $\Delta \in \mathcal A$ that is an
amalgam of $\Gamma_1$ and $\Gamma_2$. 
A class of finite
$\tau$-structures that  has the
  amalgamation property, is closed under isomorphism and
  taking induced substructures is called an
\emph{amalgamation class}.

Homomorphisms from $\Gamma$ to $\Gamma$ are
called \emph{endomorphisms} of $\Gamma$. 
An \emph{automorphism} of $\Gamma$ 
is a bijective endomorphism whose inverse is also an endomorphism; that is, they are bijective embeddings
of $\Gamma$ into $\Gamma$. 
The set containing all
endomorphisms of $\Gamma$ is denoted $\End(\Gamma)$ while the set of
all automorphisms is denoted by $\Aut(\Gamma)$. For two arbitrary sets $X$ and $Y$, a map from a subset of $X$ to $Y$ is called \emph{a partial map} from $X$ to $Y$. Let $f$ be an arbitrary partial map from $D$ to $D$. The map $f$ is called a \emph{partial isomorphism} of $\Gamma$ if $f$ is an isomorphism from $\Gamma[S]$ to $\Gamma[f(S)]$, where $S$ denotes the domain of $f$. A relational structure $\Gamma$ is called \emph{homogeneous} if every partial isomorphism of $\Gamma$ with a finite domain can be extended to an automorphism of $\Gamma$. In this paper a partial isomorphism always means a partial map with a finite domain. Homogeneous structures $\Gamma$ with finite relational signature are \emph{$\omega$-categorical}, i.e., all countable structures that satisfy the same first-order sentences as $\Gamma$ are isomorphic (see e.g.~\cite{Oligo} or~\cite{HodgesLong}).

\begin{theorem}[\Fresse; see
Theorem~7.1.2 in~\cite{HodgesLong}]\label{thm:fraisse}
Let $\mathcal A$ be an amalgamation class with countably many non-isomorphic members. Then there is a countably infinite homogeneous $\tau$-structure $\Gamma$ such that $\mathcal A$ is the class of structures that embeds into $\Gamma$.
The structure $\Gamma$, which is unique up to isomorphism, is called the \emph{\Fresse\ limit} of $\mathcal A$. 
\end{theorem}

When working with relational structures, 
it is often convenient to not distinguish between a relation
and its relation symbol. 
For instance, when we write $(L(T),C)$ for a leaf structure (Definition~\ref{def:leaf-struct}), the letter $C$ stands both
for the relation symbol, and for the relation itself.
This should never cause confusion.

\begin{proposition}[see Proposition~7 in~\cite{BodJonsPham}]\label{prop:amalgam}
The class of all leaf structures of finite rooted binary trees is an amalgamation class.
\end{proposition}

We write $(\L;C)$ for the \Fresse-limit of the
amalgamation class from Proposition~\ref{prop:amalgam}.
This structure 
is well-studied in the literature, and the relation $C$
is commonly referred to as the {\em binary branching homogeneous C-relation}. It has been
studied in particular in the context of infinite permutation groups~\cite{AdelekeNeumann,Oligo}. There is also 
a substantial literature on {\em C-minimal
structures}, which are analogous to o-minimal structures, but where a C-relation plays the role of the order in an o-minimal structure~\cite{C-minimal,MacphersonSteinhorn}.

\begin{definition}
Let $\Delta$ be a structure. Then a relational structure 
$\Gamma$ with the same domain as $\Delta$ is called a
\emph{reduct} of $\Delta$ if all relations of $\Gamma$
have a first-order definition in $\Delta$ (using conjunction, disjunction, negation, universal and existential quantification, in the usual way, but without parameters). 
That is, for every relation $R$ of arity $k$ of $\Delta$
there exists a first-order
formula $\phi$ with free variables $x_1,\dots,x_k$ such that
$(a_1,\dots,a_k) \in R$ if and only if $\phi(a_1,\dots,a_k)$ holds
in $\Delta$. 
\end{definition}

It is well-known that all structures with a first-order definition
in an $\omega$-categorical structures are again $\omega$-categorical
(we refer once again to~\cite{HodgesLong}, Theorem 7.3.8; the analogous statement for homogeneity is false).

\begin{proof} (Proposition~\ref{prop:phylo-csp})
Let $\Phi$ be a finite set of phylogeny formulas. 
Let $\Gamma_\Phi$ be the reduct of $(\mL;C)$
defined as follows. For every $\phi \in \Phi$
with free variables $x_1,\dots,x_k$, we have the
$k$-ary relation $R_\phi$ in $\Gamma_\phi$
which is defined by the formula $\phi$ over
$(\mL;C)$. It follows (in a straightforward way) from the definitions
that this structure has the properties required in the statement
of Proposition~\ref{prop:phylo-csp}. 
\end{proof}

Conversely, every CSP for a reduct $\Gamma = (\mL;R_1,\dots,R_n)$ 
of $({\mathbb L};C)$ corresponds to a phylogeny problem. 
To see this, we need the following well-known fact. 

\begin{theorem}[see, e.g.,~\cite{HodgesLong}]
\label{thm:qe}
An $\omega$-categorical structure is homogeneous
if and only if it has \emph{quantifier-elimination}, that is, 
every first-order formula over $\Gamma$ is  equivalent to a quantifier-free formula.
\end{theorem}

Let $\phi_i$ be a quantifier-free first-order definition of $R_i$ in $(\mL;C)$.
When $\Delta$ is an instance of $\Csp(\Gamma)$, 
consider the instance $\Psi$ of $\Phyl(\{\phi_1,\dots,\phi_n\})$ 
where the variables $V$ are the vertices of $\Delta$,
and where $\Psi$ contains for every tuple $(v_1,\dots,v_n) \in R_i^\Delta$ the formula $\phi_i(v_1,\dots,v_n)$. It is again
straightforward to verify that $\Delta$ homomorphically
maps to $\Gamma$ if and only if $\Psi$ is a satisfiable
instance of $\Phyl(\{\phi_1,\dots,\phi_n\})$.

Therefore, the class of phylogeny problems
corresponds precisely to the class of CSPs whose template is a reduct of $(\mL;C)$.

\section{The Universal-Algebraic Approach}\label{sect:algebraic}
We utilize the so-called {\em universal-algebraic approach} to obtain our results. For a more detailed introduction to this approach for $\omega$-categorical
templates, see Bodirsky [5]. We introduce some central concepts concerning
definability (Section \ref{sec:ppdef}), polymorphisms (Section \ref{subsec:pol}), and model-completeness and cores (Section \ref{sect:mc-cores}). In the final section, we discuss how Ramsey theory can be used for analyzing polymorphisms. By using the language of universal
algebra, one can elegantly state the border between tractability and NP-hardness
for phylogeny problems; we present this border in Section \ref{sect:mc-cores}.

\subsection{Primitive Positive Definability}\label{sec:ppdef}

Let $\phi$ denote a first-order formula over the signature $\tau$, and assume that the
variables $z_1,\dots,z_k$ are free in $\phi$. The formula $\phi$ is {\em primitive positive}
if it is of the form $\exists x_1,\dots,x_n (\psi_1 \wedge \dots \wedge \psi_m)$ where
$\psi_1,\dots,\psi_m$ are {\em atomic}, that is, each $\psi_i$ equals either $R(y_1,\dots,y_l)$
or $y_1=y_2$ where $R \in \tau$ is $l$-ary and $y_1,y_2,\dots,y_l \in \{x_1,\dots,x_n,z_1,\dots,z_k\}$. When $\Gamma$ is a $\tau$-structure, then $\phi$ 
defines over $\Gamma$ a $k$-ary
relation, namely the set of all $k$-tuples that satisfy $\phi$ in $\Gamma$.  
We let $\langle \Gamma \rangle$
denote the set of all finitary relations that are primitive positive definable in 
$\Gamma$.

Lemma~\ref{lem:basicrelsdefs} below illustrates the concept of primitive positive definability. 
The relations that appear in this lemma will be important in later sections. 
\begin{align*}
C_d := & \; \{(x,y,z) \in \mL^3 : x|yz \wedge y \neq z\}. \\
Q := & \; \{(x,y,u,v) \in \mL^4 : ((xy|u \wedge xy|v) \vee (x|uv \wedge y|uv))\}. \\
Q_d := & \; \{(x,y,u,v) \in \mL^4 : ((xy|u \wedge xy|v) \vee (x|uv \wedge y|uv)) \wedge x \neq y \wedge u \neq v\}. \\
N := & \; \{(x,y,z) \in \mL^3 : (xy|z \vee x|yz)\} \\
N_d := & \; \{(x,y,z) \in \mL^3 : (xy|z \vee x|yz) \wedge x \neq y \wedge y \neq z \}.
\end{align*}


\begin{lemma} \label{lem:basicrelsdefs}
$\langle (\mL;C) \rangle = \langle (\mL;C_d) \rangle$,
$\langle (\mL;Q) \rangle = \langle (\mL;Q_d) \rangle$, and
$\langle (\mL;N) \rangle = \langle (\mL;N_d) \rangle$.
\end{lemma}
\begin{proof}
Note that the formula $x \neq y$ is equivalent to the primitive positive formulas $\exists u. \, C(x,y,u)$, $\exists u,v. \, Q(u,x,v,y)$, and
$\exists u. \, N(x,u,y)$.
Thus, $C_d \in \langle (\mL;C) \rangle$, $Q_d \in \langle (\mL;Q) \rangle$, and $N_d \in \langle (\mL;N) \rangle$.
We have that 
\begin{displaymath}
C(x,y,z) \Leftrightarrow \exists u \,\big(C_d(x,y,u) \wedge C_d(x,z,u)\big)\;.
\end{displaymath}
 so $C \in \langle (\mL;C_d) \rangle$. 
To see that $Q \in \langle (\mL;Q_d) \rangle$, note that
\[Q(x,y,z,t) \Leftrightarrow \exists u,v \, \big(Q_d(u,x,v,z)\wedge Q_d(u,x,v,t)\wedge Q_d(u,y,v,z)\wedge Q_d(u,y,v,t)\big).\]
Finally, 
\begin{displaymath}
N(x,y,z) \Leftrightarrow \exists u,v \, (C_d(v,x,u) \wedge C_d(u,v,y) \wedge N_d(u,z,v))\;.
\end{displaymath}
which implies  that
$N \in \langle (\mL;N_d) \rangle$, because 
$C_d(x,y,z) \Leftrightarrow \big (N_d(x,z,y) \wedge N_d(x,y,z) \big )$. 
\end{proof}

The following result
motivates
why we are interested in positive primitive definability in connection
with the complexity of CSPs.

\begin{lemma}[~\cite{Jeavons}]
\label{lem:pp-reduce}
Let $\Gamma$ be a template and let $\Gamma'$
be the structure obtained from $\Gamma$ by adding the relation $R$. If $R$ is primitive positive
definable in $\Gamma$, then $\Csp(\Gamma)$ 
and $\Csp(\Gamma')$ are polynomial-time equivalent.
\end{lemma}

The following is an application of the above lemma.
\begin{lemma}\label{lem:N}
If $N \in \langle \Gamma \rangle$ then 
$\Csp(\Gamma)$ is NP-hard.
\end{lemma}
\begin{proof}
Lemma~\ref{lem:basicrelsdefs} shows that $\langle (\mL;N) \rangle = \langle (\mL;N_d) \rangle$. Since $xy|z\wedge x|yz$ implies $x\neq z$, and $x\neq y\wedge x\neq z\wedge y\neq z$ implies $x|yz\vee y|xz\vee z|xy$, it follows from the definition of $N_d$ that $N_d$ is equivalent to $(\neg y|xz) \wedge\alldiff(x,y,z)$. We have already mentioned in Example~\ref{expl:Nd} that Bryant~\cite{Bryant} showed that the CSP for $(\mL;(\neg y|xz) \wedge\alldiff(x,y,z))$ is NP-complete. By Lemma~\ref{lem:pp-reduce}, $\Csp(\Gamma)$ is NP-hard.
\end{proof}

Therefore, in the following sections we are particularly interested 
in those reducts $\Gamma$ of
$(\mL;C)$ where $N \notin \langle \Gamma \rangle$. 
We will prove later that when $\Gamma$ is a reduct
of $(\mL;C)$ with finite relational signature 
such that $C \in \langle \Gamma \rangle$
and $N \notin \langle \Gamma \rangle$,
 then $\Csp(\Gamma)$ is in P.

\subsection{Polymorphisms}\label{subsec:pol}
Primitive positive definability can be characterised
by preservation under so-called \emph{polymorphisms} --
this is the starting point of the universal-algebraic 
approach to constraint satisfaction (see, for instance, Bulatov, Jeavons, and Krokhin~\cite{JBK} for this approach over finite domains). 
The {\em (direct--, categorical--, or cross--) product} 
$\Gamma_1 \TIMES \Gamma_2$ of two relational $\tau$-structures $\Gamma_1$ and $\Gamma_2$ 
is a $\tau$-structure on the domain $D_{\Gamma_1} \TIMES D_{\Gamma_2}$. 
For all relations $R \in \tau$ the relation $R\big((x_1, y_1)$, \dots, $(x_k, y_k)\big)$ holds in $\Gamma_1 \TIMES \Gamma_2$ iff $R(x_1, \dots, x_k)$ holds in $\Gamma_1$ and $R(y_1,\dots, y_k)$ holds in $\Gamma_2$.
Homomorphisms from $\Gamma^k = \Gamma \TIMES \cdots \TIMES \Gamma$ to $\Gamma$
are called \emph{polymorphisms} of $\Gamma$.
When $R$ is a relation over the domain $D$, then we say that $f$ \emph{preserves $R$} (or that
$R$ is {\em closed under} $f$)
if $f$ is a polymorphism of $(D;R)$.
Note that unary polymorphisms of $\Gamma$ are endomorphisms of $\Gamma$. 
When $\phi$ is a first-order formula that defines $R$,
and $f$ preserves $R$, then we also say that $f$ \emph{preserves $\phi$}.
If an operation $f$ does not
preserve a relation $R$, we say that $f$ \emph{violates} $R$.



The set of all polymorphisms $\Pol(\Gamma)$ of a relational structure forms an algebraic object called a \emph{function clone}~\cite{Szendrei}, which is 
a set of operations defined on a set $D$ that is closed
under composition and that contains all projections. 
We write $\Pol^{(k)}$ for the $k$-ary functions in $\Pol(\Gamma)$. 
The set $\Pol(\Gamma)$ is locally closed in the following sense. 
A set of functions $\mathcal F$ with domain $D$ is \emph{locally closed}
if every function $f$ with the following property belongs to $\mathcal F$:  for every finite subset $A$ of $D$ 
there is some operation $g \in \mathcal F$ such that $f(a) = g(a)$ 
for all $a \in A^k$. 
We write $\overline F$ for  
the smallest set that is locally closed and
contains $F$. 
We say that $F$  \emph{generates} an operation $g$
if $g$ is in the smallest locally closed function clone that contains $F$. 

Polymorphism clones can be used
to characterize primitive positive definability over a finite structure; this follows from results by
Bodnar\v{c}uk, Kalu\v{z}nin, Kotov, and Romov~\cite{BoKaKoRo} and Geiger~\cite{Geiger}. 
This is false for general infinite structures.
However, the result remains true if the relational structure is $\omega$-categorical.

\begin{theorem}[Bodirsky \& \Nesetril~\cite{BodirskyNesetrilJLC}]
\label{thm:inv-pol}
Let $\Gamma$ be a countable $\omega$-categorical structure.  
Then the primitive positive definable relations in $\Gamma$
are precisely the relations  
preserved by the polymorphisms of $\Gamma$.
\end{theorem}

Let $G$ be a permutation group on a set $X$. 
The \emph{orbit} of a $k$-tuple $(t_1,\dots,t_k) \in X^k$ 
under $G$
is the set of all tuples of the form $(\pi(t_1), \ldots, \pi(t_k))$,
where $\pi$ is a permutation from $G$.
The following has been discovered independently by Engeler, Svenonius, and Ryll-Nardzewski. 
\begin{theorem}[See, e.g., Theorem~7.3.1 in Hodges~\cite{HodgesLong}]\label{thm:Ryll}
A countable relational structure $\Gamma$ is $\omega$-categorical 
if and only if the automorphism
group of $\Gamma$ is \emph{oligomorphic}, that is, if for each $k \geq  1$ there 
are finitely many orbits of $k$-tuples
under $\Aut(\Gamma)$. A relation $R$ has a first-order definition 
in an $\omega$-categorical structure $\Gamma$ if and only
if $R$ is preserved by all automorphisms of $\Gamma$.
\end{theorem} 

We also need the following observation.

\begin{lemma}[Bodirsky \& Kara~\cite{tcsps-journal}]
\label{lem:small-arity}
Let $\Gamma$ be a relational structure
and let $R$ be a $k$-ary relation 
that is a union of $m$ orbits of $k$-tuples of
$\Aut(\Gamma)$. If $\Gamma$ has a polymorphism $f$ that violates
$R$, then $\Gamma$ also has an at most $m$-ary polymorphism that violates $R$.
\end{lemma}

Given a function $f \colon X^k \rightarrow Y$, we tacitly extend it to tuples in the natural way: 
\begin{align*}
& f((x^1_1,\dots,x^m_1),(x^1_2,\dots,x^m_2),\dots,(x^1_k,\dots,x^m_k)) \\
= \; & 
(f(x^1_1,x^1_2,\dots,x^1_k),f(x^2_1,x^2_2,\dots,x^2_k),\dots,f(x^m_1,x^m_2,\dots,x^m_k)) \; .
\end{align*}
When $U \subseteq X^k$, we also write $f(U)$ for the
set $\{f(u) : u \in U\}$. 
These conventions will be very convenient when working with polymorphisms.

\subsection{Model-Complete Cores and the Border Between Tractability and Hardness}
\label{sect:mc-cores}
A structure $\Gamma$
 is a {\em core} if all of its endomorphisms are embeddings.
Note that endomorphisms preserve existential positive formulas
and embeddings preserve existential formulas. 
A first-order theory $T$ is said to be {\em model-complete} if every embedding between models of $T$ preserves all first-order formulas. A structure is called \emph{model-complete}
if its first-order theory is model-complete. 
Homogeneous $\omega$-categorical structures
provide examples of model-complete structures:
the reason is that if $\Gamma$ is $\omega$-categorical and homogeneous, then 
every first-order formula is equivalent to a quantifier-free formula
 (Theorem~\ref{thm:qe}). 
Since embeddings of $\Gamma$ into $\Gamma$ preserve 
quantifier-free formulas, the statement follows from
Lemma~\ref{lem:mc}.

\begin{lemma}\label{lem:core}
The structures $(\mL;C)$ and $(\mL;Q)$ are model-complete cores. 
\end{lemma}
\begin{proof}
Let $e$ be an endomorphism of $(\mL;C)$. 
Suppose for contradiction that $e(u)=e(v)$ for 
distinct elements $u,v$ of $\mL$. Then 
$uu|v$, but not $e(u)e(u)|e(u)$, in contradiction 
to the assumption that $e$ preserves $C$. 
Hence, $e$ is injective. 
Note that the negation of $x|yz$ is equivalent to 
$x=y=z \vee xz|y \vee xy|z$, and thus $\neg(x|yz)$ has an existential positive definition in $(\mL;C)$.
It follows that $e$ preserves 
$\neg(x|yz)$, too. This implies that $e$ is an embedding and $(\mL;C)$ is a core. Model-completeness of $(\mL;C)$ follows from homogeneity. 
The structure $(\mL;Q)$ is a model-complete core, too; the proof is very similar to the proof for $(\mL;C)$ and left to the reader. 
\end{proof}

If $\Gamma$ is $\omega$-categorical, then it is possible to characterize model-completeness in terms of \emph{self-embeddings} of $\Gamma$, this
is, embeddings of $\Gamma$ into $\Gamma$. 

\begin{lemma}[Lemma~13 in Bodirsky \& Pinsker~\cite{RandomMinOps}]\label{lem:mc}
A countable $\omega$-categorical structure $\Gamma$ is
model-complete if and only if the self-embeddings of $\Gamma$ are generated by the automorphisms of $\Gamma$.
\end{lemma}

If $\Gamma$ is a core, then every endomorphism of $\Gamma$ is an embedding. We get the following consequence.

\begin{corollary}\label{cor:mc-core}
A countable $\omega$-categorical structure is a model-complete core
if and only if the 
endomorphisms of $\Gamma$ are generated
by the automorphisms of $\Gamma$.
\end{corollary}

Note that every first-order expansion of an $\omega$-categorical
model-complete core remains a model-complete core.

We say that two structures $\Gamma$ and $\Delta$ are
\emph{homomorphically equivalent} if there exists a homomorphism from $\Gamma$ to $\Delta$, and one from $\Delta$ to $\Gamma$. Clearly, homomorphically equivalent structures
have identical CSPs. 

\begin{theorem}[Theorem 16 in Bodirsky~\cite{Cores-Journal}]
Let $\Gamma$ be an $\omega$-categorical structure. 
Then $\Gamma$ is homomorphically equivalent to an $\omega$-categorical model-complete core $\Delta$. The structure $\Delta$ is
unique up to isomorphism, and again 
$\omega$-categorical. 
\end{theorem}

Hence, we speak in the following of \emph{the} model-complete core 
of an $\omega$-categorical structure. Using the concept of polymorphisms and model-complete cores, we can now give a concise description of the border between
CSPs for reducts of $(\mL;C)$ that can be solved
in polynomial time, and those that are NP-complete.

\begin{theorem}\label{thm:main-polym}
Let $\Gamma$ be a reduct of $(\mL;C)$ with a finite signature,
and let $\Delta$ be the model-complete core of $\Gamma$.
If $\Delta$ has a binary polymorphism $f$
and endomorphisms $e_1,e_2$ such that 
$e_1(f(x_1,x_2)) = e_2(f(x_2,x_1))$
for all elements 
$x_1,x_2$ of $\Delta$, then $\Csp(\Gamma)$ is in P.
Otherwise, $\Csp(\Gamma)$ is NP-complete. 
\end{theorem}

The proof of Theorem~\ref{thm:main-polym} can be found in Section~\ref{sect:main}.

\subsection{Ramsey theory for trees}\label{sect:tramsey}
We apply Ramsey theory to find regular behavior in polymorphisms of constraint languages. This approach has succesfully been adopted earlier, see e.g.~\cite{tcsps-journal,BP-reductsRamsey,BPT-decidability-of-definability}.
The Ramsey theorem we use here is less well known
and will be described below. 
We first give a brief introduction to the way Ramsey theory enters the analysis of constraint languages.

 Let $\Gamma, \Delta$ be finite $\tau$-structures. We write ${\Delta \choose \Gamma}$ for the set of all substructures of $\Delta$ that are isomorphic to $\Gamma$. When $\Gamma, \Delta, \Theta$ are $\tau$-structures, then we write
$\Theta \to (\Delta)^{\Gamma}_r$
if for all colorings $\chi \colon {\Theta \choose \Gamma} \to \{1,\dots,r\}$ there exists 
$\Delta' \in {\Theta \choose \Delta}$  such that $\chi$ is constant on ${\Delta' \choose \Gamma}$.

\begin{definition}
\label{def:Ramseyclass}
A class of finite relational structures $\cal C$ that is closed under isomorphisms and substructures is called
\emph{Ramsey} if for all $\Gamma, \Delta \in \cal C$  and for every finite $k \geq 1$ there exists a $\Theta \in \cal C$
such that $\Theta \to (\Delta)^{\Gamma}_k$.
\end{definition}
 
A homogeneous structure $\Gamma$ is called \emph{Ramsey} if the class of all finite 
structures that embed into $\Gamma$ is Ramsey. 
We use Ramsey theory to show that polymorphisms
of $\Gamma$ must behave \emph{canonically} on large parts of the domain, in the sense defined below. 
A wider introduction to canonical functions can be found in Bodirsky \& Pinsker~\cite{BP-reductsRamsey} 
and Bodirsky~\cite{Bodirsky-HDR}. 

\begin{definition}
\label{def:canonicalop}
Let $\Gamma$ be a structure and $S$ be a subset of the domain $D$ of $\Gamma$.  A function 
$f \colon D^l \to D$ is \emph{canonical on $S$ with respect to $\Gamma$} if for
all $m \geq 1$, $\alpha_1,\dots,\alpha_l \in \Aut(\Gamma)$,  and $s_1,\dots,s_l \in S^m$, there exists $\beta \in \Aut(\Gamma)$ 
such that 
\begin{displaymath}
f(\alpha_1(s_1),\dots,\alpha_l(s_l)) = \beta(f(s_1,\dots,s_l)) \; .
\end{displaymath}
\end{definition}
When $\Gamma$ is Ramsey, then the following theorem allows us to work with canonical polymorphisms of the expansion of $\Gamma$ by constants. 

\begin{theorem}[Lemma 21 in Bodirsky, Pinsker, and Tsankov~\cite{BPT-decidability-of-definability}]
\label{thm:canpol}
Let $\Gamma$ be a homogeneous ordered Ramsey structure with domain $D$. Let $c_1, \ldots, c_m \in D$, 
and let $f \colon D^l \to D$ be any 
operation. Then $\{f\} \cup \Aut(\Gamma, c_1, \ldots, c_m)$ generates an operation that is canonical with 
respect to 
$(\Gamma, c_1, \ldots, c_m)$, 
and which is 
identical with $f$ on all tuples containing only values from $c_1,\dots,c_m$. 
\end{theorem}

We now discuss the Ramsey class that is relevant 
in our context. We have to work with an expansion $(\L;C,\prec)$ of $(\L;C)$ by a linear order $\prec$ on $\mathbb L$,
which is also defined as a \Fresse-limit as follows.  
A linear order $\prec$ on the elements of a leaf structure
$(L;C)$ is called \emph{convex}
if for all $x,y,z \in L$ with $x \prec y \prec z$ we have that either $x|yz$ or that $xy|z$ (but not $xz|y$). 
Let $\cal C'$ be the class of all convexly ordered leaf structures. The following can be shown by
using an appropriate variant 
of Proposition~\ref{prop:amalgam}, 
and we omit the straightforward proof. 

\begin{proposition}\label{prop:prec}
The class $\cal C'$ is an amalgamation class; 
its \Fresse-limit is isomorphic to an 
expansion $(\L;C,\prec)$ of $(\L;C)$ by a convex linear 
ordering $\prec$.
\end{proposition}

Clearly, $(\mL;C)$ has an automorphism such
that $\alpha(x) \prec \alpha(y)$ if and only if $y \prec x$;
we denote this automorphism by $-$.


\begin{theorem}[Leeb~\cite{Lee-vorlesungen-ueber-pascaltheorie}]\label{thm:tramsey}
The structure $(\mL;C,\prec)$ is Ramsey.
In other words,
for all convexly ordered leaf structures $P,H$ and for all $k \geq 2$, there exists a convexly ordered leaf structure $T$
such that $T \rightarrow (H)^P_k$.
\end{theorem}

A self-contained proof of Theorem~\ref{thm:tramsey}
can be found in Bodirsky~\cite{BodirskyRamsey}.

\section{Toolbox}
\label{sect:prelims}
In this section we collect some known results and certain straightforward consequences
of them. The section is divided into three parts where we recapitulate results concerning
endomorphisms of phylogeny languages in Section \ref{sect:classification}, binary injective polymorphisms in Section \ref{sect:injective}, and equality constraint languages in Section \ref{sect:ecsps}.

\subsection{A Preclassification}
\label{sect:classification}
We use a fundamental result which can be seen as
a classification of the endomorphism monoids of
model-complete cores of reducts of $(\mL;C)$. 

\begin{theorem}[Bodirsky, Jonsson, \& Pham~\cite{BodJonsPham}]\label{thm:endos}
Let $\Gamma$ be a reduct of $(\mL;C)$. Then it satisfies at least one of the following:
\begin{enumerate}
\item $\Gamma$ has a constant endomorphism;
\item the model-complete core of $\Gamma$ is isomorphic to a reduct of $(\mL;=)$; 
\item the set of endomorphisms of $\Gamma$ equals the set of endomorphisms of $(\mL;Q)$;
\item the set of endomorphisms of $\Gamma$ equals the set of endomorphisms
of $(\mL;C)$. 
\end{enumerate}
\end{theorem}

Item 2 in this theorem has been stated slightly differently in Theorem 1 
of \cite{BodJonsPham}, namely that $\Gamma$ is homomorphically equivalent to a reduct 
of $(\mL,=)$. Note that this is equivalent to the model-complete 
core of $\Gamma$ being isomorphic
to a reduct of $(\mL,=)$ unless $\Gamma$ has a constant endomorphism. 
The reason is that the core of a reduct $\Gamma$ of $(\mL;=)$ either has 
one element or is itself a reduct of $(\mL,=)$.

If $\Gamma$ has a constant endomorphism,
then $\Csp(\Gamma)$ is trivial. 
If $\Gamma$ is homomorphically equivalent to a reduct of 
$(\mL;=)$, then the complexity of $\Csp(\Gamma)$
can be determined by known results which we present in Section~\ref{sect:ecsps} below.
In items 3 and 4 of Theorem~\ref{thm:endos}, we can deduce 
a statement about primitive positive
definability of the relation $Q$ and $C$ in $\Gamma$.

\begin{lemma}\label{lem:preclass}
Let $\Gamma$ be a phylogeny constraint language
which does not have a constant endomorphism and
which is not homomorphically equivalent to an equality constraint language. Then $\Gamma$ is a model-complete core, 
and $C_d$ or $Q_d$ is primitive positive definable in $\Gamma$. 
\end{lemma}
\begin{proof}
We first show that the relation $C_d$ 
consists of a single orbit of $3$-tuples of $\Aut(\Gamma)$.
Arbitrarily choose $(x_1,x_2,x_3), (y_1,y_2,y_3)\in C_d$. Since the entries of the tuples in $C_d$ are pairwise distinct, we have 
that the map that sends $x_i$ to $y_i$ for all $i \in \{1,2,3\}$, is a partial isomorphism of $\Aut(\mL;C)$. Since $(\mL;C)$ is homogeneous, the partial map can be extended to an automorphism $\alpha$ of $(\mL;C)$. 
This implies that $C_d$ consists of one orbit of $3$-tuples of $\Aut(\mL;C) = \Aut(\Gamma)$. 

If $C_d$ has a primitive positive definition in $\Gamma$,
then so has $C$ by Lemma \ref{lem:basicrelsdefs}, 
and because $(\mL;C)$ is a model-complete core by Lemma~\ref{lem:core}, so is $\Gamma$ by the remark after Corollary \ref{cor:mc-core}. 
If $C_d$
does not have a primitive positive definition in $\Gamma$,
then there is
a polymorphism of $\Gamma$
that violates $C_d$ by Theorem~\ref{thm:inv-pol}. 
Since $C_d$ consists of one 
orbit of $3$-tuples of $\Aut(\Gamma)$, 
there is an endomorphism $e$ of $\Gamma$ that violates $C_d$ by Lemma~\ref{lem:small-arity}. 
This implies that $C$ is violated by $e$, too, since $\langle (\mL;C) \rangle = \langle (\mL;C_d) \rangle$
by Lemma~\ref{lem:basicrelsdefs} and the polymorphisms of $C$ and $C_d$ coincide.
Since $\Gamma$ does 
not have constant endomorphisms and is not homomorphically equivalent to an equality constraint language,
Theorem~\ref{thm:endos} implies that the
 relation $Q$ is preserved by all endomorphisms
of $\Gamma$. Since $(\mL;Q)$ is a model-complete core (Lemma~\ref{lem:core}), it follows that
in this case $\Gamma$ is a model-complete core, too. 
Recall that $Q_d$ is primitive positive definable in $(\mL;Q)$ by Lemma~\ref{lem:basicrelsdefs} 
so $Q_d$ is preserved by all endomorphisms of $\Gamma$.

For arbitrary tuples $(x_1,\dots,x_4),(y_1,\dots,y_4) \in Q_d$, we have that 
the map that sends $x_i$ to $y_i$ for all $i\in\{1,2,3,4\}$ is a partial isomorphism of $(\mL;Q)$. Since  $(\mL;Q)$ is 
homogeneous (see e.g.~Lemma~14 in~\cite{BodJonsPham}), 
this partial map can be extended to an automorphism of
$(\mL;Q)$. 
This implies that $Q_d$ is contained in one orbit of $4$-tuples of $\Aut(\mL;Q) = \Aut(\Gamma)$. 
If $Q_d$ is not preserved by some polymorphism of $\Gamma$, it follows from 
Lemma~\ref{lem:small-arity} that $Q_d$ is not 
preserved by an endomorphism of $\Gamma$ which leads to a contradiction. Therefore, $Q_d$ is preserved by all polymorphisms of $\Gamma$. 
We conclude that
the relation $Q_d$ is primitive positive definable in $\Gamma$ by Theorem~\ref{thm:inv-pol}. 
\end{proof}

The problem $\Csp(\mL;Q_d)$ has been shown to be NP-complete
by Steel~\cite{Steel}. Also recall that $\langle (\mL;C_d) \rangle = \langle (\mL;C) \rangle$ by Lemma~\ref{lem:basicrelsdefs}. Lemma~\ref{lem:preclass} therefore 
shows that in order to classify the computational complexity
of $\Csp(\Gamma)$, we can concentrate on the situation where the relations $C_d$ and $C$ are primitive positive definable in $\Gamma$. 



\subsection{Binary Injective Polymorphisms}
\label{sect:injective}
In this part we present a condition that implies that an $\omega$-categorical structure has a binary injective polymorphism. 
The existence of binary injective polymorphisms plays an important role in later parts of the paper.

The following shows a sufficient condition for the existence of a constant endomorphism.
A finite subset $S$ of the domain of $\Gamma$ is called a $k$-set if it has $k$ elements. The 
\emph{orbit of a $k$-set} $S$ is the set $\{\alpha(S):\alpha\in \Aut(\Gamma)\}$, where $\alpha(S)$ is the image of $S$ under $\alpha$. 

\begin{lemma}[Lemma 18 in Bodirsky \& K\'ara~\cite{tcsps-journal}]\label{lem:twosetorbit}
If $\Gamma$ has only one orbit of $2$-sets and a non-injective polymorphism, 
then $\Gamma$ has a constant endomorphism.
\end{lemma}

\begin{definition}\label{def:transitivity}
The automorphism group $\Aut(\Gamma)$ is called $k$-\emph{transitive} if for any two sequences $a_1,a_2,\dots,a_k$ and $b_1,b_2,\dots,b_k$ of $k$ distinct elements there is $\alpha\in \Aut(\Gamma)$ such that $\alpha(a_i)=b_i$ for any $1\leq i\leq k$.
\end{definition}

By the homogeneity of $(\mL;C)$, 
the structure $(\mL;C)$, and all its reducts, have a
$2$-transitive automorphism group. 
Also note that $k$-transitivity of $\Aut(\Gamma)$ implies that there only exists one orbit of $k$-sets. 

\begin{definition}
The relation $\neq$ is {\em 1-independent with respect to} $\Gamma$ if for all primitive positive $\tau$-formulas $\phi$, if both $\phi \wedge x \neq y$ and $\phi \wedge z \neq w$ are satisfiable over $\Gamma$, then $\phi \wedge x \neq y \wedge z \neq w$ is satisfiable over $\Gamma$, too.
\end{definition}
This terminology is explained in greater detail by Cohen, Jeavons, Jonsson, and Koubarakis~\cite{Disj}. Let
\begin{displaymath}
S_D := \{(a,b,c) \in D^3 \; | \; b \neq c \wedge (a=b \vee a=c)\}\;,
\end{displaymath}
 and 
\begin{displaymath} 
P_D := \{(a,b,c,d) \in D^4 \; | \; a=b \vee c=d\}.
\end{displaymath}
We will use the following known results.

\begin{lemma}[Corollary 2.3 in Bodirsky, Jonsson, \& von Oertzen~\cite{HornOrFull}] \label{definabilitylemma}
Let $D$ be an infinite set. Then every relation with a first-order definition in $(D;=)$ is in $\langle (D;S_D) \rangle$.
\end{lemma}

A function $f\colon D^k\to D$ is called \emph{essentially unary} if there exists an $i \in \{1,\dots,k\}$ and a function $g \colon D \to D$ such that $f(x_1,\dots,x_k) = g(x_i)$ for all $x_1,\dots,x_k \in D$. Otherwise, $f$ is called \emph{essential}.

\begin{lemma}[Lemma 1.3.1 in P\"oschel \& Kalu\v{z}nin~\cite{KaluzninPoeschel}] \label{p4blemma}
Let $D$ be an infinite set and let $f \colon D^k \rightarrow D$ be an operation. If $f$ preserves $P_D$, then $f$ is
essentially unary.
\end{lemma}

\begin{lemma}[Contraposition of Lemma 5.3 in~\cite{HornOrFull}] \label{oneindeplemma1}
Let $\Gamma$ be a structure over an infinite domain $D$. If the binary relations in $\langle \Gamma \rangle$
are $\{D^2,\neq,=,\emptyset\}$ and $S_D \notin \langle \Gamma \rangle$, then $\neq$ is 1-independent of $\Gamma$.
\end{lemma}

\begin{lemma}[Lemma 42 in Bodirsky \& Pinsker~\cite{RandomMinOps}]\label{oneindeplemma2}
Let $\Gamma$ be a countable $\omega$-categorical structure such that $\neq$ is in $\langle \Gamma \rangle$. Then the following
are equivalent.
\begin{enumerate}
\item
$\neq$ is 1-independent of $\Gamma$ and
\item
$\Gamma$ has a binary injective polymorphism.
\end{enumerate}
\end{lemma}


\begin{theorem}\label{thm:inj}
Let $\Gamma$ be an $\omega$-categorical structure over a countably infinite domain with a 2-transitive
automorphism group. Also suppose that $\Gamma$ has an essential polymorphism
and no constant endomorphism. Then $\Gamma$
has a binary injective polymorphism. 
\end{theorem}
\begin{proof}
If $S_D \in \langle \Gamma \rangle$, then $P_D \in \langle \Gamma \rangle$, since $P_D \in \langle (D;S_D)\rangle$ by Lemma~\ref{definabilitylemma}.  
Lemma~\ref{p4blemma} implies that $\Gamma$ is preserved by essentially unary operations only and this contradicts
the assumption 
that $\Gamma$ is preserved by at least one essential polymorphism.

If $R \in \langle \Gamma \rangle$ is binary, 
then $R \in \{D^2,\neq,=,\emptyset\}$, 
since the automorphism group of $\Gamma$ is 2-transitive.
We continue by showing that $\neq$ has a primitive positive definition in $\Gamma$. Assume otherwise; then by Theorem \ref{thm:inv-pol} there must be a
polymorphism of $\Gamma$ which violates $\neq$. Since $\neq$ consists of one orbit of pairs under $\Aut(\Gamma)$, by 
Lemma~\ref{lem:small-arity} there is an endomorphism $e$ of $\Gamma$ which violates $\neq$. This implies 
that $e$ is not injective. Since $\Gamma$ has a $2$-transitive automorphism group, $\Gamma$ has only 
one orbit of $2$-sets. Lemma \ref{lem:twosetorbit} implies
that $\Gamma$ has a constant endomorphism which contradicts our assumptions. 

Now, Lemma~\ref{oneindeplemma1} implies that $\neq$ is 1-independent of $\Gamma$ since $S_D \notin \langle \Gamma \rangle$.
We can now apply Lemma~\ref{oneindeplemma2} and conclude that $\Gamma$ has a binary injective polymorphism.
\end{proof}

\begin{corollary}\label{cor:bin-inj}
Every reduct $\Gamma$ of $(\mL;C)$ such that $N \notin \langle \Gamma \rangle$
and $C \in \langle \Gamma \rangle$ has
a binary injective polymorphism.
\end{corollary}
\begin{proof}
Since all endomorphisms of $\Gamma$ preserve $C$, there is no constant endomorphism, and all endomorphisms 
also preserve $N$. Since $N$ is violated by
some polymorphism of $\Gamma$ by Theorem~\ref{thm:inv-pol},  it follows that $\Gamma$ must have
an essential polymorphism. Reducts of $(\mL;C)$ have a 2-transitive automorphism group, and the statement follows from
Theorem~\ref{thm:inj}. 
\end{proof}
\subsection{Equality Constraint Satisfaction Problems}\label{sect:ecsps}
The CSPs for reducts of $(\mL;=)$
have been called \emph{equality constraint satisfaction problems} \cite{ecsps}, and the statement of Theorem~\ref{thm:main-polym} was already known in this special case. 

\begin{theorem}[Bodirsky \& K\'ara~\cite{ecsps}; see also Bodirsky~\cite{Bodirsky-HDR}]\label{thm:ecsps}
Let $\Gamma$ be a reduct of $(\mL;=)$. Then
$\Csp(\Gamma)$ is in P if 
$\Gamma$ is preserved by a constant operation or an injective binary operation. In both cases, $\Gamma$ has polymorphisms
$e_1,e_2$, and $f$ such that $e_1(f(x_1,x_2))=e_2(f(x_2,x_1))$ for all elements
$x_1,x_2$ of $\Gamma$. Otherwise, all polymorphisms of $\Gamma$
are essentially unary, and 
$\Csp(\Gamma)$ is NP-complete. 
\end{theorem}

In the case that a reduct $\Gamma$ of $(\mL;=)$ 
is preserved by an injective binary 
operation, the relations of $\Gamma$ can be characterised syntactically. A \emph{Horn formula} is a formula in conjunctive normal form where there is at most one positive literal per clause. 

\begin{lemma}[Bodirsky, Chen, \& Pinsker~\cite{BodChenPinsker}]
\label{lem:Horn}
A relation $R$ with a first-order definition over $(\mL;=)$
is preserved by a binary injective polymorphism if and only
if $R$ has a definition over $(\mL;=)$ which is quantifier-free Horn. In this particular case, each clause can contain at most one literal of the type $x=y$. 
\end{lemma}

\section{Violating the Forbidden Triple Relation}
\label{sect:violating}
In this section we assume that $\Gamma$ is a reduct of $(\mL;C)$ 
such that $C \in \langle \Gamma \rangle$ and $N \notin \langle \Gamma \rangle$. 
We will see in the following subsections 
that these assumptions have quite strong consequences on the relations in $\langle \Gamma \rangle$.

We begin in Section \ref{sec:dominance} by introducing the central concept of {\em domination} that will be intensively used in the rest of the section.
We continue in Sections \ref{sect:splits}--\ref{sect:free} by introducing the notions of {\em affine splits}, {\em separation}, and {\em freeness}. These
properties will be the basis for the
characterization of {\em affine Horn formulas} that we present in Section \ref{sect:Horn}.

\subsection{Dominance}
\label{sec:dominance}
In this part we introduce the notion of \emph{domination}
for functions $f \colon \mL^2 \rightarrow \mL$.
 
\begin{definition}
Let $S,T \subseteq \mL$. 
A function $f \colon \mL^2 \to \mL$ is called 
\begin{itemize}
\item \emph{dominated by the first argument on $S \times T$} if for 
all $a \in S^3$ and $b \in T^3$ we have $f(a,b)\in C$
whenever $a \in C$;
\item {\em dominated by the second argument on $S \times T$}
if for all $a \in S^3$ and $b \in T^3$, we have $f(a,b) \in C$
whenever $b \in C$. 
\end{itemize}
When $S=T=\mL$, we simply speak of domination by the first (or by the second) argument. Note that we extend the function f to tuples as described in the end of Section \ref{subsec:pol}.
\end{definition}

In this section, we will show that binary polymorphisms of $(\mL;C)$ that are canonical (according to Definition \ref{def:canonicalop}) with respect to
$(\mL;C,\prec)$ are dominated by one of their 
arguments. Define
\begin{align*}
O_1 := & \; \{(x,y,z)\in \mL^3 : x\prec y\prec z\wedge x|yz\}
& O_2 := & \; \{(x,y,z)\in \mL^3 : x\prec y \prec z \wedge xy|z\} \\
\tilde O_1 := & \; \{(x,y,z) \in \mL^3 : x|yz\} 
& \tilde O_2 := &  \;  \{(x,y,z) \in \mL^3 : xy|z\}
\end{align*}

The main result of this section is Lemma \ref{lem:dominance}. Its proof is based on Lemma \ref{lem:orderedorbit}.
\begin{lemma}\label{lem:orderedorbit}
Let $f \in \Pol^{(2)}(\mL;C)$ be canonical with respect to $(\mL;C,\prec)$, and let $o_1 \in O_1$
and $o_2 \in O_2$. 
If $f(o_1,o_2) \in \tilde O_1$,
then $f(o_i,o_j) \in \tilde O_i$ 
for all $i,j\in \{1,2\}$. 
Symmetrically, if $f(o_1,o_2)\in \tilde O_2$, then 
$f(o_i,o_j)\in \tilde O_j$ for all $i,j\in \{1,2\}$.
\end{lemma}
\begin{proof}
We only present the proof of the first statement, since the second statement can be shown symmetrically. Since $f$ preserves $C$, we have $f(o_1,o_1) \in \tilde O_1$ and $f(o_2,o_2) \in \tilde O_2$. It remains to be shown that $f(o_2,o_1) \in \tilde O_2$. Let $o_1 = (o_{1,1},o_{1,2},o_{1,3})$ and $o_2 = (o_{2,1},o_{2,2},o_{2,3})$. 
Choose $u_1,u_2 \in \mL$ such that $o_{i,1} \prec o_{i,2} \prec o_{i,3} \prec u_i$ for $i \in \{1,2\}$, $o_{2,1}o_{2,2}|o_{2,3}u_2$, and $o_{1,1}o_{1,2}o_{1,3}|u_1$. Since $(o_{2,1},o_{2,3},u_2) \in O_1$, it follows from the canonicity of $f$ that $f(o_{2,1},o_{1,1})|f(o_{2,3},o_{1,3})f(u_2,u_1)$. Similarly, since $(o_{2,2},o_{2,3},u_2) \in O_1$, we have  that $f(o_{2,2},o_{1,2})|f(o_{2,3},o_{1,3})f(u_2,u_1)$.
Observe that $(o_{2,1},o_{2,2},u_2) \in O_2$ and that $(o_{1,1},o_{1,2},u_1) \in O_2$, and hence $f(o_{2,1},o_{1,1})f(o_{2,2},o_{1,2})|f(u_2,u_1)$. We conclude that
\begin{displaymath}
f(o_{2,1},o_{1,1})f(o_{2,2},o_{1,2})|f(o_{2,3},o_{1,3})f(u_2,u_1) \, ,
\end{displaymath}
and therefore $f(o_2,o_1) \in \tilde O_2$.
\end{proof}

\begin{lemma}\label{lem:dominance}
Let $f \in \Pol^{(2)}(\mL;C)$ be canonical with respect to $(\mL;C,\prec)$. Then $f$ is dominated by the first or by the second argument.
\end{lemma}

\begin{proof}
By canonicity, either $f(o_1,o_2) \in \tilde O_1$ for all $o_1 \in O_1$ and $o_2 \in O_2$, or $f(o_1,o_2) \in \tilde O_2$ for all $o_1 \in O_1$ and $o_2 \in O_2$. 
We assume that the first case applies,
since the other case can be treated analogously. 
By Lemma~\ref{lem:orderedorbit}, we then have that 
for all $o_1 \in O_1$ and $o_2 \in O_2$ and
for all $i,j \in \{1,2\}$ we have 
$f(o_i,o_j) \in \tilde O_i$.

We claim that then $f$ is dominated by the first argument. 
Arbitrarily choose $a \in C$ and $b \in \mL^3$.

{\em Case 1.} The tuple $a = (a_1,a_2,a_3)$ has pairwise distinct entries. We assume without loss of generality that $a_2 \prec a_3$; 
otherwise we can rename $a_2$ and $a_3$ accordingly. Consider the case that $a_1 \prec a_2 \prec a_3$; the case that $a_2 \prec a_3 \prec a_1$ 
can be shown analogously.  

Let $u,v,s,t$ be the elements of $\mL$ 
such that $u \prec a_1 \prec a_2 \prec a_3 \prec s$, $ua_1|a_2 a_3 s$, $a_2a_3|s$ $v \prec b_i \prec t$ for $i \in \{1,2,3\}$, 
$b_1b_2b_3v|t$, and $v|b_1b_2b_3$.
It follows from our previously made assumptions and the canonicity of $f$ that
\begin{itemize}
  \item $f(u,v)f(a_1,b_1)|f(s,t)$ since $(u,a_1,s)\in O_2$ and $(v,b_1,t) \in O_2$, 
  \item $f(u,v)| f(a_2,b_2)f(s,t)$ since $(u,a_2,s)\in O_1$ and $(v,b_2,t) \in O_2$, and
  \item $f(u,v)| f(a_3,b_3)f(s,t)$ since $(u,a_3,s) \in O_1$ and $(v,b_3,t) \in O_2$.
\end{itemize}
This implies that $f(u,v)f(a_1,b_1)|f(a_2,b_2)f(a_3,b_3)f(s,t)$, so 
\begin{displaymath}
f(a_1,b_1)|f(a_2,b_2)f(a_3,b_3)\;.
\end{displaymath}
{\em Case 2.}
$a_2 = a_3$.  Arbitrarily choose $s,t \in \mL$ such that $a_1|a_2 s$ and $a_2 \neq s$. Since $a_1,a_2,s$ are pairwise distinct, we have $f(a_1,b_1)|f(a_2,b_2)f(s,t)$ by Case 1. Similarly, $f(a_1,b_1)|f(a_3,b_3) f(s,t)$. We conclude that $f(a_1,b_1)|f(a_2,b_2)f(a_3,b_3)$.
\end{proof}

\subsection{Affine Splits}
\label{sect:splits}
We begin by introducing the notion of {\em split vectors}. 
Let $t \in \mL^k$. 
Then $(s_1,\dots,s_k) \in \{0,1\}^k$ is a 
\emph{split vector} for $t$ if $t_pt_q|t_r$ for all 
$p,q,r \in \{1,\dots,k\}$ such that $s_p = s_q \neq s_r$. 
Note that when $t$ has a split vector $s =(s_1,\dots,s_n)$, then $(1-s_1,\dots,1-s_n)$ is also a split vector for $s$. 

\begin{definition}\label{def:split-relation}
The \emph{split relation $S(R)$ of $R \subseteq \mL^k$} is the $k$-ary Boolean relation that contains 
all split vectors for all tuples $t \in R$.   
\end{definition}
\begin{example}\label{ex:balancedrelation}
Let $R\subseteq \mL^4$ be given by: $R:=\{(x,y,z,t)\in \mL^4:xy|zt\vee xz|yt\vee xt|yz\}$. Then
\begin{align*}
S(R)=\{&(0,0,0,0),(1,1,1,1),(0,0,1,1),(0,1,0,1),\\
&(0,1,1,0),(1,1,0,0),(1,0,1,0),(1,0,0,1)\}.
\end{align*}
\end{example}

We will show (in Lemma \ref{lem:affine}) that when $\Gamma$ is such that $N \notin \langle \Gamma \rangle$ and $C \in \langle \Gamma \rangle$, then all split relations of relations in $\langle\Gamma\rangle$ 
are \emph{affine}, that is, they can be defined by a conjunction
of linear equations over $\{0,1\}$. It is known that a Boolean
relation is affine if and only if it is preserved by $(x,y,z) \mapsto x+y+z \mod 2$. It therefore suffices to show that the split relations
are preserved by the Boolean operation $\oplus$ defined as
$(x,y) \mapsto x+y \mod 2$ since $x+y+z \mod 2= (x\oplus y)\oplus z $. To do so, we show a lemma which will be useful also in later parts of the paper.

\begin{lemma}\label{lem:x}
Let $\Gamma$ be a reduct of $(\mL;C)$ such that $C \in \langle \Gamma \rangle$ and $N \notin \langle \Gamma \rangle$. 
Then there are $g \in \Pol^{(2)}(\Gamma)$ and $u,v\in \mL$ such that 
\begin{displaymath}
g(u,u) g(v,v) | g(u,v) g(v,u) \; .
\end{displaymath}
\end{lemma}

The following series of lemmas is needed in 
the proof of Lemma~\ref{lem:x}. In all these lemmas,
$\Gamma$ denotes 
a reduct of $(\mL;C)$ such that $C \in \langle \Gamma \rangle$ and $N \notin \langle \Gamma \rangle$,
and $R$ denotes a 4-ary relation with a primitive positive
definition in $\Gamma$. 

\begin{lemma}\label{lem:deducingrule}
Suppose that
$R$ contains two tuples $a,b$ with pairwise distinct entries. 
Then for all $1 \leq i<j<k \leq 4$ where $(a_i,a_j,a_k)$ and $(b_i,b_j,b_k)$ are in different orbits under
$\Aut(\mL;C)$,
the relation $R$ also contains a tuple $c$ such that 
\begin{itemize}
  \item[1.]  the tuples $(a_i,a_j,a_k)$, $(b_i,b_j,b_k)$, and $(c_i,c_j,c_k)$ are in pairwise distinct orbits, and
  \item[2.] for all $p,q,r \in \{1,\dots,4\}$, if $(a_p,a_q,a_r)$ and $(b_p,b_q,b_r)$ are in the same 
orbit $O$ under $\Aut(\mL;C)$, then $(c_p,c_q,c_r)$ is in $O$, too.
\end{itemize}  
\end{lemma}
\begin{proof}
By Lemma~\ref{lem:basicrelsdefs}, $N_d \notin \langle \Gamma \rangle$. As $N_d$ consists of two orbits, 
Lemma~\ref{lem:small-arity} and Theorem~\ref{thm:inv-pol} imply the existence of an $f \in \Pol^{(2)}$
that violates $N_d$. Since $a$ and $b$ have pairwise distinct entries, we can choose $f$ such that $(a_i,a_j,a_k)$, $(b_i,b_j,b_k)$, and $(f(a_i,b_i),f(a_j,b_j),f(a_k,b_k))$ are in pairwise distinct orbits. Let $c:=f(a,b)$. The first condition follows immediately and the second condition follows from the fact that $f$ preserves $C$.
\end{proof}

\begin{lemma}\label{lem:firstdeduction}
Suppose that $R$ contains tuples $a,b,c$ with pairwise distinct entries such that $a_1a_3|a_4a_2$, $b_1 b_4|b_3b_2$, 
$c_1c_2c_4|c_3$, and $c_1c_2|c_4$. 
Then $R$ also contains a tuple $z$ with 
$z_1 z_2|z_3 z_4$.
\end{lemma}

\begin{proof}
The right-hand side of Figure~\ref{fig:deduce} shows six tuples in $\mL^4$ with their corresponding binary tree. 
Note that the tuples $a,b,c$ in Figure~\ref{fig:deduce} satisfy the preconditions of the lemma. We will show that starting from $a,b,c$ we can obtain the desired tuple $z$ by repeated
applications of 
Lemma~\ref{lem:deducingrule}. The steps are shown in the digraph on the left-hand side of 
Figure~\ref{fig:deduce}. Each of the tuples $d,e,z$ is 
obtained by applying Lemma~\ref{lem:deducingrule} to the tuples of the two incoming edges in the digraph. Specifically,

\begin{itemize}
\item we obtain $d$ from $b$ and $c$ with $i=2$, $j=3$, $k=4$,
\item we obtain $e$ from $a$ and $d$ with $i=2$, $j=3$, $k=4$, and
\item we obtain $z$ from $c$ and $e$ with $i=2$, $j=3$, $k=4$.
\end{itemize}
Since $(b_2,b_3,b_4)$, $(c_2,c_3,c_4)$ 
and $(d_2,d_3,d_4)$ must be in pairwise distinct orbits, we have
$d_4d_3|d_2$. Since the tuples $(b_1,b_3,b_4)$ and $(c_1,c_3,c_4)$ are in the same orbit under $\Aut(\mL;C)$, the tuple $(d_1,d_3,d_4)$ must be in the same orbit as $(b_1,b_3,b_4)$ and $(c_1,c_3,c_4)$ so $d_1d_4|d_3$. See the 
binary tree drawn for $d$ in Figure~\ref{fig:deduce}. One can analogously justify the binary trees for $e$ and $z$. This concludes
the proof since $z_1z_2|z_3z_4$.
\end{proof}
\pagebreak
\begin{figure}[h]
\centering
\includegraphics[scale=0.26]{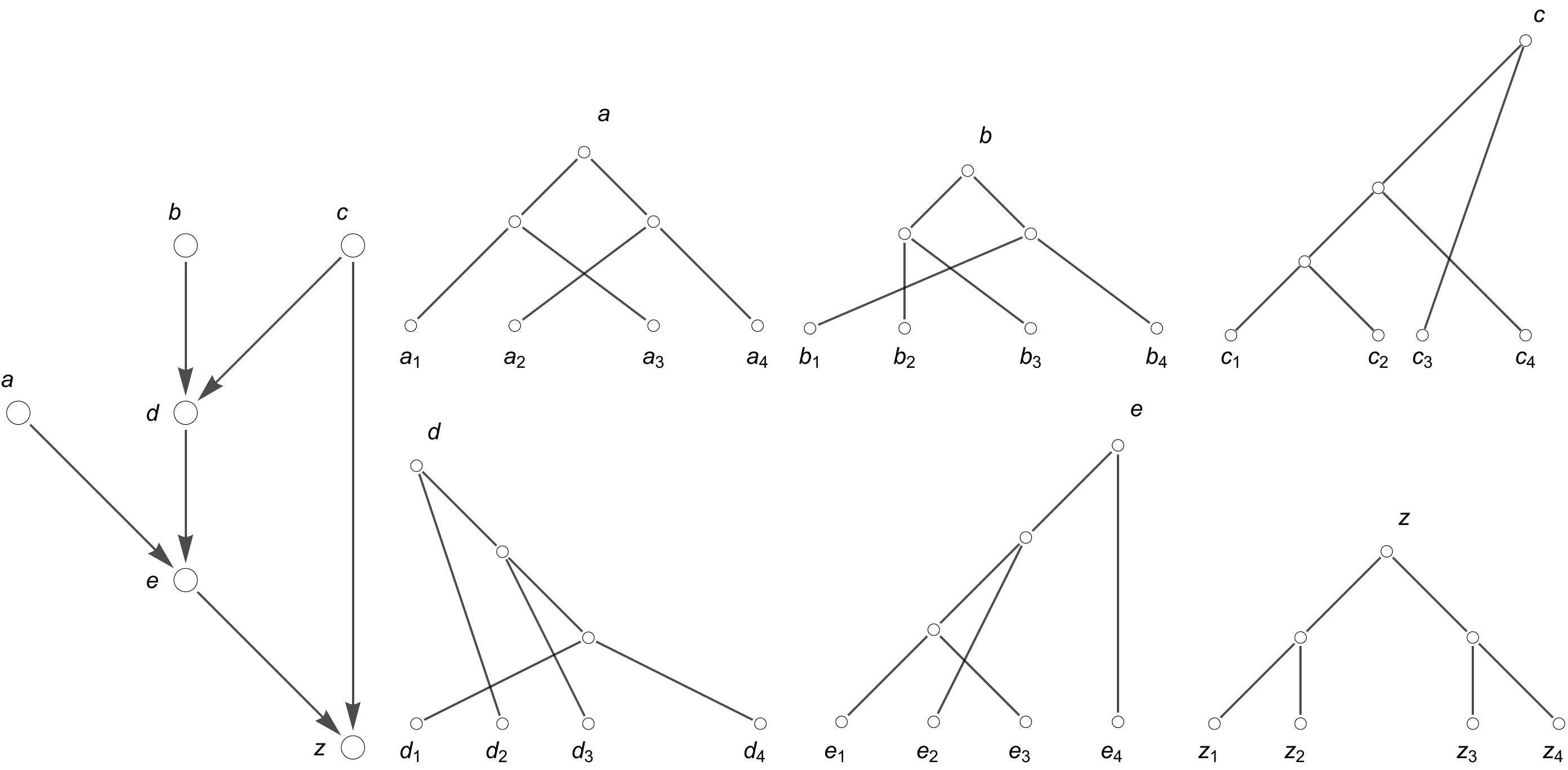}
\caption{Diagram for the proof of Lemma~\ref{lem:firstdeduction}.}
\label{fig:deduce}
\end{figure}

\begin{lemma}\label{lem:seconddeduction}
Suppose that $R$ contains tuples $a$ and $b$ with pairwise distinct entries
such that 
$a_1a_4|a_3a_2$, $b_1b_2b_4|b_3$ and $b_1|b_2b_4$.
Then there exists an $z \in R$ such that $z_1 z_2 z_4|z_3$ and $z_1 z_2|z_4$.
\end{lemma}
\begin{proof}
The proof is similar to the one for Lemma~\ref{lem:firstdeduction} in that we repeatedly apply Lemma~\ref{lem:deducingrule}.
\begin{itemize}
  \item The tuple $c$ is obtained from $a$ and $b$ for $i=2$, $j=3$, and $k=4$.
  \item The tuple $d$ is obtained from $a$ and $c$ for $i=2$, $j=3$, and $k=4$.
  \item Finally, a tuple $z$ with the desired properties is obtained from $b$ and $d$ with
  $i=1$, $j=2$, and $k=4$.
\end{itemize}  
See the diagram in Figure~\ref{fig:deduce2}.
\end{proof}


\begin{lemma}\label{lem:pre-x}
Suppose that $R$ contains two tuples $a,b$ with pairwise distinct entries
such that $a_1a_3|a_2a_4$ 
and $b_1b_4|b_2b_3$. 
Then $R$ also contains a tuple $z$ such that 
$z_1z_2|z_3z_4$. 
\end{lemma}
\begin{proof}
Let $c \in R$ be the tuple obtained by applying Lemma~\ref{lem:deducingrule} to $a$ and $b$ with $i=1,j=2,k=3$. 
We distinguish the following cases.  
\begin{enumerate}
  \item $c_1c_2|c_3c_4$. In this case we are done with $z := c$. 
 \item $c_1c_2c_4|c_3$ and $c_1c_2|c_4$. 
By applying Lemma~\ref{lem:firstdeduction} to the tuples $a$, $b$, and $c$, we obtain a tuple $z \in R$ such that $z_1z_2|z_3z_4$, and we have shown the statement.
\item $c_1 c_2 c_3|c_4$ and $c_1 c_2|c_3$. This case is similar to Case 2 since we can exchange the roles of $x_1$ and $x_2$, and $x_3$ and $x_4$, where $x\in \{a,b,c\}$.

\item $c_1c_2c_4|c_3$ and $c_2c_4|c_1$. By applying Lemma~\ref{lem:seconddeduction} to the tuples $a$ and $c$, 
we obtain a tuple $z \in R$ such that $z_1z_2z_4|z_3$ and $z_1z_2|z_4$. 
From this point on, we continue as in case 2 working with $z$ instead of $c$.
\item $c_1c_2c_4|c_3$ and $c_1c_4|c_2$. This case is similar to case 4 by exchanging the roles of $a$ and $b$. 
\end{enumerate}
These cases are exhaustive. 
\end{proof}
\pagebreak
\begin{figure}[h]
\begin{center}
 \includegraphics[scale=0.25]{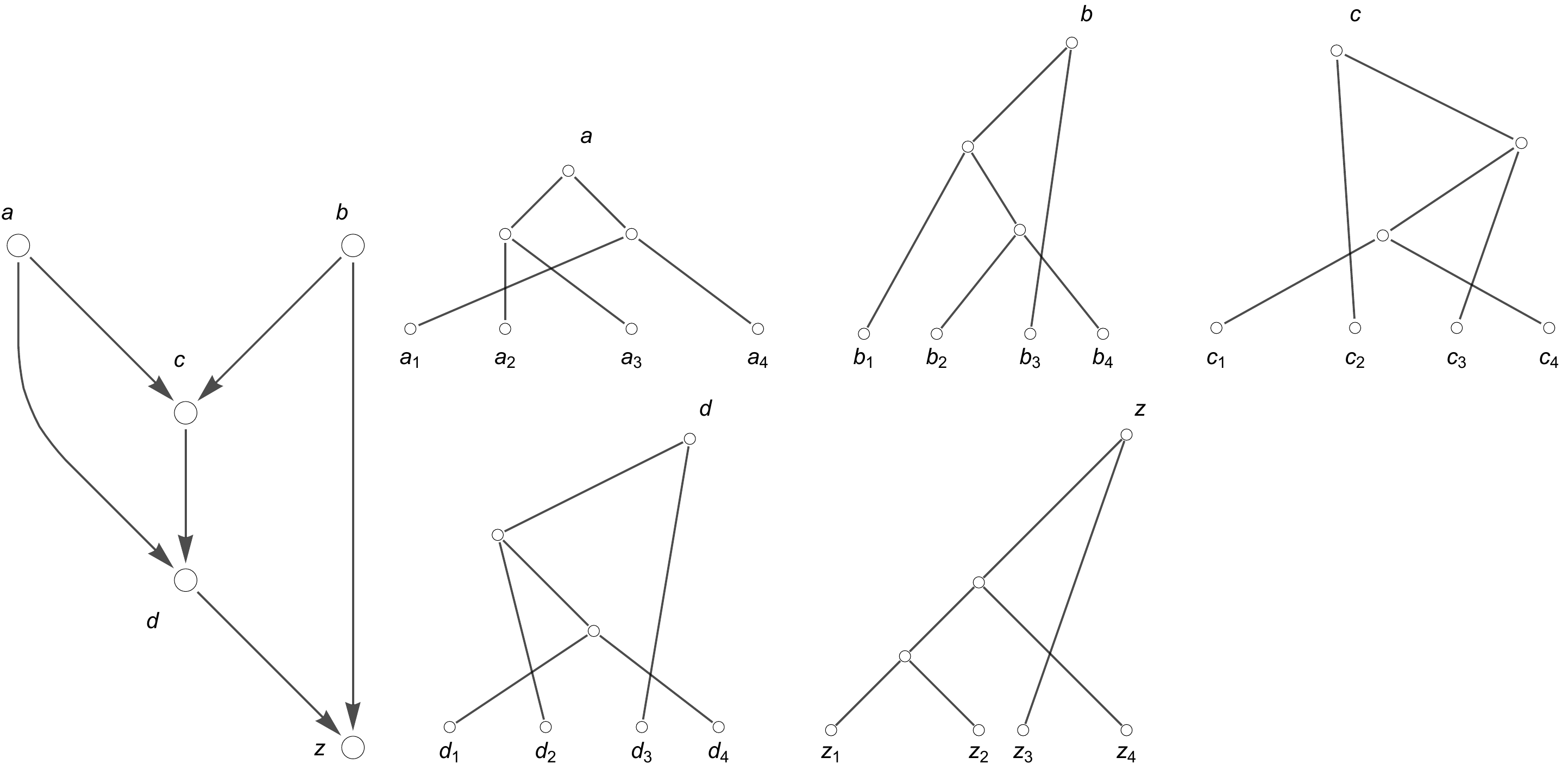}
\end{center}
\caption{Diagram for the proof of Lemma~\ref{lem:seconddeduction}.}
\label{fig:deduce2}
\end{figure}

\begin{proof} (Lemma~\ref{lem:x})
Let $a,b \in \mL^4$ be two tuples with pairwise distinct entries such that $a_1a_3|a_2a_4$ 
and $b_1b_4|b_2b_3$, and define $R := \{f(a,b) : f \in \Pol^{(2)}(\Gamma)\}$. Clearly, $R$ is preserved by all polymorphisms of
$\Gamma$, and hence $R$ is primitive positive
definable in $\Gamma$. 
By Lemma~\ref{lem:pre-x}, there exists $z \in R$ such that $z_1z_2|z_3z_4$. Let $f \in \Pol^{(2)}$ be such 
that $f(a,b)=z$. 

We claim that $f(a_1,b_1)f(a_2,b_2)|f(a_1,b_2)f(a_2,b_1)$.
Since $a_1a_3|a_4$, $b_2b_3|b_4$, and $f$ preserves $C$, 
we see that $f(a_1,b_2)f(a_3,b_3)|f(a_4,b_4)$. 
Since $a_3|a_4a_2$, $b_3|b_4b_1$, and $f$ 
preserves $C$, we see that $f(a_3,b_3)|f(a_4,b_4) f(a_2,b_1)$. Together with 
\begin{displaymath}
f(a_1,b_1)f(a_2,b_2)|f(a_3,b_3)f(a_4,b_4),
\end{displaymath}
we obtain 
\begin{displaymath}
f(a_1,b_1)f(a_2,b_2)|f(a_1,b_2)f(a_2,b_1)f(a_3,b_3)f(a_4,b_4),
\end{displaymath}
which implies the claim.


Arbitrarily choose distinct $u,v \in \mL$. Since $(\mL;C)$ is $2$-transitive, there are $\alpha,\beta\in \Aut(\mL;C)$ such that $\alpha((u,v))=(a_1,a_2)$ and 
$\beta((u,v))=(b_1,b_2)$. Then the function $g \colon \mL^2\to \mL$ 
given by $g(x,y):=f(\alpha(x),\beta(y))$
is a polymorphism of $\Gamma$,
and
$g(u,u)g(v,v)|g(u,v)g(v,u)$.
\end{proof}

With the aid of Lemma \ref{lem:x}, we can now verify that every relation in $\langle\Gamma\rangle$ has an affine split relation.
\begin{lemma}\label{lem:affine}
Let $\Gamma$ be a reduct of $(\mL;C)$ such that $C \in \langle \Gamma \rangle$ and $N \notin \langle \Gamma \rangle$. 
Then every relation in $\langle \Gamma \rangle$ has an affine split relation. 
\end{lemma}
\begin{proof}
Let $R$ be a $k$-ary relation in $\langle \Gamma \rangle$. Arbitrarily choose $s,s' \in S(R)$. 
We show that $s \oplus s' \in S(R)$.
This is clear when $s = (0,\dots,0)$, $s = (1,\dots,1)$, $s' = (0,\dots,0)$, or $s' = (1,\dots,1)$, so suppose that this
is not the case. 
Then there is a $t \in R$ with split vector $s$
and a $t' \in R$ with split vector $s'$.

Let $u,v \in \mL$ and $g \in \Pol^{(2)}(\Gamma)$ be such that $g(u,u)g(v,v)|g(u,v)g(v,u)$. Such a polymorphism $g$ exists due to Lemma \ref{lem:x}. 
By the homogeneity of $(\mL;C)$ there is $\alpha\in \Aut(\mL;C)$ such that
\begin{align*}
\alpha(\{t_i:s_i=0\})&\subset \{x\in \mL:xu|v\},~\text{and}\\
\alpha(\{t_i:s_i=1\})&\subset \{x\in \mL:xv|u\}.
\end{align*}
Informally speaking, $\alpha$ is chosen so that it maps all leaves on the left side of $t$ to the elements that are close to $u$, and maps all leaves on the right side of $t$ to the elements that are close to $v$. Similarly, there is $\beta\in \Aut(\mL;C)$ such that
\begin{align*}
\beta(\{t'_i:s'_i=0\})&\subset \{x\in \mL:xu|v\},~\text{and}\\
\beta(\{t'_i:s'_i=1\})&\subset \{x\in \mL:xv|u\}.
\end{align*}

Since $g$ preserves $C$, we see that 
$g(u,v)g(\alpha(t_i),\beta(t'_i))|g(v,u)$
for all $i \in \{1,\dots,k\}$ such that $s_i \neq s'_i$ 
and $f(u,u)f(\alpha(t_i),\beta(t'_i))|f(v,v)$  
for all $i \in \{1,\dots,k\}$ such that $s_i=s'_i$. Let $U:=\{g(u,u),g(v,v)\} \cup \{g(\alpha(t_i),\alpha(t'_i)) : s_i=s'_i\}$ and $V:=\{g(u,v),g(v,u)\} \cup \{g(\alpha(t_i),\beta(t'_i)) : s_i\neq s'_i\}$. By the choice of $g$ we know that 
$g(u,u)g(v,v)|g(u,v)g(v,u)$ so $U|V$. Consider the tuple $t'' := g(\alpha(t),\beta(t')) \in R$. It follows that $\{t''_i : s_i=s'_i\} | \{t''_i : s_i \neq s'_i\}$, and $s \oplus s'$ is a split vector of $t''$. 
\end{proof}

\subsection{Separation}
\label{sect:separated}

We introduce the notion of {\em separated} relations.

\begin{definition}\label{def:separatedness}
A relation $R \subseteq \mL^k$ is called
\emph{separated}
if for all $t,t' \in R$ such that $t_i \neq t_j$ for some $i,j \in \{1,\dots,k\}$
there exists a $t'' \in R$
such that 
\begin{itemize}
\item[S1.] for all $i,j \in \{1,\dots,k\}$ we have $t''_i \neq t''_j$
if and only if $t_i \neq t_j$ or $t'_i \neq t'_j$,
\item[S2.] the tuples $t$ and $t''$ have a common split vector $s$, and
\item[S3.] for all $i,j,l \in \{1,\dots,k\}$, if $s_i=s_j=s_l$ and $t'_it'_j|t'_l$ then $t''_it''_j|t''_l$.
\end{itemize}
\end{definition}
\begin{example}\label{ex:separatedness}
The relation $\{(x,y,z,t)\in \mL^4: (xyz|t\wedge x|yz)\vee xy|zt\}$ is not separated since it does not contain a tuple $(a,b,c,d)$ such that $abc|d\wedge ab|c$, therefore it does not satisfy S3. However, we can verify easily that the relation $\{(x,y,z,t)\in \mL^4:(xyz|t\wedge x|yz)\vee (xyz|t\wedge xy|z)\vee xy|zt\}$ is separated. 
\end{example}

Also in this section, $\Gamma$ always denotes
a reduct of $(\mL;C)$ such that $N \notin \langle \Gamma \rangle$ and $C \in \langle \Gamma \rangle$. 
We will prove that 
every relation in $\langle \Gamma \rangle$ is separated. We start with a generalisation of Lemma~\ref{lem:x}.

\begin{lemma}
\label{lem:injectivex}
There are $u,v \in \mL$ and an \emph{injective} $f \in \Pol^{(2)}(\Gamma)$ with 
\begin{displaymath}
f(u,u)f(v,v)|f(u,v)f(v,u).
\end{displaymath}
\end{lemma}
\begin{proof}
We already know from Corollary~\ref{cor:bin-inj} that $\Gamma$
has a binary injective polymorphism $p$. 
Theorem~\ref{thm:canpol} implies that
$\Aut(\mL;C) \cup \{p\}$  
generates a binary polymorphism $h$ of $\Gamma$ 
which is injective and canonical with respect to $(\mL;C,\prec)$. By Lemma \ref{lem:dominance}, the function $h$ is dominated by the first or by the second argument.

By Lemma~\ref{lem:x}, there are $u,v \in \mL$ and $g\in \Pol^{(2)}(\Gamma)$ such that
\begin{displaymath}
g(u,u)g(v,v)|g(u,v)g(v,u).
\end{displaymath}
If $h$ is dominated by the first argument, then let $f \colon \mL^2\to \mL$ be given by $f(x,y):=h(h(g(x,y),x),y)$. Since $h$ is 
injective, it follows that $f$ is injective, too. Since $h$ is dominated by the first argument and 
$g(x,x)g(y,y)|g(x,y)g(y,x)$, it follows that $f(x,x)f(y,y)|f(x,y)f(y,x)$ and we are done. The case when $h$ is 
dominated by the second argument can be treated similarly by using the function
\begin{displaymath}
f(x,y):=h(x,h(y,g(x,y))).
\end{displaymath}
\end{proof}
\begin{lemma}\label{lem:separation}
All relations in $\langle \Gamma \rangle$
are separated.
\end{lemma}
\begin{proof}
Let $R$ be a relation of arity $k$ in $\langle \Gamma \rangle$. By Lemma~\ref{lem:injectivex} there exist $u,v \in \mL$ and an injective $f \in \Pol^{(2)}(\Gamma)$ such that $f(u,u)f(v,v)|f(u,v)f(v,u)$. By Theorem \ref{thm:canpol}, the set $\{f\} \cup\Aut(\mL;C)$ 
generates a binary injective function $g$ which is canonical with respect to $(\mL;C,\prec,u,v)$ 
and identical with $f$ on $\{u,v\}$. This implies that $g(u,u)g(v,v)|g(u,v)g(v,u)$. 


Arbitrarily choose $u',v' \in \mL$ such that $u'u|v$, $u|v'v$, $u\neq u'$, and $v\neq v'$. 
Let $A_0:=\{x\in \mL:xu'|u \}$ and 
$A_1:=\{x\in \mL:xv'|v\}$. 
Since $g$ preserves $C$, we have 
\begin{align*}
\{g(u,u),g(v,v)\}\cup g((A_0^2) \cup (A_1^2)) \; | \; 
& \{g(u,v),g(v,u)\}\cup g((A_0\times A_1) \cup (A_1\times A_0)) \\
\{g(u,u)\}\cup g(A_0^2) \; | \;
& \{g(v,v)\}\cup g(A_1^2) \\ 
\{g(u,v)\}\cup g(A_0\times A_1) \; | \;
& \{g(v,u)\}\cup g(A_1\times A_0) \, .
\end{align*} 

Observe that the substructures of $(\mL;C,\prec)$ induced by $A_0$ and $A_1$ are isomorphic to $(\mL;C,\prec)$. 
For arbitrary $i\in \{0,1\}$, two tuples $x,y \in A_i^3$ 
are in the same orbit under $\Aut(\mL;C,\prec)$ if and only if 
they are in the same orbit under $\Aut(\mL;C,\prec,u,v)$. 
This implies that for any $i\in \{0,1\}$ and $j\in \{0,1\}$ 
the function $g$ is canonical on $A_i \times A_j$
with respect to $(\mL;C,\prec)$. We can therefore 
apply Lemma \ref{lem:dominance} to the restriction
of $g$ to $A_i \times A_j$ and obtain that 
$g$ is dominated by the first argument or by the second argument on $A_i \times A_j$.

Let $t,t' \in R$ be such that $t_i \neq t_j$ for some $i,j \in \{1,\dots,k\}$, and let $s$ be a split vector of $t$. 
We analyse a number of cases.

\begin{itemize}
  \item There exists an $i\in \{0,1\}$ such that $g$ is dominated by the first argument on 
$A_i \times A_0$ and on $A_i\times A_1$.
By the homogeneity of $(\mL;C)$, 
there are $\alpha,\beta \in \Aut(\mL;C)$ 
such that $\alpha(\{t'_1,t'_2,\dots,t'_k\}) \subseteq A_i$, $\beta(\{t_j:s_j=0\}) \subseteq A_0$, 
and $\beta(\{t_j:s_j=1\}) \subseteq A_1$. Let $t'':=g(\alpha(t'),\beta(t))$. Let $X:=\{g(u,u),g(v,v)\}\cup g(\{ (t_j,t'_j):s_j=i\})$ and $Y:=\{g(u,v),g(v,u)\}
\cup g(\{ (t_j,t'_j):s_j=1-i\})$. Recall that $X|Y$. Then $t''$ has split vector $s$, and condition S2 from Definition \ref{def:separatedness} holds. 
 Since $g$ is injective, 
property S1 follows directly. 
Since $g$ is dominated 
by the first argument on $A_i\times A_0$ and on $A_i\times A_1$, it is straightforward to verify that $t''$ 
satisfies S3.
\item There exists an index $i\in \{0,1\}$ such that $g$ is dominated by the second argument on $A_0\times A_i$ and on $A_1\times A_i$. This case is analogous to the previous case since we can
work with the polymorphism $g'(x,y):=g(y,x)$ instead of $g$.
\item For arbitrary $i,j\in \{0,1\}$, the operation $g$ is dominated on 
$A_i\times A_j$ by the first argument if $i=j$, 
and by the second argument if $i\neq j$. 

By the homogeneity of $(\mL;C)$, we can choose 
$\alpha,\beta\in \Aut(\mL;C)$ such that $\alpha(\{t'_1,t'_2,\dots,t'_k\}) \subseteq A_0$, 
$\beta(\{t_l:s_l=0\}) \subseteq A_0$, and $\beta(\{t_l:s_l=1\})\subseteq A_1$. 
Let $h:=g(\alpha(t'),\beta(t))$. The following facts are 
straightforward to verify.
\begin{itemize}
  \item For all $l,m \in \{1,\dots,k\}$, if $t_l \neq t_m$ or $t'_l \neq t'_m$ then $h_l \neq h_m$.
  \item The vector $s$ is a split vector of $h$.
\end{itemize}
By the assumptions concerning dominance properties of $g$, we also have the following.
\begin{itemize}
  \item for all $l,m,n\in \{1,2,\dots,k\}$ if $s_l=s_m=s_n=0$ and $t'_l|t'_mt'_n$ then $h_l|h_mh_n$, since $g$ is dominated by the first argument on $A_0\times A_0$. 
  \item for all $l,m,n\in \{1,2,\dots,k\}$ if $s_l=s_m=s_n=1$ and $t_l|t_m t_n$ then $h_l|h_mh_n$, since $g$ is dominated by the second argument on $A_0\times A_1$. 
\end{itemize}  
Choose $\alpha',\beta' \in \Aut(\mL;C)$ such that $\beta'(\{h_l:s_l=0\})\subseteq A_0$, 
$\beta'(\{h_l:s_l=1\})\subseteq A_1$ and $\alpha'(\{t'_1,t'_2,\dots,t'_k\})\subseteq A_1$. 
Let $t'':=g(\alpha'(t'),\beta'(h))$. We have the following.
\begin{itemize}
  \item $s$ is a split vector of $t''$. 
  \item For all $l,m \in \{1,\dots,k\}$, if $t_l \neq t_m$ or $t'_l\neq t'_m$ then $t''_l\neq t''_m$. This follows the fact that if $t_l \neq t_m$ or $t'_l \neq t'_m$ then $h_l \neq h_m$ which implies that $t''_l \neq t''_m$. 
\item For all $l,m,n\in \{1,\dots,k\}$ if $s_l=s_m=s_n=0$ and $t'_l|t'_mt'_n$ then $t''_l|t''_mt''_n$. 
Note that if $s_l=s_m=s_n=0$ and $t'_l|t'_mt'_n$, then $h_l|h_mh_n$. Hence, $t''_l|t''_mt''_n$ since $g$ is dominated by the second argument on $A_1\times A_0$.
\item For any $l,m,n\in \{1,\dots,k\}$ if $s_l=s_m=s_n=1$ and $t'_l|t'_mt'_n$ then $t''_l|t''_mt''_n$. This follows directly from the fact that $g$ is dominated by the first argument on $A_1\times A_1$.
\end{itemize}  
It follows from the above conditions that $t''$ satisfies S1--S3.
\item For arbitrary $i,j\in \{0,1\}$, the operation $g$ is dominated on $A_i\times A_j$ by the second argument  if $i=j$, and by the first argument if $i\neq j$.

This case can be treated similarly to the previous case by considering the polymorphism $g'(x,y) := g(y,x)$
instead of $g$.
\end{itemize}  
These cases are exhaustive, and this concludes the proof. 
\end{proof}

\subsection{Freeness}
\label{sect:free}
\def\csplit{cone split}
We introduce the notion of \emph{free} relations.
\begin{definition}\label{def:freeness}
A relation $R\subseteq \mL^k$ is called \emph{free} if for all $t,t' \in R$ that have a common split vector $s$ and that lie in the same orbit under $\Aut(\mL;=)$, there is a tuple $t'' \in R$ such that 
\begin{itemize}
\item $s$ is a split vector of $t''$,
\item for all $i,j,l \in \{1,\dots,k\}$ such that $s_i=s_j=s_l=0$, we have $t''_i|t''_j t''_l$  if and only if  $t_i|t_j t_l$, and 
\item for all $i,j,l \in \{1,\dots,k\}$ such that $s_i=s_j=s_l=1$, we have $t''_i|t''_j t''_l$ if and only if $t'_i|t'_j t'_l$.
\end{itemize}
\end{definition}


Also in this section, we assume that $C \in \langle \Gamma \rangle$ and that $N \notin \langle \Gamma \rangle$. Under these assumptions, we will prove that then every relation in $\langle \Gamma \rangle$ is free.
To do so,
we introduce the notion of \emph{\csplit s}. 

\begin{definition}
Let $t \in \mL^k$.
A sequence $I_0,I_1,\dots,I_p$ of subsets of $\{1,\ldots,k\}$ for $p \in \mathbb \{1,\dots,k\}$ is called a \emph{\csplit} of $t$ 
if $\{I_0,I_1,\dots,I_p\}$ is a partition of $\{1,\dots,k\}$ and for every $j \in \{0,\dots,p-1\}$ 
we have 
$\{t_i:i\in \bigcup_{l=0}^j I_l\} \; \big | \; \{t_i:i\in I_{j+1}\}$.
\end{definition}
The definition of cone split can be understood as assigning a level to each leaf of a tree $t$ such that any two leaves of low levels are closer to each other than a leaf of higher level, and two leaves of the same level are closer to each other than any leaf of different level.  

\begin{observations}\label{obs:csplit}\text{}
\begin{enumerate}
  \item Cone splits may not be unique.
  \item If $I_0,I_1,\dots,I_p$ is a {\csplit} of $t$, then $I_1,I_0,I_2,\dots,I_p$ is also a {\csplit} of $t$.
  \item Let $I_0,I_1,\dots,I_p$ and $J_0,J_1,\dots,J_q$ be cone splits of $t \in \mL^k$. 
If $I_0=J_0$ then $p=q$ and $I_i=J_i$ for any $0\leq i\leq p$. That is, cone splits are uniquely determined by their first set. 
  \item For any two $x,y \in \mL^k$ having the same {\csplit} $I_0,I_1,\dots,I_p$, we have that $I_0,I_1,\dots,I_p$  
is a {\csplit} of $f(x,y)$, where $f$ is an arbitrary binary polymorphism of $\Gamma$. This observation follows from the fact that $f$ must preserve $C$.
\end{enumerate}  
\end{observations}

By Lemma \ref{lem:injectivex}, there are $u,v \in \mL$ and a binary injective 
polymorphism $g$ of 
$\Gamma$ such that $g(u,u)g(v,v)|g(u,v)g(v,u)$. 
By Theorem \ref{thm:canpol}, $\Aut(\Gamma)\cup \{g\}$ generates an injective 
binary polymorphism $f$ of $\Gamma$ which is canonical with respect to $(\mL;C,\prec,u,v)$, 
and identical with $g$ on $\{u,v\}$. 
Note that this implies that $f(u,u)f(v,v)|f(u,v)f(v,u)$.

We define two subsets $A_0,A_1$ as we did in Section~\ref{sect:separated}:
arbitrarily choose $u',v'\in \mL$ such that $u'u|v$, $u|v'v$, $u\neq u'$, and $v\neq v'$, and define 
$A_0:=\{x\in \mL:xu'|u \}$ and 
$A_1:=\{x\in \mL:xv'|v\}$. 
By the canonicity of $f$ and Lemma~\ref{lem:dominance}, 
for all $i,j \in \{0,1\}$ the function $f$ is dominated on $A_i \times A_j$ by either the first or the second argument.

Let $S:=\{x\in \mL:uv|x\}$. For each $x\in S$, define $E_x:=\{y\in \mL:xy|uv \}$.

\begin{observations}\label{obs:csplitoutside}
\begin{enumerate}
\item If $y \in E_x$, then $y \in S$ and $E_x=E_y$.
\item The substructure of $(\mL;C,\prec)$ induced by $E_x$ is isomorphic to $(\mL;C,\prec)$.
\item For arbitrary $a,b,c,a',b',c'\in E_x$, we have that $(a,b,c)$ and $(a',b',c')$ are in the same orbit under $(\mL;C,\prec,u,v)$ if and only if  $(a,b,c)$ and $(a',b',c')$ are in the same orbit under $(\mL;C,\prec)$. This implies that for all $x,y\in S$ the function $f$ is canonical 
on $E_x \times E_y$ with respect to $(\mL;C,\prec)$.
The isomorphism from the previous item of the observation and Lemma~\ref{lem:dominance}
 imply that
$f$ is dominated on $E_x\times E_y$ by either the first or the second argument.
\item If $f$ is dominated by the $i$-th argument on $E_x \times E_y$ for some $x,y\in S$ and $i \in \{1,2\}$, 
then $f$ is dominated by the $i$-th argument on $E_{x'}\times E_{y'}$ for all $x',y'\in S$. This follows from the fact that $f$ is canonical with respect to $(\mL;C,\prec,u,v)$.
\end{enumerate}
\end{observations}

Let $X,Y,X',Y'$ be arbitrary subsets of $\mL$.  
We say that $g \colon \mL^2 \to \mL$ \emph{has the same domination} 
on $X\times Y$ and $X'\times Y'$ if $g$ is dominated by the first argument on both $X\times Y$ and $X'\times Y'$, 
or dominated by the second argument on both $X\times Y$ and $X'\times Y'$. Otherwise we say that 
$g$ \emph{has different domination} on $X\times Y$ and $X'\times Y'$. 

Fix $w\in S$.  
Observation~\ref{obs:csplitoutside}(3) implies that $f$ is 
dominated on $E_w \times E_w$ by either the first or the second argument. Hence, one of the following two cases applies:
\begin{enumerate}
\item
$f$ has different domination on $E_w\times E_w$ and $A_i\times A_j$ for some $i,j\in \{0,1\}$, or
\item
$f$ has the 
same domination on $E_w\times E_w$ and $A_i\times A_j$ for every choice of $i,j\in \{0, 1\}$.
\end{enumerate}

We deal with the first case in Lemma~\ref{lem:freeness1} and the second case (with the aid of Lemma \ref{lem:prefree}) in Lemma \ref{lem:freeness}.

\begin{lemma}\label{lem:freeness1}
If $f$ has different domination on $E_w\times E_w$ and $A_p\times A_q$ for some $p,q\in \{0,1\}$, then every relation in
$\langle \Gamma \rangle$ is free.
\end{lemma}
\begin{proof}
Let $R \in \langle \Gamma \rangle$ be of arity $k$,
 and arbitrarily choose 
two tuples $t,t' \in R$ 
such that
\begin{enumerate}
\item
$t$ and $t'$ have the same split vector $s$, and
\item
$t$ and $t'$ are in the same orbit of $k$-tuples of $\Aut(\mL;=)$.
\end{enumerate}
We assume that $f$ 
is dominated by the first argument on $E_w\times E_w$ since otherwise we may consider the 
polymorphism $f'(x,y) := f(y,x)$ instead of $f$. 
Then the function $f$ is dominated by the second argument on $A_p\times A_q$ by assumption. 
Note that $E_w|A_p$ and $E_w|A_q$. 

By the homogeneity of $(\mL;C)$, we can choose 
$\alpha,\beta\in \Aut(\mL;C)$ such that $\alpha(\{t_i:s_i=0\})\subset E_w$, 
$\alpha(\{t_i:s_i=1\})\subset A_p$, $\beta(\{t'_i:s_i=0\})\subset E_w$, and 
$\beta(\{t'_i:s_i=1\})\subset A_q$. Let $t'':=f(\alpha(t),\beta(t'))$. 
The tuples $t,t',t''$ are in 
the same orbit of $\Aut(\mL;=)$ since $f$, $\alpha$, and $\beta$ are injective functions. 
Since $f$, $\alpha$, and $\beta$ preserve $C$, 
we see that $\{t''_i:s_i=0\}|\{t''_i:s_i=1\}$. This implies that $s$ 
is a split vector of $t''$. 
For all $i,j,l \in \{1,\dots,k\}$ such that 
$s_i=s_j=s_l=0$, we have $t''_i|t''_jt''_l$ if and only if 
$t_i|t_jt_l$, since $f$ is dominated by the first argument on $E_w \times E_w$. 
For all $i,j,l \in \{1,\dots,k\}$ such that $s_i=s_j=s_l=1$, 
we have $t''_i|t''_jt''_l$ if and only if $t'_i|t'_jt'_l$,
since $f$ is dominated by the second argument 
on $A_p \times A_q$. Hence, $t''$ has the desired properties in 
Definition \ref{def:freeness}.
\end{proof}

\begin{lemma} \label{lem:prefree}
Assume that $f$ has for all $i,j\in \{0, 1\}$ the 
same domination on $E_w\times E_w$ and $A_i\times A_j$.
Let $R$ be a $k$-ary relation in $\langle \Gamma \rangle$. 
Let $t,t' \in R$ be such that $t$ and $t'$ are in the 
same orbit under $\Aut(\mL;=)$ and have the same {\csplit} $I_0,I_1,\dots,I_p$. Then there is a tuple 
$t'' \in R$ such that 
\begin{itemize}
\item $t,t',t''$ are in the same orbit under $\Aut(\mL;=)$,
\item $I_0,I_1,\dots,I_p$ is a {\csplit} of $t''$,
\item for all $i,j,l \in I_0$ we have $t''_i|t''_j t''_l$ if and only if $t_i|t_jt_l$, and
\item for all $m \in \{1,\dots,p\}$ and $i,j,l\in I_m$ we have $t''_i|t''_jt''_l$ if and only if $t'_i|t'_j t'_l$. 
\end{itemize}
\end{lemma}
\begin{proof}
Suppose without loss of generality that $f$ is dominated by the second argument on $E_w \times E_w$. 
We prove the lemma by induction on $n := |\{t_i : i\in I_0\}|$. 

{\em Base case.} If $n=1$, then $t'':=t'$ has the required properties. 

{\em Inductive step.} 
Assume that the statement holds whenever $n \leq n_0$ for some $n_0 \geq 1$.
We prove that it holds for $n = n_0+1$.
Since $n > 1$ there exists a partition $\{I_{0,0},I_{0,1}\}$ of $I_0$ 
such that $\{t_i : i\in I_{0,0}\}|\{t_i : i\in I_{0,1}\}$. Consider the two cone splits $I_{0,0},I_{0,1},I_1,\dots,I_p$
and $I_{0,1},I_{0,0},I_1,\dots,I_p$ in the sequel.
Note that $|I_{0,0}| < |I_0|$ and $|I_{0,1}| < |I_0|$, so the inductive
hypothesis is applicable to these two cone splits.
By the homogeneity of $(\mL;C)$ there exist $\alpha,\beta \in \Aut(\mL;C)$ such that 
\begin{itemize}
\item
$\alpha(\{t_i : i\in I_{0,0}\})\subset A_0$ and $\alpha(\{t_i : i\in I_{0,1}\}) \subset A_1$, 
\item
$\alpha\big (\{t_i : i\in \bigcup_{j=1}^p I_j\} \big) \subset S$, 
\item
$\beta(\{t'_i : i\in I_0\}) \subset A_0$, and 
\item
$\beta\big(\{t'_i: i \in \bigcup^p_{j=1} I_j \} \big)\subset S$.
\end{itemize}
Note that $\alpha$ 
maps for each $m \in \{1,\dots,p\}$ the set 
$\{t_i : i\in I_m\}$ 
to the set $E_x$ for some $x \in S$. 
Let $h:=f(\alpha(t),\beta(t'))$. By Observation~\ref{obs:csplit}(4), 
we have that $I_0,I_1,\dots,I_p$ is a {\csplit} of $h$. Moreover, the tuple $h$ has the 
following additional properties.
\begin{itemize}
  \item The tuples $t,t',h$ are in the same orbit of $\Aut(\mL;=)$, because
$f$, $\alpha$, and $\beta$ are injective.
  \item $\{h_i : i\in I_{0,0}\}|\{h_i : i\in I_{0,1}\}$. This follows from the observation made
in connection with the definitions of $A_0$ and $A_1$ above.
\item For all $i,j,l \in I_{0,0}$, or $i,j,l \in I_{0,1}$, or $i,j,l\in I_m,m \geq 1$ we have 
$h_i|h_j h_l$ if and only if $t'_i|t'_j t'_l$. This follows from the fact that $f$ is not dominated by the 
second argument on $A_0 \times A_0$, $A_1 \times A_0$, and $E_x \times E_y$ for all $x,y \in S$ --- 
see Observation~\ref{obs:csplitoutside}(4).
\end{itemize}  
The second property implies that $I_{0,0},I_{0,1},I_1,I_2,\dots,I_p$ is a {\csplit} of $t$ and $h$. 
By the inductive assumption, there is a tuple $h'\in R$ such that 
\begin{itemize}
\item $t$, $h$, and $h'$ are in the same orbit under $\Aut(\mL;=)$,
\item $I_{0,0},I_{0,1},I_1,I_2,\dots,I_p$ 
is a {\csplit} of $h'$,
\item for all $i,j,l \in I_{0,0}$ we have $h'_i|h'_j h'_l$ if and only if $t_i|t_j t_l$, and
\item for all $i,j,l \in I_{0,1}$ or $i,j,l\in I_m$, $m\geq 1$, we have $h'_i|h'_j h'_l$ 
if and only if $h_i|h_j h_l$. 
\end{itemize}
We conclude that $h'_i|h'_j h'_l$ if and only if $t'_i|t'_j t'_l$. Then $I_{0,1},I_{0,0},I_1,I_2,\dots,I_p$ is a {\csplit} of $t$ and 
$h'$ (Observation~\ref{obs:csplit}(2)). Applying the inductive assumption to $t$ and $h'$ with this {\csplit}, we obtain 
a tuple $t'' \in R$ with the following properties.
\begin{itemize}
\item
The $k$-tuples $t$, $h'$, and $t''$ are in the same orbit under $\Aut(\mL;=)$.
\item  For all $i,j,l \in I_{0,1}$ we have $t''_i|t''_j t''_l$ if and only if $t_i|t_j t_l$. 
\item For all $i,j,l \in I_{0,0}$ we have $t''_i|t''_j t''_l$ if and only if $h'_i|h'_j h'_l$. 
Hence, $t''_i|t''_j t''_l$ if and only if $t_i |t_j t_l$.
\item For all $i,j,l \in I_m$ and $m\geq 1$ we have 
$t''_i|t''_j t''_l$ if and only if $h'_i|h'_l h'_l$. This
implies that $t''_i|t''_j t''_l$ if and only if $t'_i|t'_j t'_l$.
\end{itemize}
The second and third condition imply that $t''_i|t''_j t''_l$ if and only if $t_i|t_j t_l$ for all $i,j,l \in I_0$. 
Finally, $t,t',t''$ are in the same orbit under $\Aut(\mL;=)$. 
Thus, the tuple $t''$ has all the desired properties.
\end{proof}

\begin{lemma}\label{lem:freeness}
Every relation in $\langle \Gamma \rangle$ is free.
\end{lemma}
\begin{proof} 
Let $f \in \Pol^{(2)}(\Gamma)$ and $w \in \mL$ be as introduced after 
Observation~\ref{obs:csplitoutside}. 
If $f$ has  for some $i,j \in \{0,1\}$ 
different domination on
$E_w \times E_w$ and $A_i \times A_j$, then we are done by Lemma~\ref{lem:freeness1}.
Hence, we may assume that $f$ has the 
same domination on $E_w\times E_w$ and $A_i \times A_j$ for all $i,j\in \{0,1\}$. 
We assume that $f$ is 
dominated by the second argument on $E_w\times E_w$, since otherwise we consider the polymorphism $f'(x,y)= f(y,x)$ instead.

Let $R \in \langle \Gamma \rangle$ be a $k$-ary relation, 
and let $t,t' \in R$ be such that
they lie in the 
same orbit under $\Aut(\mL;=)$ and have a common split vector $s$. We define $J_0:=\{i \in \{1,\dots,k\} : s_i=0\}$ 
and $J_1:=\{i \in \{1,\dots,k\} : s_i=1\}$. 
Clearly, $J_0,J_1$ is a {\csplit} of $t$ and $t'$. Applying Lemma~\ref{lem:prefree} to $t$ and $t'$ with this cone split there exists a tuple $t''$ such that 
\begin{enumerate}
\item
$J_0,J_1$ is a {\csplit} of $t''$,  and thus $s$ is a
split vector of $t''$, 
\item
for all $i,j,l \in J_0$ we have $t''_i|t''_j t''_l$ if and only if $t_i|t_jt_l$, and
\item
for all $i,j,l \in J_1$ we have 
$t''_i | t''_j t''_l$ if and only if $t'_i|t'_j t'_l$. 
\end{enumerate}
Hence, $t''$ satisfies the conditions in Definition \ref{def:freeness}.
\end{proof}


\section{Affine Horn Relations}
\label{sect:Horn}

Throughout this section we work with a reduct $\Gamma$ of $(\mL;C)$ such that $C \in \langle \Gamma \rangle$, and all relations in $\langle \Gamma \rangle$ are separated, free, and have affine splits. We are going to show that all relations in $\Gamma$ 
 can be defined in $(\mL;C)$ 
by quantifier-free formulas 
of a special syntactically restricted form, 
which we call \emph{affine Horn formulas}. 
Affine Horn formulas are introduced in Section~\ref{sect:syntax},
and the mentioned characterisation is shown
in two steps in Section~\ref{sect:inj-syntax}
and Section~\ref{sect:general-syntax}. 
We finally present in Section~\ref{sect:alg} a polynomial-time  algorithm for testing satisfiability of a given affine Horn formula, and this also gives a polynomial-time algorithm for
$\Csp(\Gamma)$ when $\Gamma$ has a finite signature.  

\subsection{Affine Horn Formulas}
\label{sect:syntax}
Recall that a Boolean relation $R$ is called \emph{affine} if 
can be defined by a system of 
 linear equation systems over the 2-element field. 
 It is well-known (see e.g. Chen~\cite{Rendezvous})
that a Boolean relation is affine if and only if it is preserved by the function $(x,y,z) \mapsto x +y + z \; ({\rm mod} \; 2)$. 


\begin{definition}
\label{def:phi-affine}
Let $B \subseteq \{0,1\}^n$ be a Boolean relation. 
Then $\phi_B(z_1,\dots,z_n)$ stands for the formula 
\begin{displaymath}
  z_1 = \cdots = z_n \vee \bigvee_{t \in B \backslash \{(0,0,\dots,0),(1,1,\dots,1)\}} \{z_i : t_i=0\} | \{z_i : t_i=1\} \; .
\end{displaymath}
\smallskip
The formula $\phi_B$ is called \emph{affine}
if $B \cup \{(0,0,\dots,0),(1,1,\dots,1)\}$ is affine. 
\end{definition}


\begin{definition}\label{def:eqHorn}
An \emph{affine Horn clause} is a formula 
\begin{align*}
\text{ of the form } \quad & x_1 \neq y_1 \vee \cdots \vee x_n \neq y_n,\\
\text{ or of the form } \quad & x_1 \neq y_1 \vee \cdots \vee x_n \neq y_n \vee \phi(z_1,\dots,z_k), \end{align*}
where $\phi$ is an affine formula. An \emph{affine Horn formula} is a conjunction of affine Horn clauses.
A relation $R \subseteq \mL^k$ is called \emph{affine Horn} if it can be defined by an affine Horn formula over $(\mL;C)$. 
A phylogeny constraint language is called \emph{affine Horn} if all its relations are affine Horn.
\end{definition}

Note that in Definition~\ref{def:eqHorn},
$n$ can be equal to $0$, and some of the variables $x_1,\dots,x_n$, $y_1,\dots,y_n$, $z_1,\dots,z_k$ might be equal.

\begin{example} 
A relation that is affine Horn is 
\begin{displaymath}
 \{(z_1,z_2,z_3,z_4) \in \mL^4 : z_1 z_2 | z_3 z_4 \text{ and } z_1 = z_2 \Leftrightarrow z_3 = z_4\}.
\end{displaymath}
To see this, first note that it can equivalently be defined by the formula 
\begin{displaymath}
(z_1z_2|z_3z_4 \vee z_1=z_2=z_3=z_4) \wedge (z_1 \neq z_2 \vee z_3 = z_4) \wedge (z_3 \neq z_4 \vee z_1 = z_2) \wedge z_1 \neq z_3 \; .
\end{displaymath}
It is now sufficient to verify
that each conjunct is an affine Horn formula. This is obvious for the second, third, and fourth conjunct.
For the first conjunct,
consider the relation
\begin{displaymath}
R:=\{(0,0,0,0),(1,1,0,0),(0,0,1,1),(1,1,1,1)\} \, .
 \end{displaymath}
 This Boolean relation is affine since $(z_1,z_2,z_3,z_4) \in R$ if and only if $z_1+z_2=0 \; ({\rm mod} \; 2)$ and $z_3+z_4=0 \; ({\rm mod} \; 2)$.
We see that $\phi_R(z_1,z_2,z_3,z_4)$ 
is equivalent to $z_1=z_2=z_3=z_4 \vee z_1z_2|z_3z_4$. 
\end{example}

Let us mention that a consequence of Theorem~\ref{thm:main} below is that all affine Horn relations
are separated, free, and have affine split relations. 
The converse is not true, as we see in the following. 

\begin{example}
Consider the relation $R$ defined as follows.
\begin{displaymath} 
R := \big\{(x,y,z,u) \in \mL^4 : xyz|u \wedge N(x,y,z)\big\}\;.
\end{displaymath}
Note that the relation $R$ is separated, free, and has an affine split relation. However, $R$ is \emph{not} affine Horn. 
To see this, first observe that $N(x,y,z)$ is
equivalent to the primitive positive formula $\exists u. R(x,y,z,u)$, and therefore 
$N \in \langle (\mL;R) \rangle$. The relation $N$ has the split relation
$\{(0,0,1),(0,1,1),(1,1,0),(1,0,0)\}$, which is not affine. 
Since the class of all affine Horn relations is closed
under primitive positive definability (see Corollary~\ref{cor:horn-tx} below), this shows that $R$ is not affine Horn.  
\end{example}



\subsection{The Injective Case}
\label{sect:inj-syntax}
In this section we study $k$-ary relations $R$ from 
$\langle \Gamma \rangle$ 
such that all tuples in $R$ have pairwise
distinct entries. From now on, we let $R$ denote an arbitrary such relation. The main result of this section can be found in Lemma \ref{lem:inj-syntax}.

For $s = (s_1,\dots,s_k) \in S(R)$,
define $R_s := \big\{(t_1,\dots,t_k) \in R : \{t_i:s_i=0\}|\{t_i:s_i=1\}\big\}$. By convention if $s$ is an all-equal tuple (a tuple whose entries are equal), then $R_s:=\top$. Note that $R_s \in \langle \Gamma \rangle$, and that
\begin{align*}
R(x_1,x_2,\dots,x_k) \; \Leftrightarrow \bigvee_{s \in S(R)} R_s(x_1,x_2,\dots,x_k) \; .
\label{eq:Rs}
\end{align*}

Let $i_1,\dots,i_p \in \{1,2,\dots,k\}$ with $i_1 < \cdots < i_p$ and
define 
\[R[i_1,\dots,i_p] := \{(t_{i_1},\dots,t_{i_p}) \mid t \in R\}.\]
Clearly, $R[i_1,\dots,i_p]$ is a member of $\langle \Gamma \rangle$
and our assumptions imply that the split relation $S(R[i_1,\dots,i_p]) \cup \{(0,\dots,0),(1,\dots,1)\}$ is affine.

\begin{lemma}\label{lem:boolreduc}
Let $s \in S(R)$,
and let $1 \leq i_1 < \cdots < i_p \leq k$ be either from $\{i \in \{1,\dots,k\} : s_i=0\}$ or from $\{i \in \{1,\dots,k\} : s_i=1\}$. 
Then $S(R[i_1,\dots,i_p]) = S(R_s[i_1,\dots,i_p])$.
\end{lemma}
\begin{proof}
Clearly, $S(R_s[i_1,\dots,i_p]) \subseteq S(R[i_1,\dots,i_p])$. 
To prove the reverse inclusion, it suffices to show that 
for every $r \in R$ there exists a $t \in R_s$ 
such that $(r_{i_1},\dots,r_{i_p})$ and $(t_{i_1},\dots,t_{i_p})$
have a common split vector. 
The relation $R$ is separated, so there is a $t \in R$ with split vector $s$ and for all $i,j,l \in \{1,\dots,k\}$ it holds that if $s_i=s_j=s_l$, then $t_i|t_j t_l$ if and
only if $r_i|r_jr_l$. 
It follows that $t \in R_s$, and that $(t_{i_1},\dots,t_{i_p})$ and $(r_{i_1},\dots,r_{i_p})$ 
have a common split vector. 
\end{proof}

Define the affine Horn formulas $\theta_R$ and $\psi_R$ as follows.
\begin{align*}
\theta_R(x_1,\dots,x_k) := & \; \bigwedge_{1 \leq i_1 < \cdots < i_p \leq k} 
\phi_{S(R[i_1,\dots,i_p])}(x_{i_1},\dots,x_{i_p}) \\
\psi_R(x_1,\dots,x_k) := & \; \alldiff(x_1,\dots,x_k)  \wedge \theta_R(x_1,\dots,x_n) 
\end{align*}

\begin{lemma}\label{lem:inj-syntax}
The relation $R$ has the affine Horn definition
$\psi_R(x_1,\dots,x_k)$. 
\end{lemma}
\begin{proof}
Our proof is by induction on the arity $k$ of $R$. 
The statement is trivial for $k = 1$. 
For $k>1$, we first show that the following expression defines $R_s$. 
\begin{align}
R_s(x_1,x_2,\dots,x_n) \; \Leftrightarrow \; 
\{x_i:s_i=0\}|\{x_i:s_i=1\} \wedge \psi_R(x_1,\dots,x_n) 
\label{eq:claim}
\end{align}

Let $\{i_1,\dots,i_p\} = \{i \in \{1,\dots,k\} : s_i = 0\}$ and 
$\{j_1,\dots,j_q\} = \{i \in \{1,\dots,k\} : s_i = 1\}$ be
such that
$i_1<i_2<\dots<i_p$ and $j_1<j_2<\dots<j_q$.
Note that $p + q = k$. 
Also note that the relations $P := R_s[i_1,\dots,i_p]$ and $Q := R_s[j_1,\dots,j_q]$
are in $\langle \Gamma \rangle$ and only contain tuples with
pairwise distinct entries.
For notational convenience, let us without loss of generality assume that
$i_1=1,\dots,i_p=p$ and $j_1 = p+1,\dots,j_q = k$. 


By the inductive assumption, $P$ 
has the definition $\psi_P(x_1,\dots,x_p)$ and $Q$ has the definition
$\psi_Q(x_{p+1},\dots,x_k)$.
 By assumption, 
$R_s$ is free so
\begin{align}
R_s(x_1,x_2,\dots,x_n) \; \Leftrightarrow \; &  \{x_1,\dots,x_p\} | \{x_{p+1},\dots,x_k\}  \nonumber \\
& \land \; \psi_{P}(x_1,\dots,x_p) \land \psi_{Q}(x_{p+1},\dots,x_k)  \label{formu:free} \; .
\end{align}

Thus it is sufficient to show that the conjunction on the right-hand side of (\ref{eq:claim}) is equivalent to the conjunction on the right-hand side of (\ref{formu:free}). Let $L = \{l_1,\dots,l_r\}$ be an arbitrary subset of $\{1,2,\dots,k\}$. We can without loss of generality assume that $l_1 < l_2 < \cdots < l_r$. If $L\subseteq \{1,2,\dots,p\}$, then  by definition we have $P = R_s[1,\dots,p]$, and 
$S(P[l_1,\dots,l_r]) = S(R_s[l_1,\dots,l_r]) = S(R[l_1,\dots,l_r])$ by Lemma~\ref{lem:boolreduc}.
Therefore, 
\begin{align*}
\phi_{S(P[l_1,\dots,l_r])}(x_{l_1},\dots,x_{l_r})
= \phi_{S(R[l_1,\dots,l_r])}(x_{l_1},\dots,x_{l_r}) \; .
\end{align*}

Similarly, if $L \subseteq \{p+1,\dots,k\}$, then 
\begin{align*}\label{eq:right}
\phi_{S(Q[l_1-p,\dots,l_r-p])}(x_{l_1},\dots,x_{l_r})
= \phi_{S(R[l_1,\dots,l_r])}(x_{l_1},\dots,x_{l_r}) \; .
\end{align*}

If $L \cap \{1,\dots,p\} \neq \emptyset$
and $L \cap \{p+1,\dots,k\} \neq \emptyset$, 
then $\{x_1,\dots,x_p\}|\{x_{p+1},\dots,x_k\}$ implies 
that $\phi_{S(R[l_1,\dots,l_r])}(x_{l_1},\dots,x_{l_r})$. 
To see this, observe that every $t \in \mL^k$ 
with split vector $s$ satisfies
the disjunct $\{x_{l_i} : s_{l_i}=0\} | \{x_{l_i} : s_{l_i}=1\}$
in the formula $\phi_{S(R[l_1,\dots,l_r])}(x_{l_1},\dots,x_{l_r})$. 
It follows from these three cases of $L$ and the definition of $\psi_R$, $\psi_P$ and $\psi_Q$ that
\begin{align*}
& \{x_i:s_i=0\}|\{x_i:s_i=1\} \wedge 
\psi_{R}(x_1,x_2,\dots,x_k) \\
\Leftrightarrow \quad & 
\{x_i:s_i=0\}|\{x_i:s_i=1\} \wedge  \psi_P(x_{i_1},x_{i_2},\dots,x_{i_p}) \wedge \psi_Q(x_{j_1},x_{j_2},\dots,x_{j_q}) \; ,
\end{align*}
which together with (\ref{formu:free}) implies (\ref{eq:claim}). 
To conclude, we have that
\begin{align*}
R(x_1,x_2,\dots,x_k) \Leftrightarrow & \; \bigvee_{s\in S(R)} R_s(x_1,x_2,\dots,x_k)\\
\Leftrightarrow & \;
\bigvee_{s \in S(R)} \{x_i:s_i=0\}|\{x_i:s_i=1\} \wedge \psi_R(x_1,\dots,x_k) \\
\Leftrightarrow & \; \psi_R(x_1,x_2,\dots,x_k) \wedge 
\bigvee_{s \in S(R)} \{x_i:s_i=0\}|\{x_i:s_i=1\} 
\\
\Leftrightarrow
& \; \psi_R(x_1,x_2,\dots,x_k) \wedge \phi_{S(R)}(x_1,x_2,\dots,x_k) \\
\Leftrightarrow & \; \psi_R(x_1,x_2,\dots,x_k).
\end{align*}
\end{proof}

\subsection{The General Case} \label{sect:general-syntax}
In this part, we finish the proof that every relation $R$ in $\langle \Gamma \rangle$ 
is affine Horn (Proposition~\ref{prop:gen-syntax} below). 
Let $k$ be the arity of $R$. 
For $a \in R$, define 
\[\chi_a(x_1,x_2,\dots,x_k) :=\left(\bigwedge_{a_i=a_j} x_i=x_j \right) \wedge \left( \bigwedge_{a_i\neq a_j} x_i\neq x_j \right) \; .\] 
Let $R_1,R_2,\dots,R_m$ be an enumeration of all relations
that can be defined by a formula 
$\chi_a(x_1,\dots,x_k) \wedge R(x_1,\dots,x_k)$ for some $a \in R$. 
We have only finitely many such relations $R_i$ since $R$ is a union of finitely many orbits of $\Aut(\mL;C)$ and each $R_i$ is a union of some of those orbits. Note that $R_1,R_2,\dots,R_m$ form a partition of $R$, that they are all from $\langle \Gamma \rangle$, and that for distinct $i,j \in \{1,\dots,m\}$ the relations
$R_i$ and $R_j$ are
contained in different orbits under $\Aut(\mL;=)$.
Pick $a^1,\dots,a^m$ such that for every $i \in \{1,\dots,m\}$ the formula $\chi_{a^i}(x_1,\dots,x_k) \wedge R(x_1,\dots,x_k)$ defines $R_i$.

\begin{lemma}\label{lem:eq-bin-inj}
The formula $\bigvee_{i=1}^m \chi_{a^i}$
is preserved by an injection $f \colon \mL^2 \to \mL$. 
\end{lemma}
\begin{proof}
Suppose that $t,t' \in \mL^k$ 
both satisfy the formula. Then $t$ satisfies 
$\chi_{a^i}$ and $t'$ satisfies $\chi_{a^j}$ for 
some $i,j \in \{1,\dots,m\}$. 
By Corollary~\ref{cor:bin-inj}, $\Gamma$ has a binary
injective polymorphism $f$. Since 
$f$ is injective, $f(t,t')$ and $f(a^i,a^j)$ are
in the same orbit under $\Aut(\mL;=)$. 
Since $f(a^i,a^j) \in R$, it follows that
$f(t,t')$ satisfies the formula, too. 
\end{proof}

It follows from Lemma~\ref{lem:eq-bin-inj}
in combination with Lemma~\ref{lem:Horn}
that $\bigvee_{i=1}^m \chi_{a^i}$ has a quantifier-free Horn definition $\psi_0$ over $(\mL;=)$.

Let $1 \leq q(i,1) < \cdots < q(i,p_i) \leq k$
be such that each entry of $a^i$ equals
the $q(i,l)$-th entry of $a^i$, for exactly one 
$l \in \{1,\dots,p_i\}$. The numbers $q(i,j)$ are chosen such that $a^i_{q(i,1)},a^i_{q(i,2)},\dots,a^i_{q(i,p_i)}$ is an enumeration of the elements in $\{a^i_1,a^i_2,\dots,a^i_k\}$, and therefore $p_i=|\{a^i_1,a^i_2,\dots,a^i_k\}|$. Define the affine Horn formula $\sigma_i$
as follows. 
\begin{displaymath}
\sigma_i := \left( \bigvee_{a_j^i = a_l^i} x_j \neq x_l \right) \vee \theta_{R_i[q(i,1),\dots,q(i,p_i)]}. \end{displaymath}

\begin{proposition}\label{prop:gen-syntax}
The relation $R$ has the affine Horn definition
\begin{displaymath}
\psi := \psi_0 \wedge \bigwedge_{i=1}^m \sigma_i(x_1,\dots,x_k)\;.
\end{displaymath}
\end{proposition}

Proposition~\ref{prop:gen-syntax} is the main result of this section. 
Before we present its proof, 
we establish a fact concerning
relations that are separated.

\begin{lemma}\label{lem:booleandomination} 
Let $u,v \in R$ be such that for all $j,l \in \{1,\dots,k\}$ 
if $u_j \neq u_l$ then $v_j \neq v_l$,
and let $i_1,\dots,i_p \in \{1,\dots,k\}$ be such that not all entries of $(u_{i_1},\dots,u_{i_p})$
are equal. 
Then there exists a $w \in R$ such that
\begin{itemize}
  \item $w$ and $v$ lie in the same orbit 
of $k$-tuples of $\Aut(\mL;=)$, and 
  \item $(w_{i_1},\dots,w_{i_p})$ and $(u_{i_1},\dots,u_{i_p})$ have a common split vector that is not an all-equal tuple. 
\end{itemize}  
\end{lemma}
\begin{proof}
Since not all entries of $(u_{i_1},\dots,u_{i_p})$ are equal, so are all entries of $u$ and $v$. Let $s$ be a split vector of $u$ that is not an all-equal tuple. Let $A :=\{s_{i_1},\dots,s_{i_p}\}$. If $|A|=1$, then by applying the separation of $R$ to $t := v$ and $t' := u$ there is a $w\in R$ such that the following 
holds. 
\begin{itemize}
\item For all $j,l \in \{1,\dots,k\}$ we have $w_j \neq w_l$ if and only if $u_j \neq u_l$ or $v_j \neq v_l$. Hence,  $w_j \neq w_l$ if and only if $v_j \neq v_l$ for all $j,l \in \{1,\dots,k\}$, and $w$ and $v$ lie in the same orbit under $\Aut(\mL;=)$.
\item For all $i,j,l \in \{1,\dots,k\}$ we have $w_i|w_j w_l$ 
whenever $s_i=s_j=s_l$ and $u_i|u_j u_l$. Hence, $(u_{i_1},\dots,u_{i_p})$ and $(w_{i_1},\dots,w_{i_p})$ have a common split vector.  
\end{itemize}
Therefore, $w$ has the desired properties. 
If $|A|=2$, then by applying the separation of $R$
 to $t := u$ and $t' := v$, there is a $w \in R$ such that 
the following holds. 
\begin{itemize}
  \item 
Again, $w$ and $v$ 
lie in the same orbit under $\Aut(\mL;=)$.
  \item The tuples $u$ and $w$ have a common split vector. Hence, $(w_{i_1},\dots,w_{i_p})$
  and $(u_{i_1},\dots,u_{i_p})$ have a common split vector. 
\end{itemize}  
Hence, also in this case $w$ has the desired properties.
\end{proof}

\begin{proof} (Proposition~\ref{prop:gen-syntax})
We first show that every $k$-tuple $t$ that satisfies $\psi$ is a member of $R$. Since $t$ satisfies $\psi_0$, there exists an $i \in \{1,\dots,m\}$
such that $t$ satisfies $\chi_{a^i}$. 
That is, for all $j,l$ it holds that 
$t_j=t_l$ if and only if $a^i_j=a^i_l$.
Let $p := p_i$ and 
$i_1 := q(i,1),\dots, i_p := q(i,p)$. 
Since $t$ satisfies $\sigma_i$  
it must therefore also satisfy $\theta_{R_i[i_1,\dots,i_p]}$. Since $R_i[i_1,\dots,i_p] \in \langle \Gamma \rangle$ only contains tuples
with pairwise distinct entries, 
we have that $\alldiff(x_{i_1},\dots,x_{i_p}) \wedge \theta_{R_i[i_1,\dots,i_p]}(x_{i_1},\dots,x_{i_p})$
defines $R_i[i_1,\dots,i_p]$
by Lemma~\ref{lem:inj-syntax}.
The tuple $t$ satisfies $\alldiff(x_{i_1},\dots,x_{i_p})$, and it follows that 
$t \in R_i \subseteq R$. 

It remains to be shown that every $t \in R$ satisfies $\psi$. Clearly, $t$ satisfies $\psi_0$.
Let $i \in \{1,\dots,m\}$; 
we have to verify that $t$ satisfies $\sigma_i$. 
If there are indices $j,l \in \{1,\dots,k\}$ such that $t_j \neq t_l$ and $a^i_j=a^i_l$, then $t$ 
satisfies $\sigma_i$ since $\sigma_i$ contains the disjunct $x_j \neq x_l$.
We are left with the
case that for all $j,l \in \{1,\dots,k\}$ 
if $t_j \neq t_l$, then $a^i_j \neq a^i_l$. 
Again, let $p := p_i$ and $i_1 := q(i,1),\dots, i_p := q(i,p)$.  
In order to show that $t$ satisfies 
$\theta_{R_i[i_1,\dots,i_p]}$, 
we have to show that 
$(t_{i_1},\dots,t_{i_p})$ has a split vector from $S(R_i[i_1,\dots,i_p])$. 
Lemma~\ref{lem:booleandomination} 
applied to $u := t$ and $v := a^i$ shows
that there exists a $w \in R$ such that $w$
and $a^i$ lie in the same orbit under $\Aut(\mL;=)$,
and $(w_{i_1},\dots,w_{i_p})$ and $(t_{i_1},\dots,t_{i_p})$ have a common split vector. 
Since $w \in R_i$, this split vector is in $S(R_i[i_1,\dots,i_p])$. This concludes the proof
that $t$ satisfies $\sigma_i$.
\end{proof}

\begin{corollary}
\label{cor:syntax}
Let $\Gamma$ be a reduct of $(\mL;C)$ such that 
$C \in \langle \Gamma \rangle$ and $N \notin \langle \Gamma \rangle$. Then all relations
in $\langle \Gamma \rangle$ are affine Horn. 
\end{corollary}
\begin{proof} 
By Lemma~\ref{lem:affine}, \ref{lem:separation}, and \ref{lem:freeness},
all relations in $\langle \Gamma \rangle$ have affine splits,  are free, and separated. By Proposition~\ref{prop:gen-syntax}, all relations
in $\langle \Gamma \rangle$ are affine Horn. 
\end{proof}

\subsection{Testing Satisfiability of Affine Horn Formulas}\label{sect:alg}
We will show (in Theorem \ref{thm:alg}) that satisfiability
of affine Horn formulas over $(\mL;C)$ can be solved in polynomial time. In the formulation of the algorithm, we need the following concept which originates 
from Bodirsky \& Mueller~\cite{phylo-long}.

\begin{definition}
Let $\Phi$ be an affine Horn formula. 
Then the \emph{split problem $\Psi$ for $\Phi$} is the Boolean constraint
satisfaction problem on the same variables as $\Phi$ 
which contains for each conjunct of $\Phi$
of the form $\phi_R(z_1,\dots,z_k)$
the (affine) Boolean constraint $R(x_1,\dots,x_n)$.
\end{definition}

When $\Phi$ is a formula and $X$ a non-empty set of variables of $\Phi$ then the \emph{contraction of $X$ in $\Phi$}
is the formula obtained from $\Phi$ by 
\begin{itemize}
\item replacing all variables from $X$ in $\Phi$ by a new variable $x$ and 
\item subsequently removing all disjuncts of the form $x \neq x$ and all conjuncts of the form $x=x$.
\end{itemize}

\begin{figure*}
\begin{center}
\fbox{
\begin{tabular}{l}
Solve($\Phi$): \\
// Input: An affine Horn formula $\Phi$ \\
// Output: \emph{satisfiable} or \emph{unsatisfiable} \\
If Spec$(\Phi)=$ \emph{satisfiable} then return \emph{satisfiable} \\
  else  \\
\hspace{.5cm}    $X :=$ Spec$(\Phi)$ \quad // $X$ is a set of variables. \\ 
\hspace{.5cm}      Let $\Psi$ be the contraction of $X$ in $\Phi$. \\
\hspace{.5cm}      If $\Psi$ contains an empty clause then return \emph{unsatisfiable} \\
\hspace{.5cm}      else return Solve($\Psi$) \\
end if 
\end{tabular} } \\
\fbox{
\begin{tabular}{l}
Spec$(\Phi)$: \\
// Input: An affine Horn formula $\Phi$ with variables $V$ \\
// Output: \emph{satisfiable}, or a subset $X$ of $V$ \\
If there is no non-trivial solution to the split problem $\Psi$ for $\Phi$ \\
\hspace{.5cm} return $V$ \\
else \\
\hspace{.5cm} Let $s$ be the non-trivial solution to $\Psi$. \\
\hspace{.5cm} If Spec$(\Phi[s^{-1}(0)]) = X_0 \subseteq V$ then return $X_0$ \\
\hspace{.5cm} else if Spec$(\Phi[s^{-1}(1)]) = X_1 \subseteq V$ then return $X_1$ \\
\hspace{.5cm} else return \emph{satisfiable} \\
\hspace{.5cm} end if \\
end if
\end{tabular}
}
\end{center}
\caption{A polynomial-time procedure for satisfiability of affine Horn formulas.}
\label{fig:algorithm}
\end{figure*}

\begin{lemma}\label{lem:neumann}
Let $L_1,L_2$ be two finite subsets of $\mathbb L$.
Then there is an automorphism $\alpha$ of $({\mathbb L};C)$ such that
$\alpha(L_1)|L_2$. 
\end{lemma}
\begin{proof}
An immediate consequence of the fact that every finite leaf structure embeds into $(\mL;C)$, and 
the homogeneity of $(\mL;C)$. 
\end{proof}

\begin{theorem}\label{thm:alg}
There is a polynomial-time algorithm that decides
whether a given affine Horn formula is satisfiable or not
over $(\mL;C)$. 
\end{theorem}

\begin{proof}
The algorithm can be found in Figure~\ref{fig:algorithm}. 
The correctness of the algorithm directly follows from the claim that the sub-procedure Spec described
in Figure~\ref{fig:algorithm} 
has the following properties:
\begin{itemize}
\item If Spec$(\Phi)$ returns \emph{satisfiable} then there exists
an \emph{injective} solution to $\Phi$, that is, a solution to $\Phi$ where all
variables take distinct values in ${\mathbb L}$. 
\item If Spec$(\Phi)$ returns a set of variables $X$, then the variables from $X$ take equal value in every solution of $\Phi$ ($\Phi$ may not have a solution).
\end{itemize}

We prove the claim by induction over the recursive structure of Spec.
The split problem $\Psi$ of $\Phi$ is an affine Boolean CSP, and so it can be decided in polynomial time 
by Gaussian elimination whether $\Psi$ has a non-constant solution or not.
If $\Phi$ has a non-injective solution, then the split problem for $\Phi$ has
a non-constant solution (the non-trivial solution is induced by the left and the right children of the root of the solution).
Hence, if the split problem does not have a non-constant solution,
then all variables from $\Phi$ must take equal values, and 
the output $V$ of the algorithm satisfies the claim made above.

So suppose that the split problem $\Psi$ for $\Phi$ 
does have a non-trivial solution $s$, and let $S_0 := s^{-1}(0)$
and $S_1 := s^{-1}(1)$.
If one of the recursive calls in Spec does not return \emph{satisfiable} 
but rather returns a set of variables $X$, then the correctness of the claim
follows by the inductive assumption: the reason
is that when a subset of the constraints in $\Phi$ implies that some variables must denote equal values in all solutions, then also $\Phi$ implies that 
those variables must be equal in all solutions. 
Hence, when Spec returns $X$ on input $\Phi$, this answer is correct.  

If both recursive calls return \emph{satisfiable},
then we argue that there exists a solution to $\Phi$.
We know that there are injective solutions $t_1 \colon S_1 \rightarrow \mathbb L$ and $t_2 \colon S_2 \rightarrow \mathbb L$ to $\Phi[S_0]$ and $\Phi[S_1]$, 
respectively. Let $t \colon V \rightarrow \mathbb L$
be the mapping such that $t(x)=t_1(x)$ for $x \in S_1$,
and $t(x) = \alpha(t_2(x))$ for $x \in S_2$ and an
automorphism $\alpha$ of $(\mL;C)$ such that $t(S_1) | t(S_2)$ 
(such an automorphism $\alpha$ exists due to Lemma~\ref{lem:neumann}).
We claim that $t$ is an (injective) solution to $\Phi$.

Let $\phi$ be a conjunct from $\Phi$. If
$\phi$ contains disjuncts of the form $x \neq y$,
then $\phi$ is satisfied by $t$ since $t$ is injective. 
Otherwise, if all variables of $\phi$
are in $S_0$ or all are in $S_1$, 
then $\phi$ is satisfied by $t$ by 
inductive assumptions for $t_1$ and $t_2$,
respectively. 
If $\phi$ contains variables from both sides,
suppose that $x_1,\dots,x_n$ are the free variables of $\phi$.
Let $R$ be the $n$-ary Boolean relation such that $\phi = \phi_R$. 
Then $R$ must contain a tuple $(p_1,\dots,p_n)$ such that $x_i \in S_0$
if and only if $p_i=0$, 
since $s$ is a solution to the split problem. But
then $\phi_R$ is satisfied by $t$ since it contains the disjunct 
$\{x_i : p_i=0\} | \{x_i : p_i=1\}$.

We finally show that
the running time is polynomial in the size of the input. 
There are at most $n-1$ variable contractions that can be performed,
and this bounds the number of recursive calls of the procedure Solve.
Finally, in the procedure Spec we have to solve an at most linear number of Boolean affine split problems, which can be done in polynomial time as well.
\end{proof}

\begin{corollary}\label{cor:alg}
Let $\Gamma$ be a reduct of $(\mL;C)$ which is affine Horn and has a finite signature. 
Then $\Csp(\Gamma)$ can be solved in polynomial time. 
\end{corollary}
\begin{proof}
Let $\Psi$ be an instance of $\Csp(\Gamma)$. 
Each conjunct of $\Psi$ has a definition over $(\mL;C)$ by a conjunction of affine Horn formulas. Since $\Gamma$ contains only finitely many 
relation symbols, replacing each conjunct $\psi(x_1,\dots,x_n)$ by its defining formula over $(\mL;C)$ only changes the size of the formula by a constant factor. The resulting formula $\Phi$ is a conjunction of affine Horn formulas, 
and is satisfiable over $(\mL;C)$ if and only if $\Psi$
is satisfiable over $\Gamma$.
\end{proof}

\section{Affine Tree Operations}
\label{sec:treeoperations}
The border between NP-hardness and 
polynomial-time tractability for phylogeny problems
can be stated in terms of polymorphisms, as announced in
Theorem~\ref{thm:main-polym}. 
In order to prove this result, we introduce a certain
kind of binary operations over $\mL$ which we call \emph{affine tree operations}. In this section we often use the fact that for every non-empty finite subset $X$ of $\mL$ there is a unique partition $\{X_1,X_2\}$ of $X$ such that $X_1|X_2$ and $X_1\prec X_2$. This fact follows easily from Lemma 9 in \cite{BodJonsPham} and the convexity of $\prec$.

To define affine tree operations, we need the concept of \emph{perfect dominance} which is stronger than the notion of domination introduced in Section~\ref{sec:dominance}.
Let $U,V$ be two finite subsets of $\mL$. A function $f \colon \mL^2 \to \mL$ is called \emph{perfectly dominated by the first argument on $U \times V$} if the following conditions holds.
\begin{itemize}
\item For all $u_1,u_2,u_3 \in U$ and $v_1,v_2,v_3 \in V$, if $u_1|u_2u_3$ then
\begin{displaymath}
f(u_1,v_1)|f(u_2,v_2)f(u_3,v_3)\;.
 \end{displaymath}
\item for all $u \in U$ and $v_1,v_2,v_3 \in V$, if $v_1|v_2v_3$ then $f(u,v_1)|f(u,v_2)f(u,v_3)$. 
\end{itemize}
Similarly, $f \colon \mL^2 \to \mL$
is called \emph{perfectly dominated by the second argument on $U \times V$} if the function $(x,y)\mapsto f(y,x)$ is perfectly dominated by the first argument on $U \times V$. 

Let $f \colon \mL^2\to \mL$ be an injective function, and $U$ be a finite subset of $\mL$. We inductively define whether $f$ is \emph{semidominated on $U\times U$} as follows.
\begin{itemize}
  \item If $U=\emptyset$ or $|U|=1$ then $f$ is semidominated on $U\times U$.
  \item Otherwise, $f$ is semidominated on $U\times U$ if 
  there are subsets $U_1,U_2 \subseteq U$ such that $U=U_1\cup U_2$, $U_1|U_2$, and 
  the following conditions hold.
\begin{itemize}
  \item $f$ is semidominated on $U_1\times U_1$ and $U_2\times U_2$,
  \item $f(U_1\times U_1)|f(U_2\times U_2)$ and $f(U_1\times U_2)|f(U_2\times U_1)$,
  \item $f((U_1\times U_1)\cup (U_2\times U_2))|f((U_1\times U_2)\cup (U_2\times U_1))$, and
  \item $f$ is perfectly dominated by the first argument on $U_1\times U_2$ and $f$ is perfectly dominated by the second argument on  $U_2\times U_1$.
\end{itemize}  
\end{itemize}  

\begin{definition}
An operation $f \colon \mL^2\to \mL$ is called an \emph{affine tree operation} if 
$f$ is semidominated on $U\times U$
for every finite subset $U$ of $\mL$.
\end{definition}

\subsection{Existence of Affine Tree Operations}
\label{sec:txoperation}
We prove the existence of an affine tree operation
which we denote by $\tx$. 
We start with a finite version of this statement.
We write $U \prec V$ if $u \prec v$ for every $u \in U$ and $v \in V$. 
In the following, we use the notation $f\upharpoonright_X$ to denote the restriction of a function $f$ on a subset $X$ of the domain of $f$.

\begin{lemma}\label{lem:finitetx}
For every finite subset $X$ of $\mL$ there is a function $f \colon X\times X\to \mL$ such that 
\begin{itemize}
\item for every non-empty subset $U$ of $X$ the function $f$ is semidominated on $U\times U$, and 
\item 
for all $U_0,U_1 \subseteq X$ with $U_0|U_1$ and $U_0 \prec U_1$
 we have that $f$
is perfectly dominated by the first argument on $U_0 \times U_1$. 
\end{itemize}
\end{lemma}
\begin{proof}
The proof is by induction on $|X|$. The claim is trivial if $|X|=1$ since the function $f \colon X\times X\to \mL$ given by $f(x,x):=a$ for an arbitrary $a \in \mL$ has the required properties. Suppose that $|X| \geq 2$. Let $\{X_0,X_1\}$ be a partition of $X$ such that $X_0|X_1$ and $X_0 \prec X_1$. By the inductive assumption there are functions $f_{0,0} \colon X_0\times X_0\to \mL$ and $f_{1,1} \colon X_1\times X_1\to \mL$ such that $f_{0,0}$ and $f_{1,1}$ satisfy the conditions of the claim for $X_0$ and for $X_1$, respectively. We can assume that 
$f_{0,0}(X_0\times X_0)|f_{1,1}(X_1\times X_1)$ since otherwise choose $\alpha,\beta\in \Aut(\mL;C)$ such that $\alpha(f_{0,0}(X_0\times X_0))|\beta(f_{1,1}(X_1\times X_1))$. Then we continue the argument with $\alpha \circ f_{0,0}$ and $\beta\circ f_{1,1}$ instead of $f_{0,0}$ 
and $f_{1,1}$. 

Let $f_{0,1} \colon X_0\times X_1\to \mL$ and $f_{1,0} \colon X_1\times X_0\to \mL$ be such that $f_{0,1}$ is perfectly dominated by the first argument on $X_0 \times X_1$ and $f_{1,0}$ is perfectly dominated by the second argument on $X_1 \times X_0$.  We can assume that 
\[f_{0,0}(X_0\times X_0)\cup f_{1,1}(X_1\times X_1) \big |f_{0,1}(X_0\times X_1)\cup f_{1,0}(X_1\times X_0)\]
and that 
\[f_{0,1}(X_0\times X_1)\big |f_{1,0}(X_1\times X_0)\] 
by reasoning as above. Let $f \colon X\times X\to \mL$ be given by $f(x,y):=f_{i,j}(x,y)$ if $x\in X_i,y\in X_j$. 
We show that $f$ satisfies the two conditions of the statement. 
If $U\subseteq X_0$ (resp.  $U\subseteq X_1$) then we are done since $f_{0,0}$ (resp. $f_{1,1}$) is semidominated on $U \times U$. Otherwise let $\{U_0,U_1\}$ be a partition of $U$ such that $U_0|U_1$ and $U_0\prec U_1$. Clearly we have $U_0\subset X_0$ and $U_1\subset X_1$. It is straightforward to verify that 
\begin{align*}
&& f(U_0\times U_0)\cup f(U_1\times U_1) \; \big | & \; f(U_0\times U_1)\cup f(U_1\times U_0) \, , \\
&&  f(U_0\times U_0) \; \big | & \; f(U_1\times U_1)\, , \\
\text{ and} && f(U_0\times U_1) \; \big | & \; f(U_1\times U_0) \, .
\end{align*}
Since $f_{0,1}$ (resp. $f_{1,0}$) is perfectly dominated by the first argument (resp. the second argument) on $X_0\times X_1$ (resp. on $X_1\times X_0$), it follows that $f$ is perfectly dominated by the first argument on $U_0\times U_1$ (resp. perfectly dominated by the second argument on $U_1\times U_0$). By the inductive assumption $f_{0,0}$ (resp. $f_{1,1}$) is semidominated on $U_0\times U_0$ (resp. $U_1\times U_1$), and so we have that $f$ is semidominated on $U_0\times U_0$ and on $U_1\times U_1$. 
\end{proof}

We can now prove that there exists an affine tree operation.

\begin{proposition}\label{prop:txexistence}
There exists an affine tree operation which we call $\tx$.
The operation $\tx$ has the property that for all
finite $X \subset \mL$ there exists an $\alpha \in \Aut(\mL;C)$ such that $\tx(x,y) = \alpha(\tx(y,x))$ for
all $x,y \in X$. 
\end{proposition}
\begin{proof}
Let $X$ be an non-empty finite subset of $\mL$. Let $f,g \colon X\times X\to \mL$ be two arbitrary 
binary functions that satisfy the two conditions of Lemma~\ref{lem:finitetx}. 
We prove by induction on $|X|$ that there is an automorphism $\alpha\in\Aut(\mL;C)$ such that $f(x,y)=\alpha(g(x,y))$ 
for arbitrary $x,y\in X$. This is trivial 
if $|X|=1$. If $|X|\geq 2$, let $\{X_0,X_1\}$ be a partition of $X$ such that $X_0|X_1$ and $X_0 \prec X_1$. 
By the inductive assumption, there exist $\alpha_{0,0},\alpha_{1,1}\in \Aut(\mL;C)$ such 
that $f(x,y)=\alpha_{0,0}(g(x,y))$ for all $x,y\in X_0$ and $f(x,y)=\alpha_{1,1}(g(x,y))$ 
for all $x,y\in X_1$. Since $f$ and $g$ are perfectly dominated by the first argument 
on $X_0\times X_1$ and by the second argument on $X_1\times X_0$, there are  
$\alpha_{0,1},\alpha_{1,0}\in \Aut(\mL;C)$ such that $f(x,y)=\alpha_{0,1}(g(x,y))$ 
for any $(x,y)\in X_0\times X_1$ and $f(x,y)=\alpha_{1,0}(g(x,y))$ for any 
$(x,y)\in X_1\times X_0$. Let $\beta \colon g(X\times X)\to f(X\times X)$ be given 
by $\beta(g(x,y))=\alpha_{i,j}(g(x,y))$ when $x\in X_i$ and $y\in X_j$. It follows 
from the conditions
\begin{align*}
f(X_0\times X_0)\cup f(X_1\times X_1) \; \big | & \; f(X_0\times X_1)\cup f(X_1\times X_0) \\
f(X_0\times X_0)\; \big | & \; f(X_1\times X_1) \\ 
f(X_0\times X_1)\; \big | & \; f(X_1\times X_0)\;,
\end{align*}
and 
\begin{align*}
g(X_0\times X_0)\cup g(X_1\times X_1)\; \big | & \; g(X_0\times X_1)\cup g(X_1\times X_0) \\
g(X_0\times X_0) \; \big | & \; g(X_1\times X_1) \\
g(X_0\times X_1) \; \big | & \; g(X_1\times X_0)
\end{align*} 
that $\beta$ is a partial isomorphism from $(\mL;C)$ to $(\mL;C)$.  By the homogeneity of $(\mL;C)$, it can be extended to an automorphism $\alpha$ of $(\mL;C)$. The claim follows.

This claim combined with Lemma~\ref{lem:finitetx} implies that for arbitrary finite subsets $X,Y$ of $\mL$ such that $X\subseteq Y$ 
and any function $f \colon X\times X\to \mL$ which satisfies the conditions in Lemma~\ref{lem:finitetx}, there is a function 
$f' \colon Y\times Y\to \mL$ which (1) satisfies the conditions in Lemma~\ref{lem:finitetx} and (2) satisfies $f' \upharpoonright_{X\times X}=f$. 
Since $\mL$ is 
countable, it follows that there exists an operation $\tx \colon \mL^2 \to \mL$ such that 
$\tx$ is semidominated on $U\times U$ for every finite subset $U$ of $X$. 

We prove the second statement of the proposition by induction on $|X|$. Clearly the claim holds if $|X|\leq 1$. 
We consider the case $|X|\geq 2$. Let $\gamma \colon 
\tx(X^2) \to \tx(X^2)$ be given by $\gamma(\tx(x,y))=\tx(y,x)$ for any $x,y\in X$. We will show that $\gamma$ is a partial isomorphism of $(\mL;C)$. Let $X_1,X_2$ be a partition of $X$ such that $X_1\cup X_2=X$ and $X_1|X_2$ holds in $(\mL;C)$. We can assume that $X_1\prec X_2$. Since $\tx$ is perfectly dominated by the first argument on $X_1\times X_2$, perfectly dominated by the second argument on $X_2\times X_1$ and $\tx(X_1\times X_2)|\tx(X_2\times X_1)$, the map $\gamma\upharpoonright_{\tx((X_1\times X_2)\cup (X_2\times X_1))}$ is a partial isomorphism of $(\mL;C)$. By the inductive assumption we also have that $\gamma\upharpoonright_{\tx(X_1\times X_1)}$ and $\gamma\upharpoonright_{\tx(X_2\times X_2))}$ are partial isomorphisms of $(\mL;C)$. Let $A_1:=\tx(X_1\times X_1),A_2:=\tx(X_2\times X_2), A_3:=\tx((X_1\times X_2)\cup (X_2\times X_1))$. By the properties of $\tx$ we have $A_i|A_j$ for all $i\neq j$, and $\gamma(A_i)=A_i$ for $1\leq i\leq 3$. It follows that $\gamma$ is a partial isomorphism of $(\mL;C)$ which can be extended to an automorphism $\alpha$ of $\Aut(\mL;C)$ by the homogeneity of $(\mL;C)$. 
\end{proof}

\subsection{The Operation $\tx$ and Affine Horn Formulas}
Note that $\tx$ is injective and preserves $C$.
In fact, $\tx$ preserves the much larger class of affine Horn formulas, too. We first show that $\tx$ preserves affine formulas (Definition~\ref{def:phi-affine}).

\begin{lemma}\label{lem:affinepreservation}
Let $R \subseteq \{0,1\}^k$ be such that $R\cup \{(0,\dots,0),(1,\dots,1)\}$ is affine. Then $\tx$ preserves the formula $\phi_R$.
\end{lemma}
\begin{proof} 
Let $a=(a_1,a_2,\dots,a_k)$ and $b=(b_1,b_2,\dots,b_k)$ be two tuples from $\mL^k$ that satisfy $\phi_R$. 
Let $A$ and $B$ denote the sets 
$\{a_1,a_2,\dots,a_k\}$ 
and $\{b_1,b_2,\dots,b_k\}$, respectively. 
We first consider the case $|A|=1$. 
If $|B|=1$ then the claim is trivial. Assume
instead that $|B|>1$. 
The following cases are exhaustive.
\begin{itemize}
  \item $A|B$. Since $\tx$ is perfectly dominated on $A\times B$, it follows that $\tx(a,b)$ and $b$ have a common split vector. 
Hence, $\tx(a,b)$ satisfies $\phi_R$.
\item There are $B_0,B_1 \subset B$ such that $B_0|B_1$ and $(B_0\cup A)|B_1$. Since $\tx$ is semidominated on 
$A\cup B$, we have $\tx((B_0\cup A)\times (B_0\cup A))|\tx((B_0\cup A)\times B_1)$ so $\tx(A\times B_0)|\tx(A\times B_1)$. 
This implies that $\tx(a,b)$ and $b$ have a common split vector,
and therefore $\tx(a,b)$ satisfies $\phi_R$. 
\end{itemize}   
We argue similarly in the case $|B|=1$. It remains to consider the case that $|A| \geq 2$ and $|B| \geq 2$. 
Let $X:=A\cup B=\{x_1,x_2,\dots,x_m\}$. Let $x$ denote 
the tuple $(x_1,x_2,\dots,x_m)$, and let
$s$ be a split vector of $x$. 
In the following, we view $s$ as a function 
from $\{x_1,\dots,x_m\}$ to $\{0,1\}$,
mapping $x_i$ to $s_i$. If $s$ is constant
on $A$ and on $B$ then $\tx$ is perfectly dominated 
on $A\times B$. 
This implies that $\tx(a,b)$ has a common
split vector with $a$ or with $b$. Hence, 
$\tx(a,b)$ satisfies $\phi_R$, and we are done. 

Next, consider the case that $s$ is constant on $A$, but not on $B$. 
Let $B_0:=\{b_i: s(b_i)=s(a_1)\}$ 
and $B_1:=\{b_i:s(b_i)\neq s(a_1)\}$. 
Since $(A\cup B_0)|B_1$, it follows that $\tx((A\cup B_0)^2)|\tx((A\cup B_0)\times B_1)$, because
$\tx$ is semidominated. 
Therefore, $\tx(A\times B_0)|\tx(A\times B_1)$. 
Hence, $\tx(a,b)$ has a common split vector with $b$, and hence satisfies $\phi_R$. 
We argue similarly for the case when $s$ is constant on $B$, but not constant on $A$. We are left with the case that $s$ is neither constant on $A$ nor on $B$. 
Let $\{X_0,X_1\}$ be a partition of $X$ such that $X_0|X_1$.  
Consider the case that $a$ and $b$ have a common split vector $s'$. It follows from $\tx(X_0\times X_0)|\tx(X_1\times X_1)$ that $\tx(a,b)$ also has the split vector $s'$, and hence $\tx(a,b)$ satisfies $\phi_R$. 

Finally, suppose that $a$ and $b$ do not have a common split vector. Let $s' \in R$ be a split vector of $a$, and let $s'' \in R$ be a split vector of $b$. Let $U:=\{\tx(a,b)_i : s'_i=s''_i\}$  and 
$V:=\{\tx(a,b)_i:s'_i\neq s''_i\}$; by assumption, $U$ and $V$ 
are non-empty. 
Then $f((X_0\times X_0)\cup (X_1\times X_1))|f((X_0\times X_1)\cup (X_1\times X_0))$ implies that $U|V$. Therefore, $s' \oplus s''$ is a split vector of $\tx(a,b)$. Since $R\cup \{(0,0,\dots,0),(1,1,\dots,1)\}$ is affine, $s' \oplus s'' \in R$, and
$\tx(a,b)$ satisfies $\phi_R$. 
\end{proof}

We will now generalise Lemma \ref{lem:affinepreservation} to arbitrary affine Horn formulas.

\begin{proposition}\label{prop:txpreservation}
The function $\tx$ preserves all affine Horn formulas. 
\end{proposition}
\begin{proof}
It suffices to show that $\tx$ preserves affine Horn clauses, this is, formulas
of the form 
\begin{displaymath}
x_1\neq y_1\vee x_2\neq y_2\vee \cdots \vee x_n\neq y_n \vee \phi(z_1,z_2\dots,z_k)\;, \end{displaymath}
where $\phi$ is an affine formula. Let 
$u,u'$ be two tuples which satisfy this clause, and let
$u'' := \tx(u,u')$. 
If $u$ or $u'$ satisfies $x_i \neq y_i$ for some $i \in \{1,\dots,n\}$,
then $u''$ satisfies $x_i \neq y_i$, too, since $\tx$ is injective. Hence, $u''$ satisfies the clause. 
If $u$ and $u'$ satisfy $x_i = y_i$ for all 
$1\leq i \leq k$, then $u$ 
and $u'$ must satisfy $\phi$. We have seen in 
Lemma~\ref{lem:affinepreservation} that
$\tx$ preserves the affine formula $\phi$,
so it follows that $u''$ satisfies $\phi$, and therefore
satisfies the clause.
\end{proof}

\begin{corollary}\label{cor:horn-tx}
Let $\Gamma$ be a reduct of $(\mL;C)$.
Then the following are equivalent. 
\begin{enumerate}
\item $\Gamma$ is preserved by $\tx$;
\item all relations in $\langle \Gamma \rangle$ are affine Horn;
\item all relations in $\Gamma$ are affine Horn.
\end{enumerate}
\end{corollary}
\begin{proof}
1 implies 2. The operation $\tx$ preserves $C$.
Hence, the expansion $\Gamma'$ of $\Gamma$ 
by the additional relation $C$ is also preserved by $\tx$. 
Note that the operation $\tx$ 
does not preserve the relation $N$. To see this,
arbitrarily choose pairwise distinct elements 
$x,y,y',z \in \mL$ such that $xy|z$ and $x|y'z$. 
Clearly, $(x,y,z) \in N$ and $(x,y',z) \in N$. 
By using the semidomination property of $\tx$ for 
$U:=\{x,y\}$ and $V:=\{y',z\}$ (note that $U|V$), we have $f(x,x)f(z,z)|f(y,y')$. This implies that $(f(x,x),f(y,y'),f(z,z))$ 
is not in $N$. Hence, $\tx$ does not preserve $N$,
and $N \notin \langle \Gamma' \rangle$ by Theorem~\ref{thm:inv-pol}. 
Corollary~\ref{cor:syntax} shows that all relations
in $\langle \Gamma \rangle$ are affine Horn. 

2 implies 3. Trivial. 

3 implies 1. Follows from Proposition~\ref{prop:txpreservation}. 
\end{proof}

\subsection{Symmetry Modulo Endomorphisms}
\label{sect:lift}
In this section we prove the existence of endomorphisms $e_1,e_2$
of $(\mL;C)$ such that
\begin{displaymath} 
e_1(\tx(x,y)) = e_2(\tx(y,x)) \; .
\end{displaymath}

The idea of the following lemma comes from the proof of 
 Proposition~6.6 in Bodirsky, Pinsker, and Pongracz~\cite{BPP-projective-homomorphisms}. 
 
\begin{lemma}\label{lem:lift}
Let $\Gamma$ be $\omega$-categorical, and
$f \in \Pol^{(2)}(\Gamma)$. 
Suppose that for every finite subset $A$ of 
the domain $D$ of $\Gamma$ there exists an
$\alpha \in \Aut(\Gamma)$ such
that $f(x,y) = \alpha(f(y,x))$ for all $x,y \in A$. Then 
there are $e_1,e_2 \in \overline{\Aut(\Gamma)}$
such that $e_1(f(x,y)) = e_2(f(y,x))$ for all $x,y \in D$. 
\end{lemma}
\begin{proof}
Construct a rooted tree as follows. 
Each vertex of the tree lies on some level $n \in \mathbb N$.
Let $d_1,d_2,\dots$ be an enumeration
of $D$. Let $F_n$ be the set of partial isomorphisms of
$\Gamma$ with domain $D_n := \{d_1,\dots,d_n\}$, and define
the equivalence relation $\sim$ on $F_n^2$ as follows:
$(\alpha_1,\alpha_2) \sim (\beta_1,\beta_2)$
if there exists a $\delta \in \Aut(\Gamma)$
such that $\alpha_i = \delta \circ \beta_i$ for $i \in \{1,2\}$.
Note that for each $n$, the relation $\sim$ has finitely
many equivalence classes on $F_n^2$, by the $\omega$-categoricity of $\Gamma$ and Theorem~\ref{thm:Ryll}. 

Now, the vertices of the tree on level $n$ are precisely 
the equivalence classes $E$ of $\sim$ on $F_n^2$ 
such that for every (equivalently, for some) $(\alpha_1,\alpha_2) \in E$ and $x,y \in D$
satisfying $\{f(x,y),f(y,x)\} \subseteq D_n := \{d_1,\dots,d_n\}$
we have $\alpha_1(f(x,y))=\alpha_2(f(y,x))$.  

The equivalence class of the partial map with the empty domain $D_0$ becomes the root of the tree, on level $n=0$.  
We define adjacency in the tree by restriction as follows: 
when $E$ is a vertex on level $n$, and $E'$ a vertex on level $n+1$,
and $E$ contains $(\alpha_1,\alpha_2)$ and $E'$ contains
$(\alpha'_1,\alpha'_2)$ such that $\alpha_1 = \alpha'_1\upharpoonright_{D_n}$
and $\alpha_2 = \alpha'_2\upharpoonright_{D_n}$, then we make $E$ and $E'$
adjacent in the tree. 
Note that 
the resulting rooted tree is finitely branching. 
By assumption, the tree has vertices on all levels. 
Hence, by K\"onig's tree lemma, there exists an infinite path $E_0,E_1,E_2,\dots$ in the tree, where $E_i$ is from level $i \in \mathbb N$. 

We define $e_1,e_2 \in \overline{\Aut(\Gamma)}$ as follows. 
Suppose $e_1,e_2$ are already defined on $D_n$ 
such that $\alpha_1 := e_1\upharpoonright_{D_n}$, $\alpha_2 := e_2\upharpoonright_{D_n}$, and 
$(\alpha_1,\alpha_2) \in E_n$. 
We want to define $e_1$ and $e_2$ on $d_{n+1}$, and we will do it in such a way that 
$(e_1\upharpoonright_{D_{n+1}},e_2\upharpoonright_{D_{n+1}}) \in E_{n+1}$. 
Since $E_n$ and $E_{n+1}$ are adjacent, there exist $(\beta_1,\beta_2) \in E_n$ and $(\beta'_1,\beta_2') \in E_{n+1}$ such that $\beta_1 = \beta_1' \upharpoonright_{D_n}$
and $\beta_2 = \beta_2'\upharpoonright_{D_{n}}$. 
By the definition of $\sim$ there
exists a $\delta \in \Aut(\Gamma)$ such that 
$\alpha_1  =  \delta \circ \beta_1$ and $\alpha_2 =  \delta \circ \beta_2$. 
For $j \in \{1,2\}$, define $\alpha'_j := \delta \circ \beta_j'$ 
so that $(\alpha'_1,\alpha'_2) \in E_{n+1}$
and observe that  

\begin{displaymath}
\alpha'_j \upharpoonright_{D_n} :=  \; (\delta \circ \beta_j') \upharpoonright_{D_n}  =  \delta \circ \beta_j = \alpha_j \; ,
\end{displaymath}
and hence that $\alpha'_j$ extends $\alpha_j$. 
Define $e_j(d_{n+1}) := \alpha'_j(d_{n+1})$. 
\end{proof}

\begin{corollary}\label{cor:lift}
There are endomorphisms $e_1,e_2$
of the structure $(\mL;C)$ such that $e_1(\tx(x,y)) = e_2(\tx(y,x))$. 
\end{corollary}
\begin{proof}
By Proposition~\ref{prop:txexistence}, for any finite $X \subset \mL$ there is an $\alpha \in
\Aut(\mL;C)$ such that $\tx(x,y) = \alpha(\tx(y,x))$ for all $x,y \in X$. Thus, Lemma~\ref{lem:lift}
applies to $f := \tx$ and $\Gamma := (\mL;C)$.
\end{proof}

\section{Main Results}
\label{sect:main}
In this section we complete (in Section~\ref{sect:complexity}) the proof of the complexity dichotomy for phylogeny problems that we announced in Theorem~\ref{thm:complexity-main}, via the reformulation
as CSPs for reducts of $(\mL;C)$ given in Theorem~\ref{thm:main-polym}. 
But our results 
are much stronger than the complexity classification
from Theorem~\ref{thm:main-polym}. 
We present a dichotomy for reducts of $(\mL;C)$
which remains interesting even if P=NP, and which 
we view as a fundamental result not just in the  
context of constraint satisfaction. 
Our dichotomy can be phrased in various but equivalent ways, using terminology from topology, universal algebra, or model theory. We introduce the necessary concepts in Section~\ref{sect:interpret}, \ref{sect:clone-homos}, and~\ref{sect:eqs}, and then state in Section~\ref{sect:results} how they are linked
together in the strongest formulation of our results.

\subsection{Primitive Positive Interpretations}
\label{sect:interpret}
Primitive positive interpretations are often used for proving NP-hardness results; we refer
the reader to Bodirsky~\cite{Bodirsky-HDR} for more information about this.
We will often consider the relation
$\NAE = \{0,1\}^3 \setminus \{(0,0,0),(1,1,1)\}$ in connection with
primitive positive interpretations. The problem
$\Csp(\{0,1\}; \NAE)$ is 
 called \emph{positive Not-All-Equal 3SAT} by Garey \& Johnson~\cite{GareyJohnson}
 and it is known to be NP-complete. 

\begin{definition}
A relational $\sigma$-structure $\Delta$ has a \emph{(first-order) interpretation $I$} in a $\tau$-structure $\Gamma$ if there exists a natural number $d$, called the \emph{dimension} of $I$, and
\begin{itemize}
\item a $\tau$-formula $\delta_I(x_1, \dots, x_d)$ -- called the \emph{domain formula},
\item for each atomic $\sigma$-formula $\phi(y_1,\dots,y_k)$ a $\tau$-formula $\phi_I(\overline x_1, \dots, \overline x_k)$ where the $\overline x_i$ denote disjoint $d$-tuples of distinct variables -- called the \emph{defining formulas},
\item a surjective map $h$ from all $d$-tuples of elements of $\Gamma$
that satisfy $\delta_I$ to $\Delta$ -- called the \emph{coordinate map},
\end{itemize}
such that for all atomic $\sigma$-formulas $\phi$ and all tuples in the domain of $h$ 
\begin{displaymath}
\Delta \models \phi(h(\overline a_1), \dots, h(\overline a_k)) \; 
 \Leftrightarrow \; 
\Gamma \models \phi_I(\overline a_1, \dots, \overline a_k) \; .
\end{displaymath}
\end{definition}
If the formulas $\delta_I$ and $\phi_I$ are all primitive positive, 
we say that the interpretation $I$ is \emph{primitive positive}.
We say that $\Delta$ is \emph{primitive positive interpretable (or pp interpretable)} with parameters
in $\Gamma$ if $\Delta$ has an interpretation $I$ where the formulas $\delta_I$ and $\phi_I$ may involve
elements from $\Gamma$ (the \emph{parameters}), 
that is, the interpretations in the expansion of $\Gamma$ by finitely many constants.
The importance of primitive positive interpretations
in the context of the CSP comes from the following lemma. 

\begin{lemma}[Proposition 3 in Bodirsky~\cite{BodirskySurvey}]
\label{lem:pp-interpret}
Let $\Gamma$ and $\Delta$ be structures 
with finite relational signature. 
Suppose that $\Gamma$ is
$\omega$-categorical 
and that $\Delta$ has a primitive positive interpretation
in $\Gamma$. Then there is a polynomial-time reduction from $\Csp(\Delta)$ to $\Csp(\Gamma)$.
If $\Gamma$ is a model-complete core, then
the interpretation might even be with parameters
and the conclusion of the lemma still holds. 
\end{lemma}

We present two primitive positive interdefinability results in this section.
The first one (Proposition~\ref{prop:Nd-interprets}) is concerned with the relation $N$
while the second one (Proposition~\ref{prop:Qd-interprets}) is concerned with the relation $Q$.

\begin{proposition}\label{prop:Nd-interprets}
Let $a,b$ be arbitrary distinct members of $\mL$.
Then the structure $(\{0,1\}; \NAE)$
is primitive positive interpretable in $(\mL;N,a,b)$. 
\end{proposition}
\begin{proof}
We freely use the relation $N_d$ in primitive positive
formulas, since $N_d \in \langle (\mL;N) \rangle$
by Lemma~\ref{lem:basicrelsdefs}.  
The dimension of the interpretation is one.
The domain formula is $N_d(a,x,b)$. 
The coordinate map $c$ sends $x$ to $0$ if $ax|b$,
and to $1$ if $a|xb$. The defining formula $\phi(x,y,z)$ for 
the ternary relation that we want to interpret
is 
\begin{displaymath}
\exists w_1,w_2 \, \big (N_d(x,w_1,y) \wedge N_d(w_1,w_2,z) \wedge N_d(w_1,a,w_2) \wedge N_d(w_1,b,w_2) \big) \; .
\end{displaymath}
We have to verify that $(c(x),c(y),c(z)) \in \NAE$ if and only if 
$\phi(x,y,z)$ holds in $(\L;N_d)$. 

First suppose that $c(x)=c(y)=c(z)=0$. Then $ax|b$, $ay|b$, and $az|b$. 
If $\phi(x,y,z)$ then $xw_1|y$ or $x|w_1y$, 
so in any solution we must have $axyzw_1|b$. 
Another consequence of $\phi(x,y,z)$ is that
$zw_2|w_1$ or $z|w_2w_1$. Hence, in any solution
we must have $axyzw_1w_2|b$.
Finally, $\phi(x,y,z)$ implies $bw_1|w_2$
or $bw_2|w_1$, in contradiction to $axyzw_1w_2|b$.
The situation that 
$\phi(x,y,z)$ and
$c(x)=c(y)=c(z)=1$ can be ruled out analogously, since the 
interpreting formula is symmetric in $a$ and $b$. 

Now suppose that $c(x)=c(y)=0$ and $c(z) = 1$.
In this case we can satisfy $\phi$ by assigning values
to $w_1,w_2$ such that $axyw_1|bzw_2$, $a|xyw_1$, $x|yw_1$,
and $b|zw_2$. 
Again, the case that $c(x)=c(y)=1$ and $c(z) = 0$
can be treated analogously. 

Next, consider that case that $c(x)=c(z)=0$ and $c(y)=1$. Then
we can satisfy $\phi$ by assigning values to $w_1,w_2$ such that 
$axzw_1 | byw_2$, $a|xzw_1$, $x|zw_1$, and $b|yw_2$. 
Note that the interpreting formula is also symmetric in $x$ and $y$. Hence, the case that $c(y)=c(z)=0$
and $c(x)=1$ can be treated analogously. 
Finally, the remaining two cases $c(x)=0, c(y)=c(z)=1$, and $c(y)=0, c(x)=c(z)=1$
are analogous to the previous two by the symmetry
of $a$ and $b$. 
\end{proof}


\begin{proposition}\label{prop:Qd-interprets}
The structure $\Delta := (\{0,1\}; \NAE)$
has a primitive positive interpretation in $(\mL;Q_d,a,b,c)$, where $a,b,c \in \mL$ are three pairwise distinct constants. 
\end{proposition}
\begin{proof}
Let $T_d'(x,y,z)$ be the relation defined by $Q_d(a,x,y,z)$. 
It follows from Lemma~14 in Bodirsky, Jonsson, and Pham~\cite{BodJonsPham} 
that $(\mL\backslash \{a\};T_d')$ is isomorphic to $(\mL;T_d)$. 
Let $h \colon (\mL\backslash \{a\}) \to \mL$ be such an isomorphism. 
Since $(\mL;T_d)$ is 2-transitive,
we can assume without loss of generality that $h$ fixes $b$ and $c$. 
It is straightforward to verify that there exists a 
one-dimensional 
primitive positive interpretation of $(\mL;T_d;Q_d,b,c)$ in $(\mL;Q_d,a,b,c)$: as coordinate we choose $h$, 
the interpreting formula for $T_d(x,y,z)$ is $Q_d(a,x,y,z)$,
the interpreting formula for $x=b$ is $x=b$, for $x=c$ is $x=c$, and
for $Q_d(x,y,z,t)$ is $(Q_d(a,z,x,y) \wedge Q_d(a,t,x,y)) \vee (Q_d(a,x,z,t) \wedge Q_d(a,y,z,t))$. 

Since the Boolean split relation of $Q_d$ is not affine, 
Lemma~\ref{lem:affine} implies
that $N_d$ is primitive positive definable in $(\mL;T_d,Q_d)$; recall that we throughout Section~\ref{sect:violating} 
tacitly assume that $N_d \not\in \langle \Gamma \rangle$ and $T_d \in \Gamma$.
Hence, by Proposition~\ref{prop:Nd-interprets}, 
the structure $\Delta$
has a primitive positive interpretation in $(\mL;N_d,b,c)$. It follows that $\Delta$ has a primitive positive interpretation also in $(\mL;Q_d,a,b,c)$.
\end{proof}


\subsection{Clone Homomorphisms}
\label{sect:clone-homos}
Let ${\cal C}$ and ${\cal D}$ denote two function clones as defined in Section~\ref{sect:algebraic}.
A function $\xi \colon {\cal C} \rightarrow {\cal D}$ is called a {\em clone homomorphism} if 
it sends every projection in ${\cal C}$ to the corresponding projection in ${\cal D}$, and
it satisfies the identity 
\[\xi(f(g_1,\dots,g_n))=\xi(f)(\xi(g_1),\dots,\xi(g_n))\]
for all $n$-ary $f \in {\cal C}$ and all $m$-ary $g_1,\dots,g_n \in {\cal C}$. Such a homomorphism
$\xi$ is \emph{continuous} if the map $\xi$
is continuous with respect to the topology of pointwise convergence, where the closed
sets are precisely the sets that are locally closed
as defined in Section~\ref{sect:algebraic}. 

The importance of continuous clone homomorphisms
in the context of primitive positive interpretations
comes from the following. 

\begin{theorem}[Theorem 1 in Bodirsky~\&~Pinsker~\cite{Topo-Birk}]\label{thm:topo-birk}
Let $\Gamma$ be $\omega$-categorical and
$\Delta$ be finite. Then $\Delta$ has a primitive
positive interpretation in $\Gamma$ if and only
if $\Pol(\Gamma)$ has a continuous clone homomorphism to $\Pol(\Delta)$. 
\end{theorem}

The most relevant situation are primitive positive interpretations of hard Boolean CSPs and in this case Theorem~\ref{thm:topo-birk} has an equivalent formulation that is given below. We write $\bf 1$ for the clone
on the set $\{0,1\}$ 
that only contains the projections. 

\begin{theorem}[Bodirsky \& Pinsker~\cite{Topo-Birk}]\label{thm:topo-birk-2}
Let $\Gamma$ be an $\omega$-categorical structure.
Then $(\{0,1\};\NAE)$ has a primitive positive
interpretation in $\Gamma$ if and only if
$\Pol(\Gamma)$ has a continuous clone homomorphism to $\bf 1$. 
\end{theorem}

\subsection{Taylor Operations modulo Endomorphisms}
\label{sect:eqs}
A polymorphism $f$ of $\Gamma$ of arity $n \geq 2$ is called
a \emph{Taylor polymorphism modulo endomorphisms of $\Gamma$}
if for every $i \leq n$ there are endomorphisms
$e_1,e_2$ of $\Gamma$ and $x_1,\dots,x_n,y_1,\dots,y_n \in \{x,y\}$
with $x_i \neq y_i$
such that the following holds.
\begin{displaymath}
 \forall x,y. \, e_1(f(x_1,\dots,x_n)) = e_2(f(y_1,\dots,y_n))\;.
\end{displaymath}
A special case of Taylor polymorphisms modulo
endomorphisms are \emph{symmetric polymorphisms modulo endomorphisms}, that is, the existence 
of an $f$ and endomorphisms $e_1$ and $e_2$
such that $\forall x,y. \,  e_1(f(x,y)) = e_2(f(y,x))$. 

In an $\omega$-categorical model-complete core the existence of Taylor polymorphisms modulo endomorphisms rules out the existence of an interpretation of
$(\{0, 1\}; \NAE)$ with parameters (as we explain below, this follows from Theorem 8.6 and the proof of Lemma 8.7 below). Recently,  Barto and Pinsker showed that
the existence of a Taylor polymorphism modulo endomorphisms is in fact equivalent to the non-existence of an interpretation of $({0, 1}; \NAE)$ with parameters~\cite{BartoPinsker16}. For the case of phylogeny problems, Theorem~\ref{thm:pre-main} below implies a stronger result; that is, the non-interpretability of $(\{0,1\};\NAE)$ with parameters in a model-complete core reduct of $(\mL;C)$ is equivalent to the existence of a symmetric polymorphism modulo endomorphisms.



The following lemma is stated for symmetric polymorphisms modulo endomorphisms, however the proof of the lemma can be adapted to the case of Taylor polymorphisms modulo endomorphisms. The statement is restricted for simplicity of
notation, but also because we only need it
for this special case in the statement of our main
result about reducts of $(\mL;C)$. 

\begin{lemma} \label{lem:clonehomo}
Let $\Gamma$ be an $\omega$-categorical model-complete core such that
$\Gamma$ has a symmetric polymorphism modulo endomorphisms. Then for any elements
$a_1,a_2,\dots,a_n$ of $\Gamma$ there is no clone homomorphism from
$\Pol(\Gamma,a_1,a_2,\dots,a_n)$ to $\bf{1}$.
\end{lemma}
\begin{proof}
Let $D$ denote the domain of $\Gamma$. By the assumption there exist a binary polymorphism $f$ of $\Gamma$ and $e_1,e_2\in \End(\Gamma)$ such that $e_1(f(x,y)) = e_2(f(y,x))$ for all
$x,y\in D$. Let
$\hat f \colon D \to D$ be given by
$\hat f(x):=f(x,x)$ for all $x\in D$.
Clearly, $\hat f$ is an endomorphism of $\Gamma$.
Let $a$ denote $(a_1,a_2,\dots,a_n)$. Then $a$, $\hat f(a)$, and $e_1(\hat f(a)) = e_2(\hat f(a))$ lie in the same orbit of $\Aut(\Gamma)$
because $\Gamma$ is a model-complete core.
Let $\alpha,\beta \in \Aut(\Gamma)$ be
such that $\alpha e_1(\hat f(a))=a$
and $\beta(\hat f(a)) = a$.
Let $h_1 := \alpha e_1 \beta^{-1}$ and
$h_2 := \alpha e_2 \beta^{-1}$,
and $g := \beta f$.
Clearly, we have $g(a,a)=\beta \hat f(a) = a$ by the choice of $\beta$. We will show that
$h_1(a)=a$ and $h_2(a)=a$. We have
\begin{displaymath}
h_1(a)=h_1(g(a,a))=\alpha e_1 \beta^{-1}(\beta f(a,a)) = \alpha e_1(\hat f(a)) = a \, .
\end{displaymath}
Similarly one can show that $h_2(a)=a$.
It follows that $h_1,h_2 \in \End(\Gamma,a_1,a_2,\dots,a_n)$
and that $g \in \Pol(\Gamma,a_1,a_2,\dots,a_n)$. Moreover, for all
$x,y \in D$ we have that
\begin{align*}
h_1(g(x,y))= & \; \alpha e_1 \beta^{-1} (\beta f(x,y)) \\
= & \; \alpha e_1(f(x,y)) \\
= & \; \alpha e_2(f(y,x)) \\
= & \; \alpha e_2 \beta^{-1}( \beta f(y,x)) \\
= & \; h_2(g(y,x)) \; .
\end{align*}
This shows that $(\Gamma,a_1,\dots,a_n)$ has a
symmetric polymorphism modulo endomorphisms. Thus there is no clone homomorphism from $\Pol(\Gamma,a_1,a_2,\dots,a_n)$ to ${\bf 1}$.
\end{proof}


\subsection{Algebraic-Topological Dichotomy}
\label{sect:results}
We have seen in Section~\ref{sect:classification}
that for the study of reducts $\Gamma$ 
of $(\mL;C)$, the situation
where the relation $C$ is primitive positive definable
in $\Gamma$ is the most important. We will next characterise those $\Gamma$ that have binary symmetric polymorphisms
modulo endomorphisms.
\begin{theorem}\label{thm:pre-main}
Let $\Gamma$ be reduct of $(\mL;C)$
such that $C \in \langle \Gamma \rangle$. 
Then the following are equivalent.
\begin{enumerate}
\item The relation $N$ does not have a primitive positive definition in $\Gamma$.
\item All relations with a primitive positive definition
in $\Gamma$ are free, separated, and induce an affine split relation.
\item All relations in $\langle\Gamma\rangle$ are affine Horn.
\item $\Gamma$ is preserved by the binary operation $\tx$.
\item $\Gamma$ has a binary polymorphism $f$ and
endomorphisms $e_1,e_2$ such that
\begin{displaymath}
e_1(f(x,y))=e_2(f(y,x)) \; .
 \end{displaymath}
\item For arbitrary $a_1,\dots,a_n \in \mL$,
there does not exist any clone homomorphism from $\Pol(\Gamma,a_1,\dots,a_n)$ to $\bf 1$.
\item For arbitrary $a_1,\dots,a_n \in \mL$, there does not exist any continuous 
clone homomorphism
from $\Pol(\Gamma,a_1,\dots,a_n)$ to $\bf 1$.
\item In any expansion of $\Gamma$ by finitely many constants there is no primitive positive interpretation of
$(\{0,1\}; \NAE)$. 
\end{enumerate}
\end{theorem}
\begin{proof}
We show the equivalences by proving implications
in cyclic order.

$1 \Rightarrow 2$. Combine Lemmas~\ref{lem:affine}, \ref{lem:separation}, and \ref{lem:freeness}.

$2 \Rightarrow 3$. Proposition~\ref{prop:gen-syntax}.

$3 \Rightarrow 4$. Proposition~\ref{prop:txpreservation} (or Corollary \ref{cor:horn-tx}).

$4 \Rightarrow 5$. Corollary~\ref{cor:lift}.

$5 \Rightarrow 6$. Lemma~\ref{lem:clonehomo}.

$6 \Rightarrow 7$. A fortiori.

$7 \Rightarrow 8$. Theorem~\ref{thm:topo-birk-2}. 

$8 \Rightarrow 1$. Follows from the contraposition of Proposition~\ref{prop:Nd-interprets}.
\end{proof}
 
 The fact that the requirement that the clone homomorphism to $\bf 1$ is \emph{continuous} in item~7 of
 Theorem~\ref{thm:pre-main} can simply
 be dropped in item~6 is remarkable, and
 it is not clear whether the continuity condition
 can be dropped for clone homomorphisms from polymorphism clones of general relational structures $\Gamma$ to $\bf 1$ (see the discussion in Bodirsky, Pinsker and Pongr\'acz~\cite{BPP-projective-homomorphisms}). 

In the general situation where the relation $C$ is not required to be a member of $\langle \Gamma \rangle$, we can still characterise those $\Gamma$ whose model-complete cores have binary symmetric polymorphisms modulo endomorphisms.

\begin{theorem}\label{thm:main}
Let $\Gamma$ be a reduct of $(\mL;C)$,
and let $\Delta$ be the model-complete core of $\Gamma$. Then 
the following are equivalent. 
\begin{enumerate}
\item $\Delta$ has a binary polymorphism $f$ and
an endomorphisms $e_1,e_2$ such that 
\begin{displaymath}
e_1(f(x,y))=e_2(f(y,x)) \; .
\end{displaymath}
\item For all elements $a_1,\dots,a_n$ of $\Delta$, 
there is no clone homomorphism from $\Pol(\Delta,a_1,\dots,a_n)$ to $\bf 1$.
\item For all elements $a_1,\dots,a_n$ of $\Delta$, there is no continuous 
clone homomorphism
from $\Pol(\Delta,a_1,\dots,a_n)$ to $\bf 1$.
\item In any expansion of $\Delta$ by finitely many constants there is no primitive positive interpretation of
$(\{0,1\}; \NAE)$. 
\end{enumerate}
\end{theorem}
\begin{proof}

The implication $(1) \Rightarrow (2)$ follows directly from Lemma \ref{lem:clonehomo} and the implication $(2) \Rightarrow (3)$ is trivial. The implication $(3)\Rightarrow (4)$ follows from Theorem \ref{thm:topo-birk-2}.   
For the implication
$(4) \Rightarrow (1)$, we use the classification of $\Delta$
into four types from Theorem~\ref{thm:endos}. 
For the first type, $\Delta$ has just one element and hence satisfies item 1. 
For the second type, the statement follows from results by Bodirsky and K\'ara~\cite{ecsps}; in fact, $\tx$ is a suitable polymorphism. 
For the third type, $Q_d \in \langle \Gamma \rangle$ by Lemma~\ref{lem:basicrelsdefs}
and one can show that $Q \in \langle \Gamma \rangle$, too.
Furthermore, NAE has a primitive positive definition in $(\mL;Q,a_1,a_2,a_3)$ for 
arbitrary pairwise distinct constants $a_1,a_2,a_3 \in \mL$ so
NAE has a primitive positive definition in $(\Gamma,a_1,a_2,a_3)$. 
By Theorem~\ref{thm:topo-birk-2}, there is a continuous clone
homomorphism from $\Pol(\Gamma,a_1,a_2,a_3)$ to ${\bf 1}$. We can disregard this case since it contradicts our
basic assumption.

We now focus on the fourth type.
It can be shown that in this case $C \in \langle \Gamma \rangle$ since
$C_d \in \langle \Gamma \rangle$ by Lemma~\ref{lem:basicrelsdefs}.
If $N \in \langle \Gamma \rangle$, then
NAE has a primitive positive definition in $(N,a_1,a_2)$ where $a_1,a_2 \in \mL$
are distinct constants. This contradicts $(4)$.
If $N \not\in \langle \Gamma \rangle$, then Corollary~\ref{cor:syntax} 
implies that
every relation in $\langle \Gamma \rangle$ is affine Horn. By 
Corollary~\ref{cor:horn-tx},
$\tx$ is a binary commutative polymorphism modulo endomorphisms.
\end{proof}


\subsection{Complexity Dichotomy}
\label{sect:complexity}
All the ingredients to prove the complexity classification stated in Theorem~\ref{thm:main-polym} are now available.
Recall that Theorem~\ref{thm:main-polym} states
that the CSPs in our class are in P if they have a binary polymorphism that is symmetric
modulo endomorphisms, and NP-complete otherwise.

\begin{proof}
Let $\Gamma$ be a reduct of $(\mL;C)$ with 
finite relational signature. 
Clearly, $\Csp(\Gamma)$ is in NP. 
Let $\Delta$ be the model-complete core of $\Gamma$.
If $\Delta$ has an expansion by finitely many constants 
that interprets $(\{0,1\};\NAE)$ primitively positively, then
$\Csp(\Delta)$ and therefore $\Csp(\Gamma)$ are NP-complete
by Lemma~\ref{lem:pp-interpret}. 
So let us assume that this is not the case.
Then by Theorem~\ref{thm:main}, the structure
$\Delta$ has a polymorphism $f$ and endomorphisms $e_1,e_2$
such that $e_1(f(x,y))=e_2(f(y,x))$. We consider the cases of $\Gamma$ in Theorem \ref{thm:endos}. If $\Gamma$ has a constant endomorphism,
then the model-complete core $\Delta$ of
$\Gamma$ has just one element,
and $\Csp(\Gamma)$ is trivial and in P.
If $\Delta$ is isomorphic to a reduct of $(\mL;=)$, then ${\rm CSP}(\Gamma)$ is in P by Theorem~\ref{thm:ecsps}. 
Otherwise, by Lemma~\ref{lem:preclass}, the structure
$\Gamma$ itself is a model-complete core, and
the relation $Q_d$ or the relation $C_d$
is primitive positive definable in $\Gamma$.
If $Q_d\in \langle\Gamma \rangle$, by Proposition \ref{prop:Qd-interprets} there is a pp interpretation of $(\{0,1\};{\rm NAE})$ in $\Gamma$. It follows from Theorem~\ref{thm:topo-birk-2} that there is a continuous clone homomorphism from $\Pol(\Gamma,a,b,c)$ to ${\bf 1}$. This is impossible by Lemma \ref{lem:clonehomo}. We are left with the case $C_d \in \langle \Gamma \rangle$ and therefore
$C \in \langle \Gamma \rangle$. Theorem~\ref{thm:pre-main} implies
that $\Gamma$ must have the polymorphism $\tx$, 
and that all relations of $\Gamma$ are affine Horn. 
In this case, $\Csp(\Gamma)$ is in P by Corollary~\ref{cor:alg}. 
\end{proof}

Suppose that $\Gamma$ is a reduct of $(\mL;C)$
with finite relational signature such that $C \in \langle \Gamma \rangle$. 
Then one might ask whether the \emph{meta-problem} of deciding the complexity of $\Csp(\Gamma)$
is effective. Here we assume that $\Gamma$ is given
via quantifier-free first-order definitions of its relations in $(\mL;C)$. 
We can then use the techniques developed by Bodirsky, Pinsker, and Tsankov~\cite{BPT-decidability-of-definability} to effectively test whether the relation $N$ is in $\langle \Gamma \rangle$. Thus,
the meta-problem for phylogeny problems is decidable. 

Another way of proving decidability of the meta-problem is to verify the condition
that every relation in $\langle \Gamma \rangle$ is separated, 
free, and has an affine split relation. In fact, it is sufficient to test the condition for all relations
that can be defined by existentially quantifying some of the variables in expressions of the form
$$S(x_1,\dots,x_k) \wedge \bigwedge_{\{i,j\} \in I} x_i=x_j,$$
where $S \in \Gamma$ and $I \subseteq \{1,\dots,k\}^2$.
This follows from our proof of Theorem~\ref{prop:gen-syntax}: 
when proving that a relation $R$ can be defined by an affine Horn formula,
we only consider the three properties above for relations having
this particular form. Since there are finitely many of them, decidability of the meta-problem follows.

\bibliographystyle{plain}
\bibliography{local}

\end{document}